\documentclass[%
 jmp,
cp,  
 amsmath,amsthm,amssymb,
 reprint,%
onecolumn]{revtex4-2}

\usepackage{graphicx,overpic}
\usepackage{dcolumn}
\usepackage{bm}

\usepackage[utf8]{inputenc}
\usepackage[T1]{fontenc}
\usepackage{mathptmx} 

\usepackage{amsthm}

\usepackage[dvipsnames]{xcolor}

\newtheorem{theorem}{Theorem}[section]
\newtheorem{corollary}[theorem]{Corollary}
\newtheorem{lemma}[theorem]{Lemma}
\newtheorem{proposition}[theorem]{Proposition}
\newtheorem{remark}[theorem]{Remark}
\newtheorem{definition}[theorem]{Definition}

\renewcommand{\Re}{\operatorname{Re}}
\renewcommand{\Im}{\operatorname{Im}}
\begin{document}

\title{The Elliptic Ginibre Ensemble: A Unifying Approach to Local and Global Statistics\\ for Higher Dimensions}

\author{G. Akemann} 
 \email[Corresponding author: ]{akemann@physik.uni-bielefeld.de}
 \affiliation{
 Faculty of Physics, Bielefeld University, P.O. Box 100131, 
D-33501 Bielefeld, Germany 
}
\author{M. Duits}%
 \email{duits@kth.se}
\affiliation{
Department of Mathematics, Royal Institute of Technology (KTH), SE10044 Stockholm, Sweden.
}

\author{L. D. Molag}
 \email{lmolag@math.uni-bielefeld.de}
\affiliation{Faculty of Mathematics and Faculty of Physics, Bielefeld University, P.O. Box 100131, 
D-33501 Bielefeld, Germany 
%
}%

\date{\today} 

\begin{abstract}
The elliptic Ginibre ensemble of complex non-Hermitian random matrices allows to interpolate between the rotationally invariant Ginibre ensemble and the Gaussian unitary ensemble of Hermitian random matrices. It corresponds to a two-dimensional one-component Coulomb gas in a quadrupolar field, at inverse temperature $\beta=2$. Furthermore, it represents a determinantal point process in the complex plane with corresponding kernel of planar Hermite polynomials. Our main tool is a saddle point analysis of a single contour integral representation of this kernel. 
We provide a unifying approach to rigorously derive several known and new results of local and global spectral statistics, including in higher dimensions. 
First, we prove the global statistics in the elliptic Ginibre ensemble first derived by Forrester and Jancovici. The limiting kernel receives its main contribution from the boundary of the limiting elliptic droplet of support. In the Hermitian limit, there is a known correspondence between non-interacting fermions in a trap in $d$ real dimensions $\mathbb{R}^d$ and the $d$-dimensional harmonic oscillator. We present a rigorous proof for the local $d$-dimensional bulk (sine-) and edge (Airy-) kernel first defined by Dean et al., complementing recent results by Deleporte and Lambert. 
Using the same relation to the $d$-dimensional harmonic oscillator in  $d$ complex dimensions $\mathbb{C}^d$, we provide new local bulk and edge statistics at weak and strong non-Hermiticity, where the former interpolates between correlations in $d$ real and $d$ complex dimensions. For $\mathbb{C}^d$ with $d=1$ this corresponds to non-interacting fermions in a rotating trap.
\end{abstract}

\maketitle

\section{Introduction} \label{sec:1}

The mathematics of random matrices was developed by Dyson, Mehta, Wigner and others in the 50s and early 60s in the context of applications to nuclear physics. The topic continues to be very popular among mathematicians and physicists today. In mathematics, perhaps one of the reasons is that many different areas contribute from different angles, ranging from probability theory over combinatorics to analysis to name a few. We refer to \cite{LNV,AGZ} for introductions to various methods and classical results. 
In physics even today new applications continue to appear, and we will pick up upon a recent example on non-interacting fermions in a $d$ dimensional trap, compare\cite{DeDoMaSc3}.  A non exhaustive list of more applications can be found in \cite{ABD} .

In random matrix theory, a first general distinction is between Hermitian and non-Hermitian ensembles, with real respectively  complex eigenvalues. In the asymptotic expansion at large matrix dimension $N$ we distinguish local and global scales
that depend on the location in the spectrum. Typically, different techniques are applied to different regimes. For example, on global scales resolvent methods are a popular tool, where for instance free probability or Schwinger-Dyson also called loop equations apply. 
On a local scale examples are given by orthogonal polynomials and their asymptotic analysis, see e.g. \cite{AGZ} for all three approaches. 
Moreover, when switching from Hermitian to non-Hermitian ensembles, often these tools break down or have to be substantially modified.

It is one of the goals of this article to present a unifying approach for a specific example, where different regimes all follow essentially from a single tool, both for real and complex eigenvalues.  The model we will consider consists of Gaussian random matrices with complex normal elements, the complex elliptic Ginibre ensemble \cite{Girko,SCSS}, extending the classical ensembles of Ginibre \cite{Ginibre}. It may be possible to extend this approach to the elliptic symplectic and real Ginibre ensembles as well.
The complex elliptic Ginibre ensemble corresponds to a determinantal point process given by the 
 kernel of planar Hermite polynomials, to be defined below. At the same time it has its own, genuine statistical mechanics interpretation as a Coulomb gas in a quadrupolar field, and we refer to \cite{Peter} to a detailed discussion. 
We will use a single contour integral representation of the correlation kernel, to rigorously derive known results and uncover new results, including higher order corrections. Due to the intimate relation between the harmonic oscillator in $d$ dimensions and Hermite functions, we can extend our asymptotic analysis straightforwardly,  without encountering $d$ dimensional saddle points.
Furthermore, we can extend the regime of weak non-Hermiticity introduced by Fyodorov, Khoruzhenko and Sommers in $d=1$ \cite{FyKhSo1,FyKhSo2} to $d$ complex dimensions, that interpolates between real and complex eigenvalue statistics. 

In $d$ real dimensions we have seen many new results in the application to non-interacting fermions in a trap in the past 5 years. They can be described in terms of the $d$-dimensional harmonic oscillator Hamiltonian, with the eigenfunctions given in terms of Hermite functions, and we refer to a recent review by Dean, Le Doussal, Majumdar and Schehr \cite{DeDoMaSc3}. Moreover, also in $d = 1$ complex dimension such a map from a fermionic system to the complex Ginibre ensemble exists, see \cite{BSG}. It can be either viewed as fermions in $3$ dimensions confined to $2$ dimensions in a rotating harmonic trap, or equivalently as electrons in the plane subject to a perpendicular magnetic field, the so-called Landau-levels.  These results have already attracted mathematicians to very recently prove \cite{DeLa,Charlier} and extend existing results in the rapidly developing physics literature, compare \cite{DeDoMaSc3}. The elliptic Ginibre ensemble also has a physical interpretation as electrons in a quadrupolar magnetic field \cite{FoJa}. The limit of weak non-Hermiticity that we will investigate allows us to interpolate between the Hermite eigenfunctions in $d=1$ real dimension and the monomials of the Ginibre ensemble in $d=1$ complex dimension.

In the remainder of this section we will introduce our model, describe the relation to the $d$ dimensional harmonic oscillator, and state our results. The core of our steepest descent analysis is presented in Section \ref{sec:steepest}. In the subsequent sections our theorems are proven based on this analysis. Section \ref{sec:EGE} presents our results on the elliptic Ginibre ensemble, where among others we prove a result of Forrester and Jancovici \cite{FoJa} for the two-point cluster function, see also \cite{AmCr} for recent related results. In Section \ref{sec:fermions} we present the proofs for the local bulk and edge correlations for non-interacting fermions in $d$ real dimensions, the $d$ dimensional sine- and Airy-kernel, that were given in \cite{DeDoMaSc} for $d\geq2$, see also references in \cite{DeDoMaSc3}. Our proof complements a recent work \cite{DeLa} for non-Gaussian potentials. Finally in Section \ref{sec:general} we present new results in $d$ complex dimensions that display an interesting factorisation property. This section also includes the above mentioned weak non-Hermiticity regime, that interpolates to the results from the previous section.

\subsection{Introduction of the models}
A determinantal point process on $\mathbb C^d$ is a point process on $\mathbb C^d$  for which the correlation functions have the form 
$$
    \frac{1}{k!} \det \left(K(Z^{(j)},{Z^{(\ell)})} \right)_{j,\ell=1}^{k} dZ^{(1)} \cdots d Z^{(k)},
$$
where each of $d Z^{(1)}, \ldots, d Z^{(k)}$ are the $2d$ dimensional Lebesgue measure on $\mathbb C^d$ and $K: \mathbb C^d \times \mathbb C^d \to \mathbb C$ is a function called the correlation kernel. The correlation kernel determines the point process.  
In this paper we study a sequence of determinantal point processes corresponding to a family of kernels $\mathcal K_n: \mathbb C^d \times \mathbb C^d \to \mathbb C$, with $n \in \mathbb N$, defined by  
\begin{align} \label{eq:general_kernel_pre}
\mathcal K_n(Z,Z') = \frac{1}{\pi^d (1-\tau^2)^\frac{d}{2}} \sum_{\lvert j\rvert < n} \frac{\left(\frac{\tau}{2}\right)^{\lvert j\rvert}}{j_1!\cdots j_d!} \prod_{k=1}^d \sqrt{\omega(Z_{k}) \omega(Z'_{k})} H_{j_k} \left({\frac{Z_{k}}{\sqrt{2 \tau}}} \right)\overline{ H_{j_k}\left({\frac{Z'_{k}}{\sqrt{2 \tau}}} \right)},
\end{align}
where the summation is over all multi-indices $j=(j_1,\ldots,j_d)$, with $j_1, \ldots, j_d\geq 0$, such that $\lvert j \rvert = j_1+\ldots+j_d<n$,
\begin{equation}
    \omega(Z) = \exp\left(-\frac{1}{1-\tau^2} \left(|Z|^2-\frac{\tau}{2} (Z^2+\overline Z ^2)\right)\right)= \exp\left(-\frac{(\Re Z)^2}{1+\tau}-\frac{(\Im Z)^2}{1-\tau} \right),
\end{equation}
and $H_j$ is the Hermite polynomial of degree $j$, normalized such that 
\begin{equation} \label{eq:defHermitePolyNormalization}
    \frac{1}{\pi \sqrt{1-\tau^2}}\frac{\left(\frac{\tau}{2}\right)^j}{j!} \int_{\mathbb C} H_j \left({\frac{Z}{\sqrt{2 \tau}}}  \right) \overline{H_k\left({\frac{Z}{\sqrt{2 \tau}}} \right)}  \omega(Z) dZ = \delta_{jk}.
\end{equation}
From the orthogonality condition 
proven in \cite{EvM,PdF}
it is clear that the integral operator corresponding to the kernel $\mathcal K_n$ is a self-adjoint projection operator and thus this  kernel indeed (see for example \cite{Sosh}) defines a determinantal point process. Moreover, this process is an example of a so-called biorthogonal ensemble \cite{Bor}. The kernel in \eqref{eq:general_kernel_pre} is an example of a Bergman kernel on $\mathbb C^d$,  which were introduced by Berman \cite{Berman1}.  

The motivation for studying this process is that it is an umbrella for several specializations to models that have been studied independently in the literature. Indeed, for $0< \tau < 1$ and $d=1$ this is the 
complex
Elliptic Ginibre ensemble (EGE) 
\cite{Girko,SCSS}
and the point process can alternatively be introduced as the eigenvalues of random matrices  of the form $M = \sqrt{1+\tau} H_1 + i \sqrt{1-\tau} H_2$, where $H_1, H_2$ are $n\times n$ independent Hermitian matrices chosen from the Gaussian Unitary Ensemble (GUE) (i.e., with  distribution $\sim \exp -\text{Tr}(H_j^2),\ j=1,2$). In other words, we take $M$ randomly from the probability measure on the space of complex-valued $n \times n$ matrices that is proportional to 
\begin{align}
\sim
\exp\left(- \frac{1}{1-\tau^2} \text{Tr}{\left(M^* M - \frac{\tau}{2}  \left(M^2 +(M^*)^2\right)\right)}\right) dM.
\end{align}
Here $dM$ is the product  of the Lebesgue measure on each of the (complex-valued) 
matrix 
entries. 
The EGE has been studied intensively in the literature as it interpolates naturally between the classical finite Ginibre Ensemble $\tau=0$ \cite{Ginibre} and the GUE in the Hermitian limit $\tau \uparrow 1$. 
On a global scale the limiting support of the eigenvalues follows the elliptic law \cite{Girko,SCSS}, interpolating between the circular and semi-circle law.
On the level of local correlation functions the crossover from the Ginibre ensemble to GUE was first discussed 
in \cite{FyKhSo1,FyKhSo2}. 
It was called weak non-Hermiticity limit, see also   \cite{FyKhSo3} for a more detailed discussion, including the relation with planar  Hermite polynomials. An interesting feature of this crossover is this weak non-Hermiticity regime in which one lets $\tau \uparrow 0$ simultaneously as $n \to \infty$. In that  weak non-Hermiticity regime, the microscopic process in the bulk is given by the deformed sine kernel \cite{ACV}, and at the edge one gets a deformation of the Airy kernel \cite{Be}. See also \cite{AP} for an overview and  discussion of such deformed kernels, resulting from other ensembles. We will provide a new proof for these deformed kernels and compute similar deformations in higher dimensions. 

The Gaussian Unitary Ensemble can be  associated to the one dimensional harmonic oscillator (see for example \cite{Sosh}). Similarly, for $d>1$ the Hermitian limit $\tau \uparrow 1$ of \eqref{eq:general_kernel_pre} corresponds to the fermionic process associated to the  $d$-dimensional quantum harmonic oscillator:
$$
\sum_{j=1}^d \left(-\frac{d^2}{dx_j^2}+x^2_j\right).
$$
The eigenfunctions for a positive self-adjoint operator have the form
\begin{align}
    \Psi_{j_1,j_2,\ldots,j_d}(x) = \psi_{j_1}(x_1) \psi_{j_2}(x_2) \cdots \psi_{j_d}(x_d), \qquad  x\in \mathbb R^d, \qquad j_1,\ldots,j_d=0,1,2,\ldots
    \end{align}
    where
    \begin{align}
    \psi_j(x) = \frac{1}{\sqrt{2^j j! \sqrt\pi}} H_j(x) e^{-\frac{1}{2} x^2},  \qquad x\in\mathbb R,
    \end{align}
    and $H_j$ is the $j$-th Hermite polynomial. These eigenfunctions have the eigenvalue $j_1+j_2+\ldots+j_d+\frac{d}{2}$. Note that for $d>1$ the eigenvalues of this operator are degenerate. By forming and squaring the Slater determinant corresponding to all eigenfunctions with the $n$ smallest eigenvalues, we obtain  a probability measure on $\mathbb R^d$ that induces a determinantal point process with kernel 
    \begin{equation}\label{eq:kernelferm}
       \mathcal K_n^{\rm Fermi}(x,x')=
    \sum_{0\leq j_1+\ldots+ j_d \leq n-1} \Psi_{j_1,j_2,\ldots,j_d}(x) \Psi_{j_1,j_2,\ldots,j_d}(x').
    \end{equation}
    This is the Hermitian limit $\tau \uparrow 1$ of the determinantal point process above in the following sense 
    \begin{equation}\label{eq:deltalim}
     \lim_{\tau \uparrow 1} \mathcal K_n(x+iy,x'+i y') e^{+\frac{y'^2-y^2}{2(1-\tau)}} =\delta (y) \mathcal K_n^{\rm Fermi}(x,x').
     \end{equation}
Here $\delta(y)$ is the Dirac delta function 
in $d$ dimensions 
and this factor implies that, in the limit $\tau \uparrow 1$, all points will be 
in $\mathbb{R}^d$.
We have used that the kernel can be multiplied by a cocycle, without changing the correlation functions of the point process. This asymmetric choice results into one delta function per argument in front of the determinant in the correlation functions.

The fermionic process with kernel $\mathcal K_n^{\rm Fermi}$ has been studied by several authors. In the physics literature it  has been discussed in  for example in
\cite{KM,DeDoMaSc} where they computed the limiting processes at microscopic scales in the bulk and at the edge, see \cite{DeDoMaSc3} for a review and references.
Very recently, a rigorous derivation for these limiting processes has been found in \cite{DeLa} for more general potentials. We will present an alternative rigorous proof for the special case of the harmonic oscillator. 

The determinantal point process with general parameters $0<\tau < 1$ and $d>1$ is thus an overarching process for interesting and well-studied special cases. One of the main points of this writing is to discuss a unifying approach for the asymptotic study of this model. This approach will allow us to present alternative derivations of known results, as well as deriving new results for these special ensembles.  In particular, we give a rigorous computation of the 
global, long-range correlations for two points near the edge  of the elliptic Ginibre ensemble, a new derivation of the local scaling limits of the fermionic process on $\mathbb{R}^d$ and  local scaling limits in $d$ complex dimensions in the weak non-Hermiticity regime extending this.

\subsection{Statement of main results}

Our main results are on the asymptotic behavior of the determinantal point process with kernel \eqref{eq:general_kernel_pre} on $\mathbb C^d$. In this paragraph we will single out the most important results. The proofs will be given in later sections.

 \subsubsection{The Elliptic Ginibre Ensemble}
First of all, our approach can be used to prove well-known results for the  the elliptic Ginibre ensemble, which corresponds to \eqref{eq:general_kernel_pre} with special choices $d=1$ and $0<\tau <1$. 

Let us start by  scaling   the  kernel $\mathcal K_n$ as follows
\begin{align} \label{eq:general_kernel_d=1}
    \mathbb K_n(Z,Z')= n \mathcal K_n\left(\sqrt n\, Z, \sqrt n\, Z' \right).
\end{align}
With this rescaling  the mean density of points accumulates on the elliptic domain
$$
    \mathcal E_\tau=
        \left\{
            Z \in \mathbb C \ : \ \left(\frac{\Re Z}{1+\tau}\right)^2+\left(\frac{\Im Z}{1-\tau}\right)^2<1
        \right\},
$$
with mean limiting density converging to the uniform distribution on that domain, i.e., 
$$\lim_{n \to \infty} \frac{1}{n} \mathbb K_n(z,z)=\frac{\mathfrak{1}_{z \in \mathcal E_\tau } }{\pi (1-\tau^2)}.
$$
There are many references for this result in $d=1$, starting from \cite{Girko,SCSS} up to \cite{LeRi} for most recent results, including convergence rates for the density and kernel. It it will also be a special case of a more general result for $d\geq 1$ that we will prove in this paper (see Theorem \ref{thm:limitingDensity}).

Apart from the limiting density, also the fluctuations have been well studied. For instance, it is  well known that in the bulk of this domain the correlation kernel converges locally to that of the (infinite) Ginibre process. The fluctuations on the global scale witness an interesting behavior: the correlation kernel for points at macroscopic distance in the bulk decays exponentially (with $n$). For distinct points $Z, Z'$ on the boundary of the ellipse, the correlation kernel is of order $\sqrt n$. For the elliptic Ginibre ensemble this was first computed by Forrester and Jankovici \cite{FoJa}, extending the results for the Ginibre ensemble at $\tau=0$ \cite{Choquard}. Below we will give an rigorous derivation of their result. 
But before we come to that, let us start by proving a  single expression for the asymptotic behaviour of the kernel $\mathbb K_n$ that captures both the local and the global correlations simultaneously.   To state this result, we  first need some further notation. It turns out to be convenient to express our results in terms of elliptic coordinates 
$$ 
  Z = 2 \sqrt{\tau} \cosh (\xi+i \eta),
$$
where $\xi \geq 0$ and $\eta \in (-\pi,\pi]$ if $\xi>0$, and $\eta \in [0,\pi]$ if $\xi=0$. Using these elliptic coordinates the ellipse $\mathcal E_\tau$ can be characterized as $\xi \leq \xi_\tau=-\frac12 \log \tau$. We choose to present the following theorem explicitly in terms of $Z, Z'$ and the elliptic coordinates, but we mention that a less explicit formulation of a more universal character is discussed in Remark \ref{remark:AmeurCronvall} below. 
  
\begin{theorem} \label{thm:largenKnZW_intro}
    Consider the elliptic Ginibre Ensemble at $d=1$ and $0<\tau<1$.  For $\xi,\xi'>0$ and $\eta, \eta' \in (-\pi,\pi]$, let $$Z=2 \sqrt{\tau} \cosh (\xi+i \eta),\qquad \textrm{ and }  \qquad  Z'=2 \sqrt{\tau} \cosh (\xi'+i \eta').$$ Furthermore, set $\xi_+=\frac{1}{2}(\xi+\xi')$ and $\xi_\tau=-\frac12 \log \tau$ and assume $(\xi_+ -\xi_\tau)^2+(\eta-\eta')^2>0$. Then, as $n\to\infty$ we have
    \begin{align} \nonumber
    \mathbb K_n(Z,Z') 
    = & 
    \frac{n \mathfrak{1}_{\xi_+<\xi_\tau}}{\pi (1-\tau^2)}
   \exp\left(-\frac{n}{1-\tau^2}\frac{|Z|^2+|Z'|^2-2 Z \overline{Z'}}{2}\right)C_{\tau}^n(Z,Z')
    \\ \nonumber
    &\pm \sqrt{\frac{n}{32 \pi^3 \tau(1-\tau^2)}} 
    \frac{e^{- n (\xi-\xi_\tau)^2 g(\xi+i\eta)} e^{- n (\xi'-\xi_\tau)^2 g(\xi'+i\eta')} 
    }{\sinh\left(\xi_+-\xi_\tau + i \frac{\eta-\eta'}{2}\right)\sqrt{\sinh(\xi+i\eta) \sinh(\xi'-i\eta')}}D_{\tau}^n(Z,Z')\\ \label{eq:behavKnZWgeneral_intro}
    &\hspace{7cm} + \mathcal O\left(\frac{1}{\sqrt n} e^{- n (\xi-\xi_\tau)^2 g(\xi+i\eta)} e^{- n (\xi'-\xi_\tau)^2 g(\xi'+i\eta')}\right),
    \end{align} 
    where the $\pm$ sign can be expressed explicitly in terms of $(Z,Z')$, and the function $g: \mathbb C \to (0,\infty)$ is an explicit continuous function 
    given in \eqref{eq:defpthetas}, 
    that is bounded from below by a positive constant.
The factors $C_{\tau}^n(Z,Z')=\exp\left(- \frac{i n \tau }{ 1-\tau^2} \frac{\operatorname{Im}(Z^2-Z'^2)}{2}\right)$    and $D_{\tau}^n(Z,Z')= \exp\left(i n (\eta-\eta' - \frac{e^{-2\xi}}{2} \sin(2\eta)+\frac{e^{-2\xi'}}{2} \sin(2\eta'))\right)$
    denote cocycles at the respective orders.
    The $\mathcal O$ term is uniform on compact sets of $(Z,Z')$ satisfying $\xi, \xi'>0$ and $(\xi_+ -\xi_\tau)^2+(\eta-\eta')^2>0$.
    \end{theorem}
    
    The first term on the right-hand side of \eqref{eq:behavKnZWgeneral_intro} is, up to a rescaling and a cocycle, the infinite Ginibre kernel. This  is the term that gives us the density of particles and the local correlations in the bulk. The second term on the right-hand side will have exponential decay  as long as  $\xi\neq \xi_\tau$ or $\xi'\neq \xi_\tau$.  However, for  $Z,Z'$  two  points at the edge of the ellipse we will have $\xi'=\xi=\xi_\tau$, and if these points are different then we also have $\eta\neq \eta'$. For such points we see that the second term on the right-hand side of \eqref{eq:behavKnZWgeneral_intro} is no longer exponentially decaying, but is of order $\sqrt n$ and  becomes the dominant term. 
    
    \begin{remark}
    A version of Theorem \ref{thm:largenKnZW_intro} is valid under the less rectrictive condition that $\xi, \xi'$ and $n^{1-2\nu}((\xi_+-\xi_\tau)^2+\sin^2 \frac{\eta-\eta'}{2})$ are bounded from below by a positive constant, for some fixed $0<\nu<\frac{1}{6}$. Then 
     \eqref{eq:behavKnZWgeneral_intro} holds when we replace $\frac{1}{\sqrt n}$ by $n^{3\nu}$ in the $\mathcal O$-term.
     One derives this by rescaling the integration variables by a factor $n^{-\frac{1}{2}+\nu}$ locally around the saddle point in the steepest descent analysis. We omit the details. 
    \end{remark}

    One direct consequence  of Theorem \ref{thm:largenKnZW_intro} concerns the two-point cluster function, which is defined as
\begin{align}
T_2(Z,Z') = - \lvert\mathbb K_n(Z,Z')\rvert^2. 
\end{align}
We prove the following result for $T_2$:
\begin{theorem} \label{thm:2ptCluster_intro}
Consider the elliptic Ginibre Ensemble $d=1$ and $0<\tau<1$. 
For every $Z_0, Z'_0\in \partial \mathcal E_\tau$ and $Z_0\neq Z'_0$ there exists an open neighbourhood $U$ of $(Z_0, Z'_0)$ such that
\begin{align} \label{eq:cluster2_intro}
T_2(Z,Z')
= -\frac{n}{16 \tau \pi^3 (1-\tau^2)} \frac{e^{- 2 n (\xi-\xi_\tau)^2 g(\xi+i\eta)} e^{- 2 n (\xi'-\xi_\tau)^2 g(\xi'+i\eta')}}{\cosh (2(\xi_+-\xi_\tau)) - \cos(\eta-\eta')} \frac{1}{|\sinh(\xi+i\eta) \sinh(\xi'+i\eta')|}  \left(1+ \mathcal O\left(\frac{1}{n}\right)\right), 
\end{align}  
uniformly for $(Z,Z') \in U$ as $n\to\infty$. 
In particular, as $n \to \infty$
\begin{align} \label{eq:cluster1_intro}
T_2(Z_0,Z'_0)
=& -\frac{n}{16 \tau \pi^3 (1-\tau^2)} \frac{1}{1-\cos(\eta_0-\eta'_0)} \frac{1}{|\sinh(\xi_\tau+i\eta_0) \sinh(\xi_\tau-i\eta'_0)|}  \left(1+ \mathcal O\left(\frac{1}{n}\right)\right).
\end{align} 

\end{theorem}
For the complex Ginibre ensemble at $\tau=0$, Theorem \ref{thm:2ptCluster_intro} was first investigated by 
Choquard, Piller and Rentsch \cite{Choquard}. It was extended to the EGE by Forrester and Jancovici \cite{FoJa}. In their final result they  only give the numerator in \eqref{eq:cluster2_intro} on the ellipse with $\xi=\xi'=\xi_\tau$, but \eqref{eq:cluster2_intro} can be reconstructed from their derivation.
The peaked behavior of the correlation kernel near the edge is a phenomenon that one expects to be universal and hold for general 2D Coulomb gases.  A general result in this direction can be found in a recent work of Ameur-Cronvall \cite{AmCr} (see Remark \ref{remark:AmeurCronvall} below). 

Another manifestation of this phenomenon can be seen in fluctuations of linear statistics for $2D$ Coulomb gases and the corresponding Gaussian log-correlated fields. Indeed, the boundary gives a non-trivial contribution, see for example \cite{LeSy, AmHeMa,RiVi}. This is in sharp contrast to the global fluctuations for tiling models, that are universally governed by the Gaussian Free Field with  Dirichlet boundary conditions. Since the EGE  interpolates between a 2D and 1D Coulomb gas, it is natural to ask how this transition takes place.   In the Hermitian limit, it is the behavior on the boundary that survives and dictates the global fluctuations. In a forthcoming work we will  return to this issue, and describe the transition from the one dimensional to the two-dimensional log-correlated fields in detail. 

\begin{remark}
    The first term on the right-hand side of  \eqref{eq:behavKnZWgeneral_intro} contributes whenever $\xi_+<\xi_\tau$ and this may happen even if one of the points is outside the elliptic domain $\mathcal E_\tau$. In fact, for $(Z,Z') = 2\sqrt\tau (\cosh(\xi+i\eta),\cosh(\xi'+i\eta))$, where $\xi<\xi_\tau<\xi'$ such that $\xi_+<\xi_\tau$ (and $\eta$ arbitrary), a careful analysis of our arguments below will show that the first term is dominant over the second term on the right-hand side of \eqref{eq:behavKnZWgeneral_intro}. As the first term is representing the infinite Ginibre Ensemble describing the local correlations, we find it remarkable that it is still dominant for these points that are at macroscopic distance. 
 \end{remark}

 \begin{remark} \label{remark:AmeurCronvall}
    Our results should be compared to a very recent paper by Ameur and Cronvall \cite{AmCr}(Theorem 1.3 in particular). They consider a normal matrix model with general potential $Q$. Then the corresponding mean limiting density has a compact support, which is called `the droplet'. With $S_Q$ denoting the (closure of the) droplet, they consider the unique conformal map $\phi = \phi_Q : \mathbb C\setminus S_Q\to \{Z\in \mathbb C : |Z| > 1\}$ with $\phi(\infty)=\infty$ and $\phi'(\infty)>0$. Such a conformal map can always be extended to an open set that contains the boundary of $S_Q$. In a $\delta_n$-neighborhood of $\mathbb C\setminus S_Q$, with $\delta_n$ of order $\sqrt{\frac{\log \log n}{n}}$, and under the condition that $|\phi(Z) \overline{\phi(Z')}-1|\geq \mu$ for some constant $\mu>0$, they find that a corresponding correlation kernel satisfies
\begin{align} \label{eq:AmeurCronvall1}
\mathbb K_n(Z,Z') = \sqrt{\frac{n}{2\pi}} e^{n G_Q(Z,Z')} 
\left(\phi(Z) \overline{\phi(Z')}\right)^n
\frac{\sqrt{\phi'(Z) \overline{\phi'(Z')}}}{\phi(Z) \overline{\phi(Z')} -1}
\left(1+\mathcal O(n^{-\beta})\right),
\end{align}
uniformly as $n\to\infty$, where $\beta$ is any number in $(0,\frac{1}{4})$, and $G_Q$ is a function that can be explicitly determined from $Q$ (see the paper for details). For the elliptic Ginibre ensemble the droplet is $\mathcal E_\tau$, and it is not hard to see that
\begin{align} \label{eq:AmeurCronvall2}
\phi(Z) = \frac{Z+\sqrt{Z^2-4\tau}}{2}.
\end{align}
Hence, in our case, the conformal map $\phi$ extends much further than the boundary of the droplet, we only have to exclude the motherbody $[-2\sqrt\tau, 2\sqrt\tau]$, i.e., the interval where the zeros of the (rescaled) Hermite polynomials accumulate. We mention that $\phi(Z) \overline{\phi(Z')}=\tau e^{\xi+\xi'+i(\eta-\eta')}$ 
by Corollary \ref{cor:saddleinZW}. Their result should hold in particular for $Z$ and $Z'$ on or microscopically close to the ellipse boundary. Indeed, substituting \eqref{eq:AmeurCronvall2} in \eqref{eq:AmeurCronvall1}, and determining $G_Q$, their result agrees with \eqref{eq:cluster1_intro} when expressed in elliptic coordinates. 
Moreover, as Theorem \ref{thm:largenKnZW_intro} shows, their result \eqref{eq:AmeurCronvall1} actually holds on a $\delta$-neighborhood, i.e., $\delta_n$ does not have to approach $0$. Note that, with these notations, Theorem \ref{thm:largenKnZW_intro} can be expressed as
\begin{multline*}
\mathbb K_n(Z,Z') = 
 n\frac{\mathfrak{1}_{|\phi(Z) \overline{\phi(Z')}|<1}}{\pi (1-\tau^2)}
   \exp\left(-\frac{n}{1-\tau^2}\frac{|Z|^2+|Z'|^2-2 Z \overline{Z'}}{2}\right)C_{\tau}^n(Z,Z')
 \\
+ \sqrt{\frac{n}{2\pi}} e^{n G_Q(Z,Z')} 
\left(\phi(Z) \overline{\phi(Z')}\right)^n
\frac{\sqrt{\phi'(Z) \overline{\phi'(Z')}}}{\phi(Z) \overline{\phi(Z')} -1}
\left(1+\mathcal O\left(\frac{1}{n}\right)\right),
\end{multline*}
as $n\to\infty$, uniformly for $Z$ and $Z'$ that stay a distance $\mu>0$ away from the motherbody $[-2\sqrt\tau, 2\sqrt\tau]$, and such that \text{$|\phi(Z) \overline{\phi(Z')}-1|\geq \mu$}. It is an interesting question whether such a result can be generalized to a larger class of potentials. \end{remark}

\subsubsection{Fermions in $d$ real dimensions $\mathbb{R}^d$}
A second set of special cases that we wish to present is that of fermions in $d$ dimensions. As before we scale $\mathcal K_n^{\rm Fermi}$ in \eqref{eq:kernelferm} as
$$
   \mathbb K_n^{\rm Fermi}(X,X')=n^{d/2} \mathcal K_n^{\rm Fermi}(\sqrt{n}\, X,\sqrt{n}\, X)
$$
We will proceed by deriving several results on that process, starting with the limiting density. 
The standard scalar product and norm on $\mathbb{R}^d$ is denoted by $\langle X,X'\rangle$ and $\lvert X \rvert=\sqrt{\langle X,X\rangle}$ respectively, for $X,X'\in\mathbb{R}^d$.

\begin{theorem} \label{thm:fermionsBulkDensity_intro}
    Let $X\in\mathbb R^d$. If $\lvert X\rvert < \sqrt 2$, then as $n\to\infty$ we have 
    \begin{align} \label{eq:fermionsBulkDensity1_intro}
    \frac{1}{n^d/d!} \mathbb K_n^{\rm Fermi} (X,X) =  \frac{d!}{2^d \pi^\frac{d}{2}} \frac{(2-\lvert X \rvert^2)^\frac{d}{2}}{\Gamma\left(\frac{d}{2}+1\right)} 
    +\mathcal O\left(\frac{1}{n}\right),
    \end{align}
    where the convergence is uniform on compact subsets.    
    If $\lvert X\rvert >\sqrt 2$, then $R_1(X)$ vanishes exponentially as $n\to\infty$, uniformly on compact subsets. 
    \end{theorem}

Note that for $d=1$, this is the semi-circle law for the Gaussian Unitary Ensemble. Note also that the vanishing exponent at the boundary depends on the dimension $d$.

Now that we have the limiting density, the next question is on the microscopic processes in the bulk and at the edges. The following results we discussed in \cite{DeDoMaSc}. Very recently Deleporte and Lambert rigorously proved the universality of these results, using results for Schr\"odinger operators \cite{DeLa}. Here we obtain a rather direct proof for the special case that the operator is the quantum harmonic oscillator. 

We start with the fluctuations in the bulk:
\begin{theorem} \label{thm:fermionsScalingBulk_intro}
    Let $X, U, V\in\mathbb R^d$ and suppose that $\lvert X\rvert < \sqrt 2$. For fixed $X$, denote $\nu(X) = 2^{-d} (2-\lvert X\rvert^2)^\frac{d}{2}$.\\
    Then as $n\to\infty$ we have
    \begin{align} \label{eq:fermionsScalingBulk_intro}
    \frac{1}{\nu(X) n^d}{\mathbb K}_n^{\rm Fermi}\left(X+\frac{U}{\nu(X)^\frac{1}{d} n}, X + \frac{V}{\nu(X)^\frac{1}{d} n}\right) 
    =  \frac{J_\frac{d}{2}\left(2\lvert U-V\rvert\right)}{(\pi\lvert U-V\rvert)^\frac{d}{2}}
    +\mathcal O\left(\frac{1}{n}\right),
    \end{align}
    uniformly for $(U,V)$ in compact sets. 
    \end{theorem} 
    For $d=1$, the leading term on the right-hand side of \eqref{eq:fermionsScalingBulk_intro} reduces to the sine-kernel. This follows from a well-known expression for the Bessel function $J_{\frac 12}(z)=\sin(z)\sqrt{2/\pi z}$. It may be interesting to point out that, to the best of our knowledge, the derivation of \eqref{eq:fermionsScalingBulk_intro} is different from other derivations in the literature, and, in our opinion, this is an interesting independent contribution to the vast literature on the universality of the sine process. 

The next result is on the microscopic processes near the edges. 
    
\begin{theorem} \label{thm:fermionsEdgeDensity_intro}
    Let $X, U, V\in\mathbb R^d$ and suppose that $X$ is fixed and $\lvert X\rvert = \sqrt 2$. 
    Then as $n\to\infty$ 
    \begin{align} \label{eq:fermionsAiryd_intro}
    \frac{1}{(n^\frac{2}{3} \sqrt 2)^d} \mathbb K_n^{\rm Fermi}\left(X+\frac{U}{n^\frac{2}{3} \sqrt 2}, X + \frac{V}{n^\frac{2}{3} \sqrt 2}\right) 
    =  \frac{1}{(2\pi)^d}\int_{\mathbb R^d} e^{-i \langle Q,U-V\rangle} \int_0^\infty\operatorname{Ai}\left(2^{2/3} \lvert Q\rvert^2 + \frac{\langle U+V,X\rangle}{2^{1/3}\lvert X\rvert}+s\right) ds\,d^d Q
    +\mathcal O\left(n^{-\frac{1}{3}}\right).
    \end{align}
    \end{theorem}     
In \cite{DeDoMaSc} the integrated Airy function was written as $\operatorname{Ai}_1(\zeta) = \int_\zeta^\infty \operatorname{Ai}(s) ds$.    
    For $d=1$, the expression yields the Airy kernel\cite{DeDoMaSc}. 
    Note that in \cite{BDMS} the $d$ dimensional integral was performed explicitly for $d\geq 2$, reducing it to a single integral including a Bessel function.

    \subsubsection{The general case: Strong non-Hermiticity regime in $\mathbb{C}^d$}

    We now proceed in stating our main results on the general case \eqref{eq:general_kernel_pre}, starting with the limiting density on $\mathbb C^d$. As in \eqref{eq:general_kernel_d=1} we scale the kernel $\mathcal K_n$ as 
    \begin{equation} \label{eq:gener_kernel_d}
        \mathbb K_n(Z,Z')= n^d \mathcal K_n\left(\sqrt n\, Z, \sqrt n\, Z' \right).
\end{equation}
    We show that under this rescaling the points will accumulate on the  $2d-$dimensional ellipsoidal domain
    \begin{align*}
    \mathcal E_\tau^d := \left\{Z\in \mathbb C^d : \frac{\lvert \operatorname{Re} Z\rvert^2}{(1+\tau)^2}
    + \frac{\lvert \operatorname{Im} Z\rvert^2}{(1-\tau)^2} < 1\right\}. 
    \end{align*}
    The limiting density will be uniform on this domain:
    \begin{theorem} \label{thm:limitingDensity}
        Fix $\tau\in (0,1)$. Let $Z\in \mathbb C^d$ and assume that $Z\not\in \partial \mathcal E_\tau^d$. As $n\to\infty$ we have
        \begin{align}
        \mathbb K_n(Z,Z) = \frac{n^d \mathfrak{1}_{Z\in \mathcal E_\tau^d}}{\pi^d (1-\tau^2)^d} + \mathcal O\left(e^{-cn}\right),
        \end{align}
        where $c>0$ can be chosen such that the $\mathcal O$ term is uniform on compact subsets of $\mathbb C^d \setminus \partial \mathcal E_\tau^d$.
        \end{theorem}

        The next step is to compute the local scaling limit of the kernel in the bulk. For simplicity we restrict our selves to the origin, but the same limit holds for any other point in the bulk. 
        \begin{theorem} \label{thm:Cd_bulk_intro}
        Let $Z\in \mathcal E_\tau^d$, and $U, V\in\mathbb{C}^d$. Then we have 
        the following factorisation into one-dimensional Ginibre kernels:
        \begin{align*}
            \lim_{n\to\infty} \frac{1}{n^d} \mathbb K_n\left(Z+\frac{U}{\sqrt n}, Z+\frac{V}{\sqrt n}\right) \prod_{j=1}^n C_{\tau}^n(U_j, V_j)
            =  \frac{1}{\pi^d (1-\tau^2)^d} \prod_{j=1}^d
   \exp\left(-\frac{|U_j|^2+|V_j|^2-2 U_j \overline{V_j}}{2(1-\tau^2)}\right),
            \end{align*}
            where $C_{\tau}^n(U_j,V_j) = \exp\left(\frac{i n \tau }{ 1-\tau^2} \frac{\operatorname{Im}(U_j^2-V_j^2)}{2}\right)$ is a cocycle. The convergence is uniform on compact sets. 
        \end{theorem}
        It is important to observe that the factors $C_{\tau}^n$  drop out when computing the determinant for the correlation functions and are therefore not relevant. 
        It is interesting that for $\mathbb C^d$ the limiting kernel is a product of bulk kernels for $d=1$, but in the Hermitian limit $\tau \uparrow 1$ it is a more complicated function of the $2d$ variables. 
           The connection between these two regimes is provided by the weak non-Hermiticity regime below. A counterpart of Theorem \ref{thm:Cd_bulk_intro} in the setting of compact complex manifolds can be found in \cite[Theorem 1.1]{Berman} (indeed, $\mathbb C^d$ is compactified by the complex projective space $\mathbb P^n$).
         
Before we move to the next section let us mention that  also at strong non-Hermiticity a particular edge regime should exist, that generalises the well-known complementary error function kernel, known as the Faddeeva plasma kernel, at $d=1$, compare \cite{TV}. However, to attain such a regime   for $d>1$, a different saddle point analysis would have to be made that is more involved. Very recently, this analysis has been carried out in \cite{Molag}, and a higher dimensional analogue of the Faddeeva plasma kernel emerges as the asymptotic behavior of the correlation kernel near the edge $\partial\mathcal E_\tau^d$.

        \subsubsection{The general case: Weak non-Hermiticity regime in $\mathbb{C}^d$}

        The next results provide an interpolation between these correlations in $\mathbb{R}^d$ and in $\mathbb{C}^d$ at strong non-Hermiticity.
  
\begin{theorem} \label{thm:weakNonHd_intro}
    Let $U, V\in \mathbb C^d$ with product $UV=\sum_{j=1}^dU_jV_j$, and assume that $X\in\mathbb R^d$ satisfies $\lvert X\rvert < 2$. 
    Denote 
    \begin{align*}
    \nu(X) = \frac{1}{(2\pi)^d} (4-\lvert X \rvert^2)^\frac{d}{2} \qquad\text{ and }\qquad
    \tau = 1 - \frac{\kappa}{\nu(X)^\frac{1}{d} n},
    \end{align*} 
    for some constant $\kappa>0$.  Then, as $n\to\infty$, we have 
    \begin{multline} \label{eq:thm:weakNonHd_intro}
    \frac{1}{\nu(X)^2 n^{2d}} \mathbb K_n\left(X+\frac{U}{\nu(X)^\frac{1}{d} n}, X+\frac{V}{\nu(X)^\frac{1}{d} n}\right)
    = C_X(U,V)
       \frac{e^{-\frac{\lvert \operatorname{Im} U\rvert^2+\lvert \operatorname{Im} V\rvert^2}{2\kappa}}}{(4 \kappa)^\frac{d}{2}} 
   4 \int_0^{1} e^{- \kappa \pi^2 t^2}  t^d \widehat{J}_{\frac{d}{2}-1}\left((U-\overline V)^2 \pi^2 t^2\right) dt
    + \mathcal O\left(\frac{1}{n}\right),
    \end{multline}
    where we define the even function $\widehat{J}_\nu(z^2)=(2/z)^\nu J_\nu(z)$ and 
        $C_X(U,V)= e^{i \nu(X)^{-\frac{1}{d}} \frac{1}{2} \langle X, \operatorname{Im}(U-V)\rangle}$ is a cocycle. 
    \end{theorem}
    
Notice that the first factor on the right-hand side of \eqref{eq:thm:weakNonHd_intro} is a cocycle, and, upon multiplication with another cocycle as in \eqref{eq:deltalim}, the second factor becomes a delta function as $\kappa\downarrow 0$.  In the limit $\kappa\downarrow 0$, the remaining integral in \eqref{eq:weakNonHd2} becomes proportional to \cite[6.561.5]{Grad}
\begin{align}
\xi^{1-\frac{d}{2}} \int_0^{1}  t^\frac{d}{2} J_{\frac{d}{2}-1}\left(
\xi\pi t\right) dt
= 
\frac{\pi^{
-1}}{
\xi^\frac{d}{2}} J_\frac{d}{2}\left(\xi\pi 
\right),
\end{align}
where the right-hand side is again even in $\xi$. Indeed, this limit provides  exactly 
the claimed interpolation,  
the bulk limit for non-interacting fermions in $d$ real dimensions when $U, V\in\mathbb R^d$ in Theorem \ref{thm:fermionsScalingBulk_intro}. 

The interpolation to strong non-Hermiticity in the limit $\kappa\to \infty$ can also be checked. Here, the  complex arguments have to be rescaled, keeping $U_j/\sqrt{\kappa}=u_j$ and $V_j/\sqrt{\kappa}=v_j$ fixed in the limit. After rescaling the integration variable in \eqref{eq:weakNonHd2} $t\to s =\sqrt{\kappa}\pi\nu(X)t$, the integral extends to infinity and can be computed, using \cite[6.631.4]{Grad}
\begin{equation*}
\int_0^\infty t^{\nu+1}J_\nu(bt) e^{-at^2} =\frac{b^\nu}{(2a)^{\nu+1}} e^{-b^2/4a}.
\end{equation*}
Thus, the exponential function on the right-hand side factorises in the $u_j$ and $v_j$. Collecting all prefactors
we arrive at Theorem \ref{thm:Cd_bulk_intro} in terms of these rescaled variables.\\


Finally, we also  compute the weak non-Hermiticity limit near the edge.
      
\begin{theorem} \label{thm:EGEdbulkWeakNonH_intro}
    Let $\tau = 1 - \kappa n^{-\frac{1}{3}}$, for some constant $\kappa>0$. Let $U,V\in\mathbb C^d$ and let $X = (1+\tau) \Omega$, where $\Omega$ is a fixed element of the unit sphere in $\mathbb R^d$. Then, as $n\to\infty$, we have 
    \begin{multline*}
    \frac{1}{n^\frac{4 d}{3}} \mathbb K_n\left(X+\frac{U}{n^\frac{2}{3}}, X+\frac{V}{n^\frac{2}{3}}\right)
    = 
    C_\Omega^n(U,V)
  \frac{1}{(\kappa\pi)^{\frac{d}{2}}(2\pi)^d}
   e^{-\frac{\lvert \operatorname{Im} U\rvert^2+\lvert \operatorname{Im} V\rvert^2}{2\kappa}}
e^{\frac{1}{6}\kappa^3+\frac{1}{2}\kappa \langle U+\overline V, \Omega\rangle}\\
    \times 
    \int_{\mathbb R^{d}} e^{-i Q(U-\overline V)}\int_0^\infty e^{2^{\frac23}\kappa s}
    \operatorname{Ai}\left(2^\frac{2}{3} \lvert Q\rvert^2 + 2^{-\frac{4}{3}} (\kappa^2 + 2 \Omega(U+\overline V)) +s\right)  d^{d}Q\ ds
    +\mathcal O(n^{-\frac{1}{3}}),
    \end{multline*}
    where $C_\Omega^n(U,V) = e^{i n^\frac{1}{3} \langle \Omega, \operatorname{Im}(U-V)\rangle}$ is a cocycle. 
    \end{theorem}
    
It is not difficult to see that in the limit $\kappa\to0$ we recover the real $d$-dimensional Airy-kernel at the edge \eqref{eq:fermionsAiryd_intro}, after removing the appropriate cocyle and the delta function $\delta(\Im U)$ resulting from the limit.

In principle, Theorem \ref{thm:EGEdbulkWeakNonH_intro} could be taken as a starting point to reach strong non-Hermiticity at the edge in the limit $\kappa\to\infty$, in terms of rescaled variables $U/\sqrt{\kappa}$ and $V/\sqrt{\kappa}$. Clearly, this requires a saddle point analysis on $\mathbb{R}^d$ that depends on the dimension. For $d=1$ the matching of the above theorem with strong non-Hermiticity at the edge has been performed in \cite{ABe}.

\subsection{Overview of the proofs}

The proofs of our main results are based on an integral representation of the kernel \eqref{eq:general_kernel_pre}, rescaled as in \eqref{eq:gener_kernel_d}. We start by recalling the Mehler kernel: for $z, w\in\mathbb C$
\begin{equation} \label{eq:Mehler}
    \frac{1}{(1-\tau^2)^{1/2}} \exp\left(-\frac{\tau^2(z^2+w^2)}{1-\tau^2} + \frac{2\tau z w}{1-\tau^2}\right)
= \sum_{j=0}^{\infty} \frac{(\tau/2)^j}{j!} H_j(z) H_j(w).
\end{equation}
The square root $(1-\tau^2)^{1/2}$ is taken such that it is analytic in $\mathbb C\setminus \left((-\infty,-1]\cup [1, \infty)\right)$ and takes positive values for $\tau \in (-1,1)$. 
By taking products in the $d$ variables we obtain 
\begin{equation} \label{eq:MehlerFordFactors}
    \frac{1}{(1-\tau^2)^{d/2}} \exp\left(-\frac{\tau^2\sum_{j=1}^d (z_j^2+w_j^2)}{1-\tau^2} + \frac{2\tau \sum_{j=1}^d z_j w_j}{1-\tau^2}\right)
= \sum_{j_1, \ldots,j_d=0}^{\infty} \frac{(\tau/2)^{j_1+\ldots +j_d}}{j_1! \cdots j_d!} \prod_{\ell=1}^d H_{j_\ell}(z_\ell) H_{j_\ell}(w_\ell).
\end{equation}
Note that \eqref{eq:general_kernel_pre} is a product of the weight functions and a cut-off of this infinite series, and this cut-off can be obtained by a contour integral. 
We define, for $z,w \in \mathbb C^d$
\begin{equation}
    K_n(z,w)= \frac{1}{2\pi i} \oint_{\gamma_0} \exp\left(-\frac{s^2\sum_{j=1}^d (z_j^2+w_j^2)}{1-s^2} + \frac{2 s \sum_{j=1}^d z_j w_j}{1-s^2}\right)
    \frac{1-(\tau/s)^n}{s-\tau} \frac{ds}{(1-s^2)^{d/2}},
\end{equation} 
where $\gamma_0$ is a small (e.g., $|\gamma_0|<\tau$) counter-clockwise oriented contour around $s=0$. As before, the square root $(1-s^2)^{d/2}$ is taken such that it is analytic in $\mathbb C\setminus \left((-\infty,-1)\cup (1, \infty)\right)$ and takes positive values for $s \in (-1,1)$.  Note that we can rewrite this as 
\begin{align} \nonumber
    K_n(z,w)&= - \frac{1}{2\pi i} \oint_{\gamma_0} \exp\left(-\frac{s^2\sum_{j=1}^d (z_j^2+w_j^2)}{1-s^2} + \frac{2 s \sum_{j=1}^d z_j w_j}{1-s^2}\right)
    \frac{(\tau/s)^n}{s-\tau} \frac{ds}{(1-s^2)^{d/2}}\\ \label{eq:integral_representation_kernel}
    &=- \frac{1}{2\pi i} \oint_{\gamma_0} \exp\left(\frac{s\sum_{j=1}^d (z_j+w_j)^2}{2(1+s)} - \frac{s\sum_{j=1}^d (z_j-w_j)^2}{2(1-s)}\right)
    \frac{(\tau/s)^n}{s-\tau} \frac{ds}{(1-s^2)^{d/2}},
\end{align} 
provided that we take $\gamma_0$ such that it does not go around  the pole $s=\tau$. 
Then \eqref{eq:gener_kernel_d} has the representation
\begin{equation} \label{eq:kernelboldKtoK}
 \mathbb K_n(Z,Z')= \frac{n^d}{\pi^d (1-\tau^2)^\frac{d}{2}} 
 \sqrt{\prod_{\ell=1}^d \omega(\sqrt n Z_\ell) \omega(\sqrt n Z'_\ell)}  K_n\left(\sqrt{\frac{n}{2\tau}} Z, \sqrt{\frac{n}{2\tau}} {Z'}\right).
\end{equation}
To prove our main results we are thus left to compute a steepest descent analysis of the single integral representation of  $K_n$. An important point is that the integrand depends only in a simple way on the dimension $d$. This should be compared to more standard strategies based on integral formulas of the Hermite functions. Indeed, in \cite{Be}  the analysis was based on a double integral formula for the kernel in the case $d=1$. By following that approach one would obtain a representation with a $2d$ dimensional integral, and the complexity increases with the dimension $d$. Using the representation \eqref{eq:integral_representation_kernel} we are able to perform an analysis for all dimensions $d$ simultaneously.  

In Section \ref{sec:steepest} we will perform a steepest descent analysis for the integral \eqref{eq:integral_representation_kernel} using general parameters. These results will be translated to the case $d=1$ in Section \ref{sec:EGE} and used to derive various asymptotic results, including those mentioned above. Similarly, the Hermitian limit $\tau \uparrow 1$ will be discussed in Section \ref{sec:fermions} and the general non-Hermitian case will be treated in \text{Section \ref{sec:general}}.


\section{Steepest descent analysis} \label{sec:steepest}

In this section we consider a fixed $\tau\in(0,1]$ and here we let $d$ be a complex number. Using steepest descent arguments, we shall find the large $n$ behavior of 
\begin{align} \label{eq:defIn}
I_n(d,\tau;z,z') := -\frac{1}{2\pi i} \oint_{\gamma_0} \frac{e^{n F(s)}}{s-\tau} \frac{ds}{(1-s^2)^\frac{d}{2}},
\end{align}
where $\gamma_0$ is a small contour with positive orientation, that encloses $0$ but not $\tau$, $(1-s^2)^\frac{d}{2}$ is defined with cut $(-\infty,-1]\cup [1,\infty)$ and positive on $(-1,1)$, and $F(s) = F(\tau;z, z'; s)$ is defined by
\begin{align} \label{eq:defF}
F\left(\tau;z, z'; s\right) :=& \frac{s(z+z')^2}{2(1+s)}-\frac{s (z-z')^2}{2(1-s)} - \log s + \log \tau.
\end{align}
The branch for the logarithm is not relevant for \eqref{eq:defIn}. We remark that, since $\tau$ merely enters \eqref{eq:defF} in the form of an additive constant, the dependence on $\tau$ is irrelevant for the steepest descent analysis of $F$ when $\tau$ is fixed. 

Note that \eqref{eq:defIn} coincides with $K_n\left(\sqrt{\frac{n}{2\tau}} Z, \sqrt{\frac{n}{2\tau}}\, \overline{Z'}\right)$, as defined via \eqref{eq:integral_representation_kernel}, if we correctly identify $z$ and $z'$. 

\subsection{The saddle points and associated expressions}

It turns out that the saddle points of $F$ have a remarkably simple form if we view them in elliptic coordinates. We denote
\begin{align} \label{eq:defEllipticCoord}
\begin{array}{l}
z = \sqrt 2 \cosh\left(\xi+i\eta\right),\\
z' = \sqrt 2 \cosh\left(\xi'+i\eta'\right),
\end{array}
\end{align}
where $\xi\geq 0$ and $\eta\in (-\pi,\pi]$ when $\xi>0$, while $\eta\in [0,\pi]$ when $\xi=0$, and similarly for $\xi'$ and $\eta'$. Note that any constant value of $\xi$ corresponds to an ellipse with vertex $\sqrt 2 \cosh \xi$ and co-vertex $\sqrt 2 \sinh \xi$. We will see later that the case $\xi=\xi_\tau$ is of particular importance for the asymptotics of our models, where
\begin{align}
\xi_\tau = -\frac{1}{2} \log \tau. 
\end{align}
Next we define the expressions 
\begin{align} \label{eq:defab}
a = e^{\xi+\xi'} e^{i(\eta+\eta')}\qquad\text{ and }\qquad 
b = e^{\xi-\xi'} e^{i(\eta-\eta')}.
\end{align}
We will show that the saddle points generically are given by $a,a^{-1},b$ and $b^{-1}$. Degenerate situations occur when saddle points collide with each other or coalesce with one of the poles at $s=\pm 1$. It  is important to know  when these different situations happen and this is explained in  the following proposition.  

\begin{proposition} \label{prop:saddlePointsDef}
Let $z,z'\in \mathbb C\setminus \{-\sqrt 2, \sqrt 2\}$, then the saddle points of $s\mapsto F(\tau;z,z';s)$ are simple, and we have the following:
\begin{itemize}
\item[(i)] When $z \neq  \pm z'$, there are exactly four saddle points  given by $a, a^{-1}, b$ and $b^{-1}$.
\item[(ii)] When $z=\pm z'$ and $z\neq 0$, there are exactly two  saddle points, which are given by $a$ and $a^{-1}$.
\end{itemize}
If $z\in \{-\sqrt{2},\sqrt{2}\}$ or $z'\in \{-\sqrt{2},\sqrt{2}\}$, then all saddle points have order two, and we have the following:
\begin{itemize}
    \item[(iii)] If $z\in \{-\sqrt 2, \sqrt 2\}$ and $z'\not\in \{-\sqrt 2, \sqrt 2\}$, then we have two saddle points $a=b^{-1}$ and $a^{-1}=b$.
        \item[(iv)] If $z\not\in \{-\sqrt 2, \sqrt 2\}$ and $z'\in \{-\sqrt 2, \sqrt 2\}$, then we have two saddle points $a=b$ and $a^{-1}=b^{-1}$;
    \item[(v)] If $z=\pm z'\in\{-\sqrt 2, \sqrt 2\}$, then we have one saddle point $a^{-1}=a=b=b^{-1}=\pm 1$.
\end{itemize}
Finally, when $z=z'=0$ there are no saddle points.
\end{proposition}

\begin{proof}
First we prove cases (i), (iii) and (iv). One may verify that $a$ and $b$ do not equal $1$ or $-1$ in each of these cases. 
A simple calculation yields
\begin{align} \label{eq:dFroots}
F'(s) = -\frac{s^4 - 2 z z' s^3 + 2(z^2+z'^2-1) s^2 - 2 z z' s+1}{s (s^2-1)^2}. 
\end{align}
A well-known multiplication formula for hyperbolic cosines gives
\begin{align}
z z' = \cosh(\xi+\xi'+i(\eta+\eta'))+\cosh(\xi-\xi'+i(\eta-\eta')),
\end{align}
and
\begin{align}
z^2+z'^2-1
&= \cosh(2\xi+2i\eta) + \cosh(2\xi'+2i\eta')+1\\
&= 2 \cosh(\xi+\xi'+i(\eta+\eta')) \cosh(\xi-\xi'+i(\eta-\eta'))+1.
\end{align}
This means that
\begin{align} \nonumber
(s-a)&(s-a^{-1})(s-b)(s-b^{-1})\\ \nonumber
&= (s^2-2 \cosh(\xi+\xi'+i(\eta+\eta')) s +1) (s^2-2 \cosh(\xi-\xi'+i(\eta-\eta')) s +1)\\ \label{eq:saddlePolynomial}
&= s^4 - 2 z z' s^3 + 2(z^2+z'^2-1) s^2 - 2 z z' s+1.
\end{align}
Comparing with \eqref{eq:dFroots}, we conclude that $F$ has saddle points $a, a^{-1}, b$ and $b^{-1}$. 

Now let us consider case (ii). Then we have
\begin{align}
F'(s) = -\frac{s^2\mp 2(z^2- 1))s+1}{s(s\pm 1)^2}.
\end{align}
Indeed, we have
\begin{align}
(s-a)(s-a^{-1})
= s^2\mp 2 \cosh(2\xi+2i \eta) s +1
= s^2\mp 2(z^2-1) s +1.
\end{align}
We may conclude that $a$ and $a^{-1}$ are the saddle points in this case, if we can argue that $a \neq \mp 1$. If this was the case, then we would have $\xi=\xi'=0$ and either $\eta=\eta'\in\{0,\pi\}$ or $\eta=\pi-\eta'$. The first option is excluded, and the other option corresponds with $z=-z'$ and $a=-1 \neq 1$. 

The reader may verify that case (v) can be proved by some straightforward algebra. 
\end{proof}


It is clear from Proposition \ref{prop:saddlePointsDef} and the definition in \eqref{eq:defab} that $a^{-1}$ is always the closest saddle point to the origin. It is on the unit circle when both $z, z'\in [-\sqrt 2, \sqrt 2]$, but otherwise strictly inside the unit disk. In general, the radii of $b$ and $b^{-1}$ are sandwiched between those of $a^{-1}$ and $a$. Only when $z\in [-\sqrt 2, \sqrt 2]$ or $z'\in [-\sqrt 2, \sqrt 2]$, can it happen that $b$ or $b^{-1}$ have the same radius as $a^{-1}$. See also Figure \ref{Fig1a}.
\begin{corollary} \label{cor:saddleinZW}
Assume that $z, z'\in \mathbb C\setminus [-\sqrt 2, \sqrt 2]$.
When $z\neq\pm z'$, the saddle points of $s\mapsto F(\tau;z,z';s)$ can be expressed as
\begin{align}
\begin{array}{ll}
a = \frac{1}{2} \left(z+\sqrt{z^2-2}\right) \left(z'+\sqrt{z'^2-2}\right), & a^{-1} = \frac{1}{2} \left(z-\sqrt{z^2-2}\right) \left(z'-\sqrt{z'^2-2}\right) \\
{ } & { }\\
b = \frac{1}{2} \left(z+\sqrt{z^2-2}\right) \left(z'-\sqrt{z'^2-2}\right), & b^{-1} = \frac{1}{2} \left(z-\sqrt{z^2-2}\right) \left(z'+\sqrt{z'^2-2}\right).
\end{array}
\end{align}
When $z=\pm z'$, there are only two saddle points, which are given by $a$ and $a^{-1}$ above.\\
In particular, the saddle points are analytic functions of $(z,z')$ on $(\mathbb{C}\setminus [-\sqrt 2,\sqrt 2])^2$.\\
\end{corollary}

\begin{proof}
With $\xi >0$ and $\eta\in (-\pi,\pi]$, we have $\xi+i\eta = \operatorname{arccosh} \frac{z}{\sqrt 2}$. The inverse hyperbolic cosine can alternatively be written as $z\mapsto \log\left(z+\sqrt{z^2-1}\right)$ (where we might have to add multiples of $2\pi i$ depending on our choice of cut). We infer that
\begin{align}
e^{\xi+i\eta} = \frac{z+\sqrt{z^2-2}}{\sqrt 2}.
\end{align}
The proof follows after the observation that taking the multiplicative inverse of this expression amounts to changing the $+$ sign to a $-$ sign in the right-hand side.  
\end{proof}

\begin{proposition} \label{prop:Finsaddleelliptic}
Let $z, z'\in \mathbb C$. We have
\begin{align} \label{prop:Finsaddleelliptica}
F(a) &= 1 + \log \tau - \xi - \xi' - i(\eta+\eta') +  \frac{1}{2} e^{2 (\xi + i \eta)} + \frac{1}{2} e^{2 (\xi'+ i \eta')},\\ \label{prop:Finsaddleelliptica-}
F(a^{-1}) &= 1 + \log \tau + \xi + \xi' + i(\eta+\eta') +  \frac{1}{2} e^{-2 (\xi + i \eta)} + \frac{1}{2} e^{-2 (\xi'+ i \eta')},\\
F(b) &= 1 + \log \tau - \xi + \xi' - i(\eta-\eta') +  \frac{1}{2} e^{2 (\xi + i \eta)} + \frac{1}{2} e^{-2 (\xi'+ i \eta')},\\
F(b^{-1}) &= 1 + \log \tau + \xi - \xi' + i(\eta-\eta') +  \frac{1}{2} e^{-2 (\xi + i \eta)} + \frac{1}{2} e^{2 (\xi'+ i \eta')}.
\end{align}
These identities hold up to multiples of $2\pi i$ depending on the choice of the branch of the logarithm in \eqref{eq:defF} but this is irrelevant for the integral $I_n$ in \eqref{eq:defIn}.
\end{proposition}
\begin{proof}
First we treat the case $z\neq \pm z'$. Taking $s=-1$ and $s=1$ in \eqref{eq:saddlePolynomial} we find
\begin{align} \label{eq:aa-bb-1}
(1+a) (1+a^{-1}) (1+b) (1+b^{-1}) &= 2 (z+z')^2.\\ \label{eq:aa-bb-2}
(1-a) (1-a^{-1}) (1-b) (1-b^{-1}) &= 2 (z-z')^2
\end{align}
This implies that
\begin{align*}
F(a) &= \frac{(1+a)(1+b)(1+b^{-1})}{4} + \frac{(1-a)(1-b)(1-b^{-1})}{4} - \log a + \log \tau\\
&= \frac{4 + 2a(b+b^{-1})}{4} - \log a + \log \tau,
\end{align*}
which, expressed in elliptic coordinates, becomes \eqref{prop:Finsaddleelliptica}. One can also extract from this reasoning that, given a fixed $z' \neq \pm \sqrt 2$, the function $z\mapsto F(a)$ is entire. Hence the case $z=\pm z'$ follows by continuity. The cases where both $z, z'\in \{-\sqrt 2, \sqrt 2\}$ follow by some straightforward algebra. 

The expressions in $a^{-1}, b$ and $b^{-1}$ follow in similar fashion. 
\end{proof}

For the saddle point contributions we also need to understand $F''$ in the saddle points. 

\begin{proposition} \label{prop:F''a-1}
We have
\begin{align}
F''(a^{\pm 1}) = \mp 2 a^{\mp 2} \frac{\sinh(\xi+i\eta) \sinh(\xi'+i\eta')}{\sinh(\xi+i\eta+\xi'+i\eta')}
\qquad\text{ and }\qquad
F''(b^{\pm 1}) = \pm 2 b^{\mp 2} \frac{\sinh(\xi+i\eta) \sinh(\xi'+i\eta')}{\sinh(\xi+i\eta-\xi'-i\eta')},
\end{align}
unless these are not well-defined, this happens when $z, z'\in\{-\sqrt 2, \sqrt 2\}$ for both identities,  $z=-z'\in (-\sqrt 2, \sqrt 2)$ for the first identity, and $z=z'\in (-\sqrt 2, \sqrt 2)$ for the second identity. 
\end{proposition}

Note however that $b$ is not a saddle point when $z=\pm z'$ due to Proposition \ref{prop:saddlePointsDef}, unless $z, z'\in\{-\sqrt 2, \sqrt 2\}$.

\begin{proof}
We have that
\begin{align}
F''(s) &= - \frac{(z+z')^2}{(1+s)^3} - \frac{(z-z')^2}{(1-s)^3} + \frac{1}{s^2}. 
\end{align}
In what follows we assume that $z\neq\pm z'$. Using \eqref{eq:aa-bb-1} and \eqref{eq:aa-bb-2}, we find
\begin{align} \nonumber
F''(a^{-1}) &= - \frac{1}{2} \frac{(1+b)(1+b^{-1})(1+a)}{(1+a^{-1})^2} - \frac{1}{2} \frac{(1-b)(1-b^{-1})(1-a)}{(1-a^{-1})^2} + a^2\\
&= \frac{2-a (b+b^{-1})}{1-a^{-2}} + a^2
= \frac{(a - b)(a - b^{-1})}{1-a^{-2}}. 
\end{align}
An easy calculation shows that
\begin{align}
(a-b)(a-b^{-1}) &= 4 a \sinh(\xi+i\eta) \sinh(\xi'+i\eta').
\end{align}
The case $z=\pm z'$ can also be worked out, and yields the same formula in the end.

The expressions in $a, b$ and $b^{-1}$ follow in similar fashion. 
\end{proof}

    In the next paragraph we will make the first step in the  steepest descent analysis by deforming the contour $\gamma_0$ such that it passes through one or more saddle points. We will show that we can always deform the contour such that it passes through the saddle point $a^{-1}$. In special cases, it passes through other saddle points as well. To give the reader some intuition of how $\gamma_0$ can be deformed, we mention that, in general (i.e., $z\neq z'$), the set of $s$ such that $\operatorname{Re} F(s)\geq r$ consist of three small  regions when $r$ is a large positive number. One region is a neighborhood of the origin, another region has $1$ on its boundary, while the last region has $-1$ on its boundary. Since our initial contour has to go around the origin,  we can always deform, for sufficiently large $r$, the contour such that it goes around the region enclosing the origin.  As we lower $r$, the regions grow in size, and start to merge, exactly when they touch in a saddle point (or several saddle points simultaneously). Since we always take our contour through $a^{-1}$, we should visualize the region(s) $\operatorname{Re} F(s)\geq r$ in the case $r=\operatorname{Re} F(a^{-1})$, as has been done in Figures \ref{Fig1}-\ref{Fig2aab} below. 

\subsection{Integration contour}

We will now discuss how to deform the integration contour in the saddle point analysis. Our choice is based on the following result. 

\begin{theorem} \label{lem:as=1}
Let $z,z' \in\mathbb C$. We have the inequality
\begin{align} \label{eq:as=1}
\operatorname{Re} F(s) &\leq \operatorname{Re} F(a^{-1}), & |s| = |a|^{-1}.
\end{align} 
\begin{itemize}
\item[(i)] When $z\not\in [-\sqrt 2, \sqrt 2]$ and $z'\not\in [-\sqrt 2, \sqrt 2]$, we have equality if and only if $s = a^{-1}$.
\item[(ii)] When $z\not\in [-\sqrt 2, \sqrt 2]$ and $z'\in [-\sqrt 2, \sqrt 2]$, we have equality if and only if $s = a^{-1}$ or $s=b^{-1}$. 
\item[(iii)] When $z\in [-\sqrt 2, \sqrt 2]$ and $z'\not\in [-\sqrt 2, \sqrt 2]$, we have equality if and only if $s = a^{-1}$ or $s=b$.
\item[(iv)] When $z\in [-\sqrt 2, \sqrt 2]$ and $z'\in [-\sqrt 2, \sqrt 2]$, we have equality for all $s$. 
\end{itemize}
\end{theorem}
Before we come to the proof of this theorem (which will take the rest of this paragraph), we discuss how we deform the contour $\gamma_0$.  Instead of working with the exact steepest descent paths leaving from the relevant saddle points, we will work with different paths that are of steep descent but not necessarily steepest. The benefit is that they are particularly simple. Indeed, apart from the  special case where $z$ and $z'$ both are in $[-\sqrt 2, \sqrt 2]$, we deform the contour $\gamma_0$ to the circle with radius $|a|^{-1}$, centered at the origin. By Theorem \ref{lem:as=1}(i), if neither $z \in [-\sqrt 2, \sqrt 2]$ nor $z' \in [-\sqrt 2, \sqrt 2]$ then $a^{-1}$ is the only saddle point on that circle and this point gives the main contribution to the integral over the circle. This is the situation in Figure \ref{Fig1}.  In case either $z \in [-\sqrt 2, \sqrt 2]$  or $z' \in [-\sqrt 2, \sqrt 2]$ then  either $b$ or $b^{-1}$ is a second saddle point on the circle, as illustrated in Figure \ref{Fig1a}. In case both $z \in [-\sqrt 2, \sqrt 2]$  and $z' \in [-\sqrt 2, \sqrt 2]$ the circle is no longer a valid choice for a saddle point analysis since the real part of $F$ is constant on that circle. Still, starting from the circle it is not difficult to see that one can deform the circle slightly to obtain contours on which the real part of $F$ is maximal only in the saddle points. To avoid cumbersome details, we will often not explicitly spell out the contours in these cases, but explain our choice by means of a figure.

We will prove Theorem \ref{lem:as=1} in several steps, beginning with the case that $z$ and $z'$ are on the same ellipse. 

\begin{lemma} \label{lem:as=1sameEllipse}
Theorem \ref{lem:as=1} holds when $\xi=\xi'$.
\end{lemma}

\begin{proof}
Let us first consider the case $\xi=\xi'>0$. Then we have $|a|^{-1}<1$. Using \eqref{eq:aa-bb-1}, we have
\begin{align} \nonumber
-2 (z+z')^2 \frac{1}{1+a^{-1} s} 
&= -(2+b+b^{-1}) (1+a^{-1})^2 a \frac{1}{1+a^{-1} s}\\ \nonumber
&= -(2+b+b^{-1}) (1+a^{-1})^2 a \left(-\frac{a^{-1}}{1-|a|^{-2}} \frac{\overline a^{-1} + s}{1+a^{-1} s} + \frac{1}{1-|a|^{-2}}\right)\\ \label{eq:Blaschke}
&= 2 (1 + \cos(\eta-\eta')) \frac{|1+a^{-1}|^2}{1-|a|^{-2}} \left(\frac{1+a^{-1}}{1+\overline a^{-1}} \frac{\overline a^{-1} + s}{1+a^{-1} s}\right) + \text{constant},
\end{align}
where we used that $2 + (b+b^{-1}) = 2 (1 + \cos(\eta-\eta'))$. A Blaschke factor maps the unit circle bijectively to itself. Hence the real part of \eqref{eq:Blaschke} attains its maximum in $s=1$ when restricted to the unit circle, and this maximum is unique unless $1+\cos(\eta-\eta')=0$. 
Entirely analogously, using \eqref{eq:aa-bb-2}, we have
\begin{align*}
-2 (z-z')^2 \frac{1}{1-a^{-1} s} 
=  2 (1 - \cos(\eta-\eta')) \frac{|1-a^{-1}|^2}{1-|a|^{-2}} \left(\frac{1-a^{-1}}{1-\overline a^{-1}} \frac{-\overline a^{-1} + s}{1-a^{-1} s}\right) + \text{constant}.
\end{align*}
Again, we infer that the real part attains its maximum in $s=1$ when restricted to the unit circle, and this maximum is unique unless $1 - \cos(\eta-\eta')=0$. 

By the preceding, we conclude that the real part of
\begin{align} \label{eq:diffwithFconst}
\frac{(z+z')^2}{2} \frac{s}{1+s}
- \frac{(z-z')^2}{2} \frac{s}{1-s}
= -\frac{(z+z')^2}{2} \frac{1}{1+s}
- \frac{(z-z')^2}{2} \frac{1}{1-s} + \text{constant},
\end{align}
attains its maximum in $s=a^{-1}$ when restricted to $|s|=|a|^{-1}$.  At least one of $2 \pm (b+b^{-1})$ is non-zero, hence this maximum is unique. Since the real part of \eqref{eq:diffwithFconst} differs from $\operatorname{Re} F(s)$ by only an additive constant on the circle $|s|=|a|^{-1}$, we conclude that $\operatorname{Re} F$ attains its maximum uniquely in $s=a^{-1}$. 

Finally, we treat the case $\xi=\xi'=0$. Then we have, using \eqref{prop:Finsaddleelliptica-} in the last line, that for all $t\in (-\pi,\pi]$
\begin{align*}
\operatorname{Re} F(e^{it}) &= \frac{2(\cos\eta+\cos\eta')^2}{2} \operatorname{Re}\left(\frac{e^\frac{it}{2}}{2 \cos\frac{t}{2}}\right) 
- \frac{2(\cos\eta-\cos\eta')^2}{2} \operatorname{Re}\left(\frac{e^\frac{it}{2}}{-2 i \sin\frac{t}{2}}\right) + \log\tau\\
&= 1 + \frac{\cos 2\eta}{2} + \frac{\cos 2\eta'}{2}+\log\tau= \operatorname{Re} F(a^{-1}).
\end{align*}
This proves case (iv). 
\end{proof}

To prove Theorem \ref{lem:as=1} for the remaining $(z,z')$, we introduce the following auxiliary function. 

\begin{definition} \label{def:defh}
For $\zeta\in \mathbb C\setminus \{a, - a, \overline a^{-1}, -\overline a^{-1}\}$, we define the function $h(\zeta)=h(\tau;z,z';\zeta)$ by
\begin{align}
h(\tau;z,z';\zeta) = \frac{(\zeta-1)^2 (\alpha \zeta^2+\beta \zeta+\frac{a^2}{\overline a^2} \overline \alpha)}{(\zeta^2-a^2)(\zeta^2-\overline a^{-2})},
\end{align}
where $\alpha = \frac{1}{2}(z^2+z'^2)- \operatorname{Re} F(a^{-1}) + \log \tau + \log |a|$ and $\beta =  2 \alpha +  \overline{a^{-1} z z'} - a z z'$.\\
\end{definition}

Before clarifying the relation between $h$ and $\operatorname{Re} F$, we derive formulas for $\alpha$ and $\beta$ in terms of elliptic coordinates \eqref{eq:defEllipticCoord}. 

\begin{lemma} \label{lem:formulasalphabeta}
With $\alpha$ and $\beta$ as in Definition \ref{def:defh}, we have
\begin{align*} 2\alpha &= \sinh(2\xi) e^{2i\eta} + \sinh(2\xi') e^{2i\eta'}\\
-\frac{\overline a}{a} \beta &= \sinh(2\xi+2\xi'). 
\end{align*}
\end{lemma}

\begin{proof}
From \eqref{prop:Finsaddleelliptica-} we deduce that
\begin{align} \nonumber
\operatorname{Re} F(a^{-1}) - \log\tau - \log|a|
&= 1+\log\tau + \xi+\xi'+ \frac{1}{2} e^{-2 \xi} \cos(2\eta) + \frac{1}{2} e^{-2\xi'} \cos(2\eta') - \log\tau - \log |a|\\ \label{eq:alphaEq1}
&= 1 + \frac{1}{2} e^{-2 \xi} \cos(2\eta) + \frac{1}{2} e^{-2\xi'} \cos(2\eta').
\end{align}
On the other hand, we have
\begin{align} \nonumber
z^2+z'^2
&= 2\cosh^2(\xi+i\eta)+2\cosh^2(\xi'+i\eta')\\ \nonumber
&= 2 + \cosh(2\xi+2i\eta) + \cosh(2\xi'+2i\eta')\\ \label{eq:alphaEq2}
&= 2 + \cosh(2\xi) \cos(2\eta) + \cosh(2\xi') \cos(2\eta')
+ i(\sinh(2\xi) \sin(2\eta) + \sinh(2\xi') \sin(2\eta')).
\end{align}
Combining \eqref{eq:alphaEq1} and \eqref{eq:alphaEq2}, we obtain
\begin{align} \nonumber
2\alpha &= \sinh(2\xi) \cos(2\eta) + \sinh(2\xi') \cos(2\eta')
+ i (\sinh(2\xi) \sin(2\eta) + \sinh(2\xi') \sin(2\eta'))\\ \label{eq:alphaEq3}
&= \sinh(2\xi) e^{2i\eta} + \sinh(2\xi') e^{2i\eta'}. 
\end{align}
Now let us verify the formula for $\beta$. We have
\begin{align*} 
2\overline a z z'
&= 2\overline a(\cosh(\xi+\xi'+i(\eta+\eta'))+\cosh(\xi-\xi'+i(\eta-\eta'))\\ 
&= e^{2(\xi+\xi')} + e^{-2i (\eta+\eta')}
+ e^{2\xi-2i\eta'} + e^{2\xi'-2i\eta},
\end{align*}
and similarly
\begin{align*} 
2 a^{-1} \overline{z z'} = e^{-2 i(\eta+\eta')} + e^{-2(\xi+\xi')} + e^{-2\xi'-2i\eta} +  e^{-2\xi-2i\eta'}.
\end{align*}
Hence, we have
\begin{align*}
a^{-1} \overline{z z'} - \overline a z z'
= - \sinh(2\xi+2\xi') - \sinh(2\xi) e^{-2i\eta'}
-\sinh(2\xi') e^{-2i\eta}.
\end{align*}
Combining this with \eqref{eq:alphaEq3},  we conclude that 
\begin{align*}
\frac{\overline a}{a} \beta &= 2\frac{\overline a}{a} \alpha + a^{-1} \overline{z z'} - \overline a z z'
= - \sinh(2\xi+2\xi'). 
\end{align*}
\end{proof}

The following lemma clarifies the relation between $h$ and $\operatorname{Re} F$. 

\begin{lemma} \label{lem:RFandhrelation}
For all $\zeta$ on the unit circle, we have
\begin{align}
\operatorname{Re} F(a^{-1} \zeta) = \operatorname{Re} F(a^{-1}) +  h(\zeta).
\end{align}
\end{lemma}

\begin{proof}
First, we define
\begin{align} \label{eq:defAuxg}
g(\zeta) = F\left(\tau;z,z';a^{-1}\zeta\right)+F\left(\tau;\overline z, \overline{z'}; \overline a^{-1}/\zeta\right).
\end{align}
Notice that $g(\zeta) = 2 \operatorname{Re} F(a^{-1}\zeta)$ for $\zeta$ on the unit circle. We can rewrite $g$ as
\begin{align*}
g(\zeta) &= \frac{(z^2+z'^2) \zeta^2 - 2 a z z' \zeta}{\zeta^2-a^2}
+ \frac{-(\overline{z^2+z'^2}) \overline a^{-2} + 2 \overline{z z' a^{-1}} \zeta}{\zeta^2-\overline a^{-2}} + 2\log|a| + 2\log\tau\\
&= \frac{q(\zeta)}{(\zeta^2-a^2)(\zeta^2-\overline a^{-2})} + 2\log|a| + 2 \log\tau,
\end{align*}
where some straightforward algebra shows that
\begin{align*}
q(\zeta) &= (\zeta^2-\overline a^{-2}) [(z^2+z'^2)\zeta^2 - 2 a z z' \zeta] + \overline a^{-2} (\zeta^2-a^2)[2 \overline{a z z'} \zeta - \overline{(z^2+z'^2)}]\\
&= (z^2+z'^2) \zeta^4 
- 2 a z z' \zeta^3
- \overline a^{-2} (z^2+z'^2) \zeta^2
+ 2 \overline a^{-2} a z z' \zeta\\
&\qquad + 2 \overline{a^{-1} z z'} \zeta^3
-  \overline a^{-2} \overline{(z^2+z'^2)} \zeta^2
- 2 a^2 \overline{a^{-1} z z'} \zeta
+ \overline a^{-2} a^2 \overline{(z^2+z'^2)}\\
&= (z^2+z'^2) \zeta^4 
+2(\overline{a^{-1} z z'} - a z z') \zeta^3
-2 \overline a^{-2} \operatorname{Re}(z^2+z'^2) \zeta^2\\
&\qquad + 2 (\overline a^{-2} a z z' - a^2 \overline{a^{-1} z z'}) \zeta
+ \overline a^{-2} a^2 \overline{(z^2+z'^2)}.
\end{align*}
Consequently, for $\gamma =  \operatorname{Re} F(a^{-1}) - \log\tau-\log|a|$ we have
\begin{multline*}
q(\zeta) - 2 \gamma  (\zeta^2-a^2) (\zeta^2-\overline a^{-2})\\
= 2\alpha \zeta^4
+2(\overline{a^{-1} z z'} - a z z') \zeta^3
+ 2\left((a^2+\overline a^{-2}) \gamma - \overline a^{-2} \operatorname{Re}(z^2+z'^2)\right) \zeta^2
+ 2 (\overline a^{-2} a z z' - a^2 \overline{a^{-1} z z'}) \zeta
+2\frac{a^2}{\overline a^2} \overline \alpha.
\end{multline*}
In other words, for $\zeta$ on the unit disc, we have (suppressing the first and second order term)
\begin{align} \label{eq:numeratorReFazetaa}
\operatorname{Re} F(a^{-1}\zeta) - \operatorname{Re} F(a^{-1})
= \frac{\alpha \zeta^4+(\overline{a^{-1} z z'} - a z z') \zeta^3
+\ldots+\frac{a^2}{\overline a^2} \overline \alpha}{(\zeta^2-a^2)(\zeta^2-\overline a^{-2})}.
\end{align}
That $s=a^{-1}$ is a saddle point of $F$, implies that $\zeta=1$ is a saddle point of $g$. This is because
\begin{align*}
i g'(1) = \left. \frac{d}{dt} g(e^{i t})\right|_{t=0} &= \left. i a^{-1} e^{i t} F'\left(\tau;z,z'; a^{-1} e^{i t}\right) - i \overline a^{-1} e^{- it} F'\left(\tau;\overline z,\overline{z'}; \overline a^{-1} e^{- i t}\right)\right|_{t=0}\\
&= i a^{-1} F'\left(\tau;z,z'; a^{-1}\right) + \overline{i a^{-1} F'\left(\tau;z,z'; a^{-1}\right)}
= 2\operatorname{Re}\left(i a^{-1} F'(a^{-1})\right) = 0.
\end{align*}
We conclude that the numerator in \eqref{eq:numeratorReFazetaa} must be divisible by $(\zeta-1)^2$, i.e.,
\begin{align*}
\alpha \zeta^4+(\overline{a^{-1} z z'} - a z z') \zeta^3
+\ldots+\frac{a^2}{\overline a^2} \overline \alpha
= (\zeta-1)^2 \left(\alpha \zeta^2 + \beta \zeta + \frac{a^2}{\overline a^2} \overline \alpha\right)
\end{align*}
for some constant $\beta$, that can easily be determined from the third order term and turns out to coincide with $\beta$ as introduced in Definition \ref{def:defh}. It follows that $g(\zeta)=g(1)+2 h(\zeta)$. On the unit circle, this gives us the formula that we are after.\\
\end{proof}

We infer from Lemma \ref{lem:RFandhrelation} that finding the maximum of $\operatorname{Re} F$ on the circle $|s|=|a|^{-1}$, is equivalent to finding the maximum of $h$ on the unit circle. In particular, looking at what we want to prove in Theorem \ref{lem:as=1}, this maximum should be $0$, and it should be attained in $\zeta=1$. A first step towards proving this is the following lemma. 

\begin{lemma} \label{prop:localmax1}
Let $z,z'\in \mathbb C\setminus [-\sqrt 2, \sqrt 2]$. Then $\zeta=1$ is a local maximum of the restriction of $h$ to the unit circle satisfying $h(\zeta)=0$, and it is unique with this property. 
\end{lemma}

\begin{proof}
From the explicit form of $h$ it is obvious that $h$ has a saddle point in $\zeta=1$. Thus we have 
\begin{align*}
\left. \frac{d}{dt} h\left(e^{i t}\right)\right|_{t=0} = i h'(1) = 0.
\end{align*}
Using the relation $g(\zeta)=g(1)+2 h(\zeta)$, with $g$ as defined in \eqref{eq:defAuxg}, we also have that
\begin{align*}
\left. 2\frac{d^2}{dt^2} h(e^{it})\right|_{t=0} = \left. \frac{d^2}{dt^2} g(e^{it})\right|_{t=0} &= - a^{-2} F''(\tau;z,z'; a^{-1}) - \overline a^{-2} F''(\tau;\overline z,\overline{z'}; \overline a^{-1})\\
&= -2 \operatorname{Re}\left(a^{-2} F''(a^{-1})\right)\\
&= -4 \operatorname{Re}\left(\frac{\sinh(\xi+i\eta) \sinh(\xi'+i\eta')}{\sinh(\xi+i\eta+\xi'+i\eta')}\right),
\end{align*}
where we used Proposition \ref{prop:F''a-1} to obtain the last line. For $\zeta=1$ to be a local maximum we need this expression to be negative.  This follows by using some identities for hyperbolic functions, namely
\begin{multline*}
\overline{\sinh(\xi+i\eta+\xi'+i\eta')} \sinh(\xi+i\eta) \sinh(\xi'+i\eta')\\
= |\sinh(\xi+i\eta)|^2 \overline{\cosh(\xi'+i\eta')} \sinh(\xi'+i\eta')
+ |\sinh(\xi'+i\eta')|^2 \overline{\cosh(\xi+i\eta)} \sinh(\xi+i\eta)\\
= \frac{1}{2} |\sinh(\xi+i\eta)|^2 \left(\sinh(2\xi') + i \sin(2\eta')\right)
+ \frac{1}{2} |\sinh(\xi'+i\eta')|^2 \left(\sinh(2\xi) + i \sin(2\eta)\right).
\end{multline*}
This will have positive real part in general, except when $\xi=\xi'=0$, when $\xi=0$ and $\eta\in \{0,\pi\}$, or when $\xi'=0$ and $\eta'\in \{0,\pi\}$. All these cases are excluded by the conditions of the lemma. 

At this point, we know that $\zeta=1$ is a local maximum on the unit circle with $h(\zeta)=0$. In order to reach a contradiction, we assume that there exists another point with this property. That is, we have $\tilde\zeta\neq 1$ such that the restriction of $h$ to the unit circle has a local maximum in $\zeta=\tilde\zeta$, and $h(\tilde\zeta)=0$. Then we have 
\begin{align*}
i\tilde\zeta h'\left(\tilde\zeta\right) = \left. \frac{d}{dt} h\left(e^{i t}\right)\right|_{t=\arg \tilde\zeta} = 0.
\end{align*}
This means that $\alpha \zeta^2+\beta \zeta+\frac{a^2}{\overline a^2} \overline \alpha$ must have only one root $\tilde\zeta$, with multiplicity $2$. Hence the discriminant is $0$, i.e., $\beta^2 - 4 \alpha \frac{a^2}{\overline a^2} \overline \alpha = 0$. 
We can rewrite this as
\begin{align*}
\left(\frac{\overline a}{a} \beta\right)^2 = 4 |\alpha|^2.  
\end{align*}
Using the formulas from Lemma \ref{lem:formulasalphabeta}, we can write this as
\begin{align*}
\sinh^2 (2(\xi+\xi'))
&= \left|\sinh(2\xi) e^{2i\eta} + \sinh(2\xi') e^{2i\eta'}\right|^2.
\end{align*}
We have
\begin{multline} \label{eq:sinhsform1}
\sinh^2(2\xi+2\xi') = \left(\sinh(2\xi)\cosh(2\xi')+\sinh(2\xi') \cosh(2\xi)\right)^2\\
= \sinh^2(2\xi)\cosh^2(2\xi')
+\sinh^2(2\xi')\cosh^2(2\xi)
+ 2 \sinh(2\xi) \sinh(2\xi') \cosh(2\xi) \cosh(2\xi'),
\end{multline}
and 
\begin{align} \nonumber
&\left|\sinh(2\xi) e^{2i\eta} + \sinh(2\xi') e^{2i\eta'}\right|^2\\ \nonumber
&= \left(\sinh(2\xi) \cos(2\eta) + \sinh(2\xi') \cos(2\eta')\right)^2 
+ \left(\sinh(2\xi) \sin(2\eta) + \sinh(2\xi') \sin(2\eta')\right)^2\\ \label{eq:sinhsform2}
&= \sinh^2 (2\xi) + \sinh^2 (2\xi') 
+ 2\sinh(2\xi) \sinh(2\xi') \cos(2(\eta-\eta')).
\end{align}
Subtracting \eqref{eq:sinhsform2} from \eqref{eq:sinhsform1}, we find
\begin{align*}
2 \sinh^2(2\xi) \sinh^2(2\xi') + 2 \sinh(2\xi) \sinh(2\xi')  \left(\cosh(2\xi) \cosh(2\xi') - \cos(2(\eta-\eta'))\right) = 0.
\end{align*}
Since $\cosh(2\xi) \cosh(2\xi') - \cos(2(\eta-\eta')\geq 0$, this equation can only be satisfied when $\xi=0$ or $\xi'=0$.\\ 
\end{proof}

Note that Lemma \ref{prop:localmax1} does not state that there are no other local maxima, but simply that they cannot satisfy $h(\zeta)=0$. However, by Lemma \ref{lem:as=1sameEllipse}, we know that $h(\zeta)<0$ for all $|\zeta|=1$ with $\zeta\neq 1$, in the case that $\xi=\xi'>0$. The idea now, is that we can continuously deform $h$ to cases where $\xi, \xi'>0$ with $\xi\neq \xi'$. By Lemma \ref{prop:localmax1}, during this process of deformation, it must remain true that $h(\zeta)<0$ for all $|\zeta|=1$ with $\zeta\neq 1$. Ignoring some caveats, this is the intuition behind what remains of the proof of Theorem \ref{lem:as=1}.\\

\noindent\textit{Proof of Theorem \ref{lem:as=1}.}
The case $z,z'\in [-\sqrt 2, \sqrt 2]$ corresponds to $\xi=\xi'=0$, and this case follows from Lemma \ref{lem:as=1sameEllipse}. We exclude it in what follows. 
Without loss of generality, we may assume that $z,z'\not\in [-\sqrt 2,\sqrt 2]$ in proving the inequality \eqref{eq:as=1}. The justification for this, is that the inequality should still hold in the limits $\xi\to 0$ with $\xi'>0$ fixed, and $\xi'\to 0$ with $\xi>0$ fixed. Now let us define the following continuous path on $\left(\mathbb C\setminus [-\sqrt 2, \sqrt 2]\right)^2$.
\begin{align*}
\Gamma(\lambda) &:= \sqrt 2 \left(\cosh\left(\frac{\xi+\xi'}{2}+\frac{\xi-\xi'}{2} \lambda + i\eta\right),\cosh\left(\frac{\xi+\xi'}{2}-\frac{\xi-\xi'}{2} \lambda+ i \eta'\right)\right), & \lambda\in [0,1].
\end{align*}
This path starts in a point where both components are on the same ellipse, and ends in $(z,z')$. Note that both components of $\Gamma$ have a finite distance to $[-\sqrt 2, \sqrt 2]$. The observation that $a$ is invariant under $\Gamma$ is somewhat helpful intuitively, but not crucial in what follows. Suppose that 
\begin{align*}
\lambda_0 = \inf\left\{\lambda\in[0,1] : \max_{|\zeta|=1} h(\tau;\Gamma(\lambda);\zeta) > 0\right\}
\end{align*}
exists. We will show that this leads to a contradiction. We may construct a decreasing sequence $(\lambda_k)_k$ in $[0,1]$ such that $\lambda_k\to \lambda_0$, and a corresponding sequence of global maxima $\zeta(\lambda_k)$ with $h(\tau;\Gamma(\lambda_k);\zeta(\lambda_k))>0$. Without loss of generality, we may assume that this sequence $\zeta(\lambda_k)$ converges to some point $\zeta(\lambda_0)$ on the unit circle. By continuity and given that $\lambda_0$ is an infimum, we must have $h(\tau;\Gamma(\lambda_0);\zeta(\lambda_0)) = 0$. 
Then Lemma \ref{prop:localmax1} forces $\zeta(\lambda_0)$ to equal $1$. We infer that
\begin{align*} 
h(\tau;\Gamma(\lambda_0);\zeta) 
= \lim_{k\to\infty} h(\tau;\Gamma(\lambda_k);\zeta)  
= \frac{(\zeta-1)^3 q_0(\zeta)}{(\zeta^2-a^2) (\zeta^2 - \overline a^{-2})}, 
\end{align*}
where $q_0(\zeta)$ is some polynomial of order $1$. This would imply that $h''(\tau;\Gamma(\lambda_0);1)=0$, which contradicts with our arguments in the first half of the proof of Lemma \ref{prop:localmax1}. We conclude that $\lambda_0$ does not exist, i.e., we have $h(\tau;\Gamma(\lambda);\zeta)\leq 0$ for all $|\zeta|=1$. This is actually a strict inequality when $\zeta\neq 1$ by Lemma \ref{prop:localmax1}. This proves the inequality \eqref{eq:as=1} for $z,z'\not \in [-\sqrt 2,\sqrt 2]$, and also that the maximum is only attained in $s=a^{-1}$. As stated before, by taking limits we may conclude that \eqref{eq:as=1} is valid for $(z,z')\not\in [-\sqrt 2, \sqrt 2]^2$. 

What remains is to find out when \eqref{eq:as=1} is an equality for some $\zeta \neq 1$. By the preceding, this can only happen if $\xi=0$ or $\xi'=0$. Without loss of generality, we take $\xi=0$ and $\xi'>0$. It follows from Proposition \ref{prop:Finsaddleelliptic} that $F(b)-F(a^{-1}) = -2i\eta+i \sin(2\eta)$ when this is the case. Since $b$ is a saddle point of $F$, it follows with similar arguments as before that $h$ then has a double zero in $\zeta = a b = e^{2i\eta}$. Some rewriting yields 
\begin{align*}
h(e^{it}) &= -16 \alpha e^{2i\eta} \overline a^2 \frac{\sin^2 \frac{t}{2} \sin^2 \frac{t-2\eta}{2}}{|e^{2it}-a^2|^2}, & t\in (-\pi,\pi].
\end{align*}
Indeed, using Lemma \ref{lem:formulasalphabeta}, $16 \alpha e^{2i\eta} \overline a^2 = 8 e^{2\xi'} \sinh(2\xi')>0$, and we clearly see that $\zeta=1$ and $\zeta=e^{2i\eta}$ are the only points where the global maximum is attained. 
\qed

\subsection{The large $n$ behavior of $I_n(d,\tau;z,z')$}

Having found a suitable contour to deform $\gamma_0$ to, we can now complete the method of steepest descent and find out the large $n$ behavior of $I_n(d,\tau;z,z')$. Let us start by mentioning that by deforming $\gamma_0$ in to $|s|=|a|^{-1}$, we  pick up a residue at $s=\tau$ if $|a|^{-1}> \tau$, and thus we can write:
\begin{equation}
    I_n(d,\tau;z,z') = \frac{e^{n F(\tau)}}{(1-\tau^2)^\frac{d}{2}} \mathfrak{1}_{|a|^{-1}>\tau} -\frac{1}{2\pi i} \oint_{\gamma_0} \frac{e^{n F(s)}}{s-\tau} \frac{ds}{(1-s^2)^\frac{d}{2}},
\end{equation}
where $\mathfrak{1}_{|a|^{-1}>\tau}$ denotes the indicator function of all $(z,z')$ such that $|a|^{-1}>\tau$. In the figures indicating the integration contours below, we depict the residue contribution by a small clockwise circle around $\tau$. Note that, when $a^{-1}\neq \tau$, but $|a|^{-1}=\tau$, we may slightly alter the integration path $|s|=|a|^{-1}$ such that $\tau$ is not enclosed. We shall tacitly do that when necessary.

As remarked before there are a few exceptions where we cannot deform $\gamma_0$ to $|s|=|a|^{-1}$, corresponding to the case $z, z'\in [-\sqrt 2, \sqrt 2]$. It can be divided in the case $z\neq \pm z'$ and the case(s) $z=\pm z'$. The corresponding deformations of $\gamma_0$ are depicted in Figure \ref{Fig17} and Figure \ref{Fig2a} respectively. 

 As a preparation, we start with the following inequality.
\begin{proposition} \label{thm:errorTermUniform1}
Let $d$ be a positive integer. For all $\tau\in (0,1]$, $z, z'\in\mathbb C$, and $n=1,2,\ldots$, we have
\begin{align} \label{thm:errorTermUniform1d}
\left|I_n(d,\tau;z,z') - \frac{e^{n F(\tau)}}{(1-\tau^2)^\frac{d}{2}} \mathfrak{1}_{|a|^{-1}>\tau}\right|
\leq  \frac{1}{(1-|a|^{-2})^\frac{d}{2}} \frac{e^{n \operatorname{Re} F(a^{-1})}}{\left|1- |a|\tau\right|}.
\end{align} 
Furthermore, for $d=1$ we have
\begin{align} \label{thm:errorTermUniform1d=1}
\left|I_n(1,\tau;z,z') - \frac{e^{n F(\tau)}}{\sqrt{1-\tau^2}} \mathfrak{1}_{|a|^{-1}>\tau}\right|
\leq  K \frac{e^{n \operatorname{Re} F(a^{-1})}}{\left|1- |a|\tau\right|},
\end{align}
where 
$K= \displaystyle\frac{1}{\pi} \int_0^\pi \frac{dt}{\sqrt{2 \sin t}}$.
\end{proposition}

\begin{proof}
When $|a|^{-1}=\tau$, there is nothing to prove, thus we exclude this case. For the same reason, we can exclude  $|a|^{-1}=1$ when proving \eqref{thm:errorTermUniform1d} and assume  that $|a|^{-1}<1$. Now we deform $\gamma_0$ to the circle $|s|=|a|^{-1}$. Then, using Theorem \ref{lem:as=1}, we get
\begin{align*}
\left|-\frac{1}{2\pi i} \oint_{\gamma_0} \frac{e^{n F(s)}}{(1-s^2)^\frac{d}{2}} \frac{ds}{s-\tau} - \frac{e^{n F(\tau)}}{(1-\tau^2)^\frac{d}{2}} \mathfrak{1}_{|a|^{-1}>\tau}\right|
&=  \frac{1}{2\pi} \left|\int_{-\pi}^\pi \frac{e^{n F(a^{-1} e^{it})}}{(1-a^{-2} e^{2it})^\frac{d}{2}} \frac{i a^{-1} e^{it}}{a^{-1} e^{it}-\tau} dt\right|\\
&\leq \frac{1}{2\pi} \frac{e^{n \operatorname{Re} F(a^{-1})}}{|1-|a|\tau|} \int_{-\pi}^\pi \frac{dt}{|1-a^{-2} e^{2it}|^\frac{d}{2}},
\end{align*} 
where we used that the residue at $s=\tau$ will contribute exactly when $|a|^{-1}>\tau$. After a trivial estimate of the remaining integral, we obtain \eqref{thm:errorTermUniform1d}. Let us focus on the case $d=1$ now. By Parseval's theorem
\begin{align*}
\int_{-\pi}^\pi \frac{dt}{|1-a^{-2} e^{2it}|^\frac{1}{2}}
&= \int_{-\pi}^\pi \left|\sum_{k=0}^\infty \binom{-1/4}{k} a^{-2k} e^{2i k t}\right|^2 dt
= 2\pi \sum_{k=0}^\infty \binom{-1/4}{k}^2 |a|^{-4k}. 
\end{align*}
This is clearly a decreasing function of $|a|$, and thus
\begin{align*}
\int_{-\pi}^\pi \frac{dt}{|1-a^{-2} e^{2it}|^\frac{1}{2}}
&\leq \int_{-\pi}^\pi \frac{dt}{|1-e^{2it}|^\frac{1}{2}}
= 2 \int_{0}^\pi \frac{dt}{\sqrt{2\sin t}}. 
\end{align*}
Note that for the estimates $d=1$ we can even take $|a|^{-1}=1$. 
\end{proof}

We emphasize that the inequalities \eqref{thm:errorTermUniform1d} and \eqref{thm:errorTermUniform1d=1} are not sharp. Indeed, a saddle point analysis, 
as we see in the next two theorems, will improve the inequality with a factor $ n^{-1/2}$ under further assumptions on the locations of $z,z'$. However, the proof of Proposition \ref{thm:errorTermUniform1} is short and straight to the point, it holds for all $z,z'$ and it already has some important consequences. For instance, it can be used to prove convergence for the local correlations in the bulk for the elliptic Ginibre ensemble and shows that the correction terms are exponentially small.

\begin{figure} 
\centering
\begin{overpic}[width=0.5\textwidth]{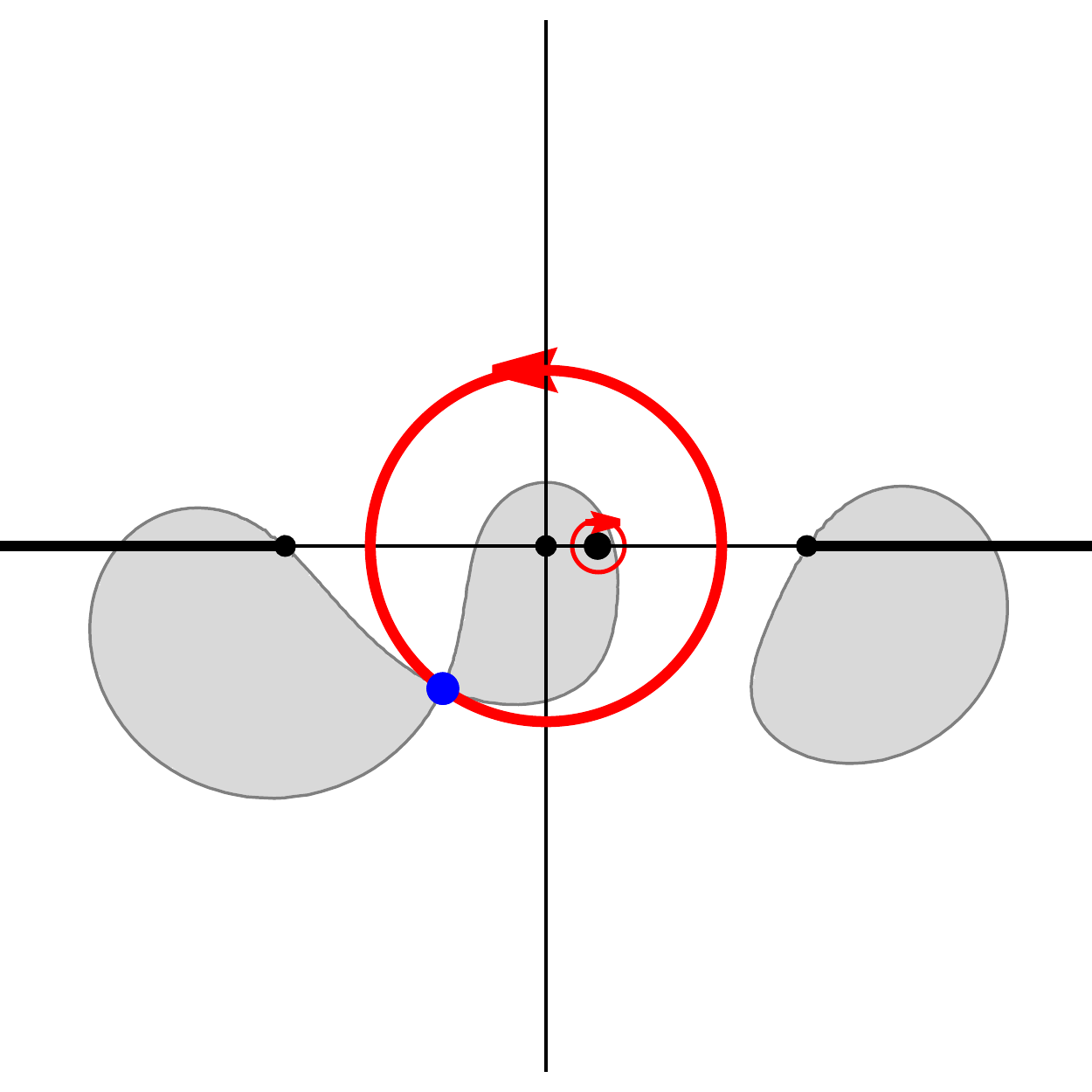}
    \put (22,45) {$-1$}
    \put (74,45) {$1$}
    \put (31,30) {$a^{-1}$}
    \put (53,44) {$\tau$}
    \put (46,46) {$0$}
\end{overpic}
\caption{When $z,z'\not\in [-\sqrt 2, \sqrt 2]$ and $z\neq\pm z'$ we deform the contour $\gamma_0$ to be the circle with radius $|a|^{-1}$ centered at the origin. The shaded areas represent the regions where $\operatorname{Re} F(s)\geq \operatorname{Re} F(a^{-1})$.
It is also possible that the shaded area containing $0$ touches the shaded area going through $1$ in $s=a^{-1}$. In some cases, there is a single shaded area, going through both $-1$ and $1$. The small circle around $s=\tau$ has to be added when deforming $\gamma_0$ to $|s|=|a|^{-1}$. 
  \label{Fig1}}
\end{figure}

We now proceed to perform a steepest descent analysis. It will be convenient to define the subset $S=S(z,z')$ of the saddle points by
\begin{align*}
S(z,z') = 
\begin{cases}
\{a^{-1}\}, & z\not\in [-\sqrt 2, \sqrt 2], z'\not\in [-\sqrt 2, \sqrt 2],\\
\{a^{-1},b^{-1}\}, & z\not\in [-\sqrt 2, \sqrt 2], z'\in [-\sqrt 2, \sqrt 2],\\
\{a^{-1}, b\}, & z\in [-\sqrt 2, \sqrt 2], z'\not\in [-\sqrt 2, \sqrt 2],\\
\{a^{-1}, a, b, b^{-1}\}, & z\in [-\sqrt 2, \sqrt 2], z'\in [-\sqrt 2, \sqrt 2].  
\end{cases}
\end{align*}
Generally, $S$ specifies which saddle point contributions have to be added, except in the special case that $z=\pm z'\in [-\sqrt 2, \sqrt 2]$. For the latter, we define $\Delta$ as the union of the diagonal and anti-diagonal on $[-\sqrt 2, \sqrt 2]^2$, that is
\begin{align*}
\Delta = \{(z,z) : z\in [-\sqrt 2, \sqrt 2]\} \cup \{(z,-z) : z\in [-\sqrt 2, \sqrt 2]\}.
\end{align*}
The case $(z,z')\in \Delta$ will be important, for instance when we consider one-point-correlation functions in later sections.\\

\begin{figure} 
    \centering
    \begin{overpic}[width=0.5\textwidth]{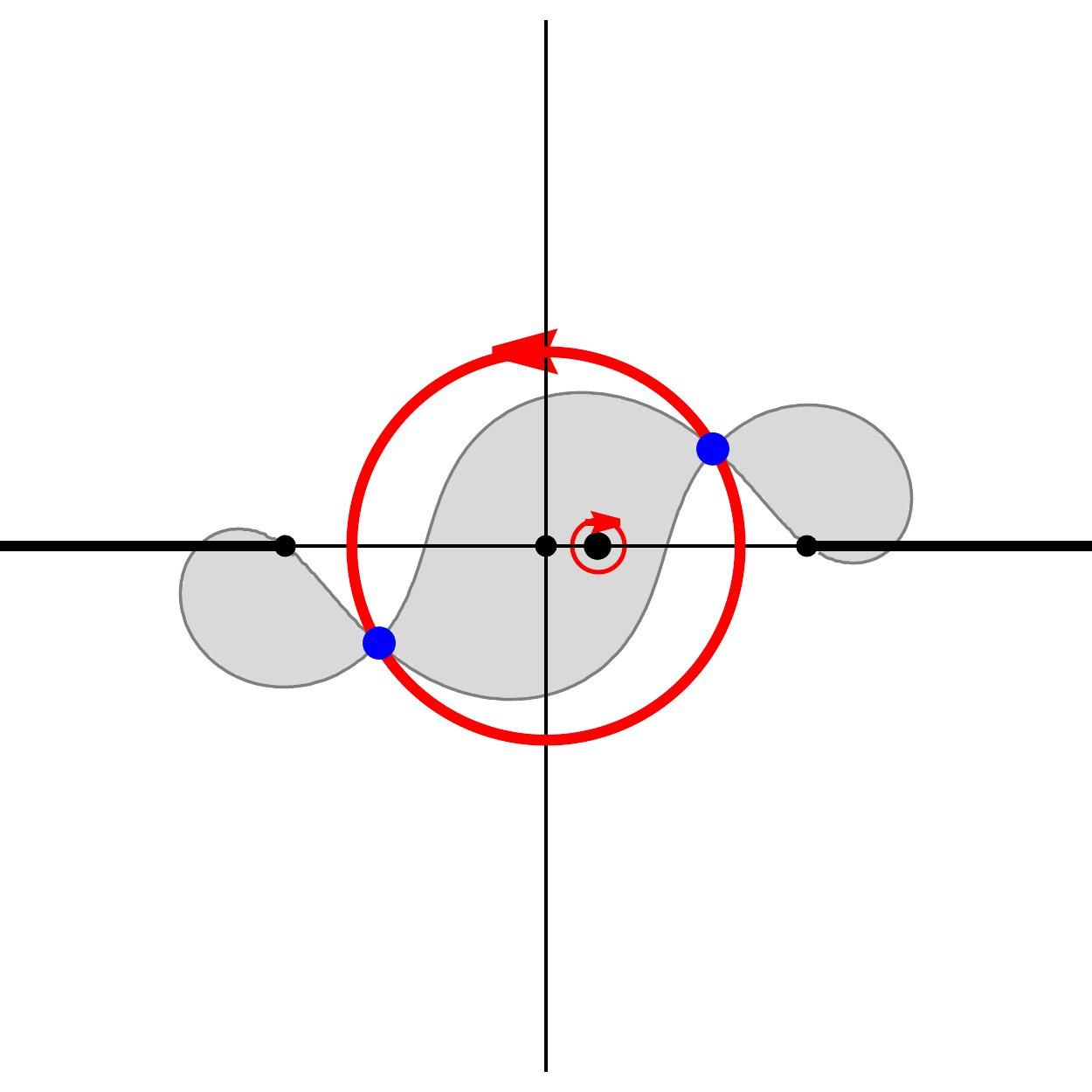}
        \put (22,45) {$-1$}
        \put (74,45) {$1$}
        \put (30,35) {$a^{-1}$}
        \put (65,62) {$b^{-1}$}
        \put (51,44) {$\tau$}
        \put (43,46) {$0$}
    \end{overpic}
    \caption{When $z\not\in [-\sqrt 2, \sqrt 2]$ but $z'\in [-\sqrt 2, \sqrt 2]$ (or vice versa but with $b^{-1}$ replaced by $b$),  we again deform the contour $\gamma_0$ to be the circle with radius $|a|^{-1}$ centered at the origin.   The shaded areas represent the regions where $\operatorname{Re} F(s)\geq \operatorname{Re} F(a^{-1})$. In this case, we have two saddle points on the circle $|s|=|a|^{-1}$. 
      \label{Fig1a}}
    \end{figure}

In both Theorem \ref{thm:InGeneral} and Theorem \ref{thm:InGeneralTau1}, we do not specify the signs of the square root in $\sqrt{\frac{2\pi}{-n F''(s)}}$. These should be taken in a way that respects the direction of the integration contour. 
In the proofs of our results in later sections, we shall specify the signs. 

\begin{theorem} \label{thm:InGeneral}
Let $\tau\in (0,1)$ be fixed and let $z,z'\in \mathbb C\setminus \{-\sqrt 2, \sqrt 2\}$.
\begin{itemize}
\item[(i)] If $(z,z')\not\in \Delta$ and $\tau\not\in S$, then we have 
\begin{align} \label{eq:generalSteepestDresult1}
I_n(d, \tau;z,z')
= \frac{e^{n F(\tau)}}{(1-\tau^2)^\frac{d}{2}} \mathfrak{1}_{|a|^{-1}>\tau} 
- \frac{1}{2\pi i} \sum_{s\in S} \sqrt{\frac{2\pi}{-n F''(s)}} \frac{e^{n F(s)}}{s-\tau} \frac{1}{(1-s^2)^\frac{d}{2}}
+ \mathcal O\left(\frac{e^{n \operatorname{Re} F(a^{-1})}}{n\sqrt n}\right),
\end{align}
as $n\to\infty$, where the $\mathcal O$ term is uniform on compact subsets.
\item[(ii)] If $(z,z')\in \Delta$, then we have 
\begin{align} \nonumber
I_n(d, \tau;z,z')
= \frac{e^{n F(\tau)}}{(1-\tau^2)^\frac{d}{2}} 
&- \frac{1}{2\pi i} \sum_{s\in S\setminus\{-1,1\}} \sqrt{\frac{2\pi}{-n F''(s)}} \frac{e^{n F(s)}}{s-\tau} \frac{1}{(1-s^2)^\frac{d}{2}}\\ \label{eq:generalSteepestDresult2}
&- \sum_{s\in S\cap \{-1,1\}} \frac{e^{n F(s)}}{2^{d-1} \Gamma\left(\frac{d}{2}\right)} \frac{|2 nF'(s)|^{\frac{d}{2}-1}}{s-\tau} 
+ \mathcal O\left((1+n^{\frac{d-1}{2}}) \frac{e^{n \operatorname{Re} F(a^{-1})}}{n\sqrt n}\right),
\end{align}
as $n\to\infty$, where the $\mathcal O$ term is uniform on compact subsets.
\end{itemize}
\end{theorem}


\begin{figure} 
\centering
\begin{overpic}[width=0.5\textwidth]{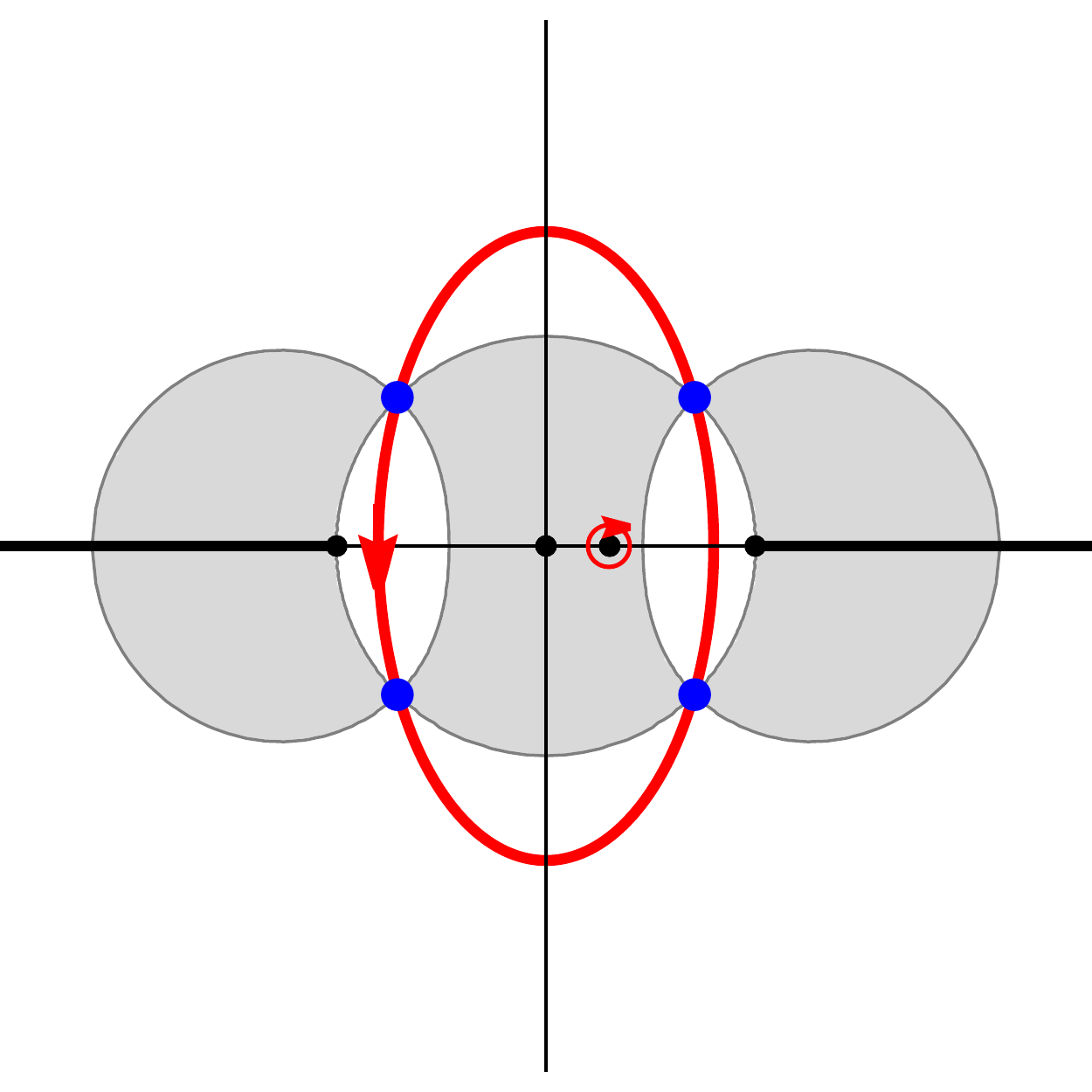}
    \put (35,30) {$b$}
    \put (32,68) {$b^{-1}$}
    \put (63,67) {$a$}
    \put (64,30) {$a^{-1}$}
    \put (25,45) {$-1$}
    \put (70,45) {$1$}
    \put (55,45) {$\tau$}
    \put (47,46) {$0$}
\end{overpic}
\caption{The integration contour when $z,z'\in [-\sqrt 2, \sqrt 2]$ and $z\neq\pm z'$.  In this case, the circle is no longer a good choice, but we can deform it slightly to be any bounded contour passing through the saddle points and avoiding the shaded regions (where, as before, $\operatorname{Re} F(s)\geq \operatorname{Re} F(a^{-1})$).
 \label{Fig17}}
\end{figure}

\begin{proof} Let us start with case (i). Here there are essentially two different subcases, each requiring their own deformation of the contour $\gamma_0$. We first consider the case that $(z,z')\in \mathbb C^2\setminus [-\sqrt 2, \sqrt 2]^2$. As explained above, we deform $\gamma_0$ to the circle $|s|=|a|^{-1}$. If $z,z'\not\in (-\sqrt 2, \sqrt 2)$, then $S(z,z')=\{a^{-1}\}$ and only the saddle point $a^{-1}$ is on $\gamma_0$. This case is illustrated in Figure \ref{Fig1}. If either $z\in (-\sqrt 2, \sqrt 2)$ or $z'\in (-\sqrt 2, \sqrt 2)$  (but not both) then also the saddle point $b$ or $b^{-1}$ respectively  is on  $\gamma_0$, illustrated in Figure \ref{Fig1a}. In all these cases, the main contribution in the asymptotic expansion comes from the saddle points. We cut the contour into several parts:  for each saddle point we consider the part that lies inside a small neighborhood of that saddle point, and then the remaining parts of the contour that stay at a postive distance to the saddle point(s). In the neighborhoods of the saddle point we  change to local variables in a standard way. For instance, the part of the contour close to $a^{-1}$ gives the term
\begin{align} \label{eq:saddleContributiana-}
-\frac{1}{2\pi i} \sqrt{\frac{2\pi}{-n F''(a^{-1})}} \frac{e^{n F(a^{-1})}}{a^{-1}-\tau} \frac{1}{(1-a^{-2})^\frac{d}{2}} \left(1+\mathcal O\left(\frac{1}{n}\right)\right).
\end{align}
If $b\in S(z,z')$ or $b^{-1}\in S(z,z')$ then we get additional terms with $a^{-1}$ replaced by $b$ or $b^{-1}$ respectively.  The integral over the remaining part of the contour will be exponentially smaller and can thus be absorbed into the $\mathcal O\left(\frac{1}{n}\right)$ term. Finally, since the residue at $\tau$ contributes exactly when $|a|^{-1}<\tau$, we arrive at \eqref{eq:generalSteepestDresult1} pointwise.  

We move to the case that $(z,z')\in (-\sqrt 2, \sqrt 2)^2\setminus \Delta$. Here we have saddle points $a = e^{i(\eta+\eta')}, a^{-1}=e^{-i(\eta+\eta')}, b=e^{i(\eta-\eta')}$ and $b^{-1}=e^{i(\eta'-\eta)}$. In this case, there is no descent for $\Re F$ on the contour $\gamma_0$ by Theorem \ref{lem:as=1} and this is not a contour that we can use. However, we can slightly deform it as shown in Figure \ref{Fig17} so that $\gamma_0$ does consist of contour of steep descent. Again, the main asymptotic contributions come from small neighborhoods of the saddle points.  This concludes case (i) of the theorem. 

Now we move to case (ii). Here we have saddle points $a = \pm e^{2i\eta}$ and $a^{-1} = \pm e^{-2i\eta}$. Let us first consider the case $z=z'$ with $z\neq 0$. Also in this case, the circle with radius $|a|^{-1}=1$ is not a steepest descent contour by Theorem \ref{lem:as=1}. We remedy this by further  deforming the  circle $\gamma_0$ to two  contours $\gamma_1$ and $\gamma_3$ as shown In Figure \ref{Fig2a}. The contour $\gamma_1$ goes through the saddle points $a^{-1}$ and $a$,  starts and end at infinity and lies fully in the region where $\Re F(s) < \Re F(a)=\Re F(a^{-1})$. The contour $\gamma_3$ goes around the cut $[1,\infty)$. 

The contour $\gamma_1$ picks up saddle point contributions at $a$ and $a^{-1}$ as before.  The contour $\gamma_3$ gives further contributions that we will now discuss. First of all, note that $\Re F(s)$ is dscreasing for $s>1$. Hence the dominant contribution in the integral over $\gamma_3$ comes from a small neighborhood around $s=1$.  Indeed, under the substitution $s = 1+\frac{t}{n}$ we have
\begin{align*}
\frac{e^{n F(s)}}{s-\tau}\frac{1}{(1-s^2)^\frac{d}{2}} 
= e^{n F(1)+ F'(1) t} n^\frac{d}{2} (-2 t)^{-\frac{d}{2}} \frac{1}{1-\tau} \left(1+\mathcal O\left(\frac{t+t^2}{n}\right)\right).
\end{align*}
Since the tail goes to $0$ fast enough, we infer from this that
\begin{align} \label{eq:gamma3band}
-\frac{1}{2\pi i} \int_{\gamma_3} \frac{e^{n F(s)}}{s-\tau} \frac{1}{(1-s^2)^\frac{d}{2}} ds 
= -\frac{1}{2\pi i} 2^{-\frac{d}{2}} n^{\frac{d}{2}-1} \frac{1}{1-\tau} e^{n F(1)} \int_{-1+\gamma_3} e^{F'(1) t} (-t)^{-\frac{d}{2}} dt \left(1+\mathcal O\left(\frac{1}{n}\right)\right).
\end{align}
When $\operatorname{Re} d<2$ we may take the bandwidth of $\gamma_3$ to $0$, and then we have
\begin{align} \nonumber
\int_{-1+\gamma_3} e^{F'(1) t} (-t)^{-\frac{d}{2}} dt
&= 2i\sin\left(\frac{d \pi}{2}\right) \int_0^\infty e^{F'(1) t} t^{-\frac{d}{2}} dt\\ \label{eq:bendGamma}
&= 2i\sin\left(\frac{d \pi}{2}\right) \Gamma\left(1-\frac{d}{2}\right) |F'(1)|^{\frac{d}{2}-1}
= \frac{2\pi i}{\Gamma\left(\frac{d}{2}\right)} |F'(1)|^{\frac{d}{2}-1}.
\end{align}
Here we used that $F'(1) = \frac{1}{2} z^2 - 1<0$. The identity in \eqref{eq:bendGamma} is actually true for all $d\in\mathbb C$ by analytic continuation, and we are done after plugging it into \eqref{eq:gamma3band}. For the error term, we may again use Theorem \ref{lem:as=1}. 

In the case that $z=-z'$ the picture is rather symmetric, here we use a contour that bends around $(-\infty, -1]$ instead of $[1,\infty)$ (see Figure \ref{Fig2a}). For $z=z'=0$, one has to use both a contour that bends around $[1,\infty)$, and a contour that bends around $(-\infty,-1]$. 

For the uniformity of the $\mathcal O$ term, we remark that for compact subsets in case (i), our deformations stay a finite distance away from the singularities $-1, 0$ and $1$. This, in combination with the analyticity of $F$, essentially gives the uniformity of the $\mathcal O$ term. For compact subsets in case (ii) the argument is a little more subtle, because the contours are not bounded, but an investigation of the tails of the corresponding integrals will give us the uniformity of the $\mathcal O$ term. These arguments can be made more precise, but to avoid being overly technical, we choose not to present this.  
\end{proof} 

\begin{figure}
\centering
\begin{overpic}[width=0.5\textwidth]{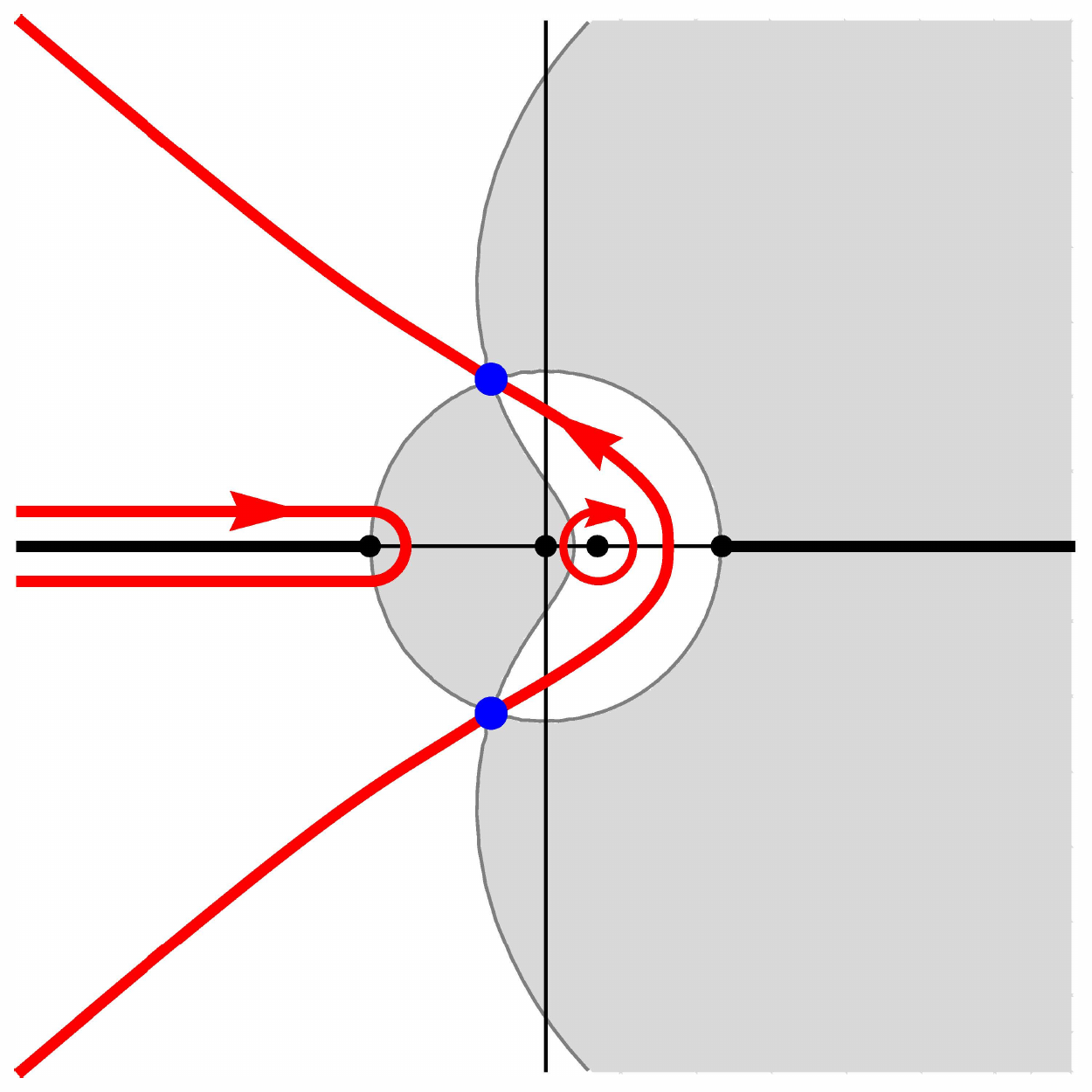}

\put (45,68) {$a^{-1}$}

\put (45,30) {$a$}

\put (28,47) {$-1$}

\put (63,45) {$1$}

\put (54,44) {$\tau$}

\put (46,46) {$0$}

\put (35,55) {$\gamma_3$}

\put (55,55) {$\gamma_2$}

\put (30,80) {$\gamma_1$}

\end{overpic}\begin{overpic}[width=0.5\textwidth]{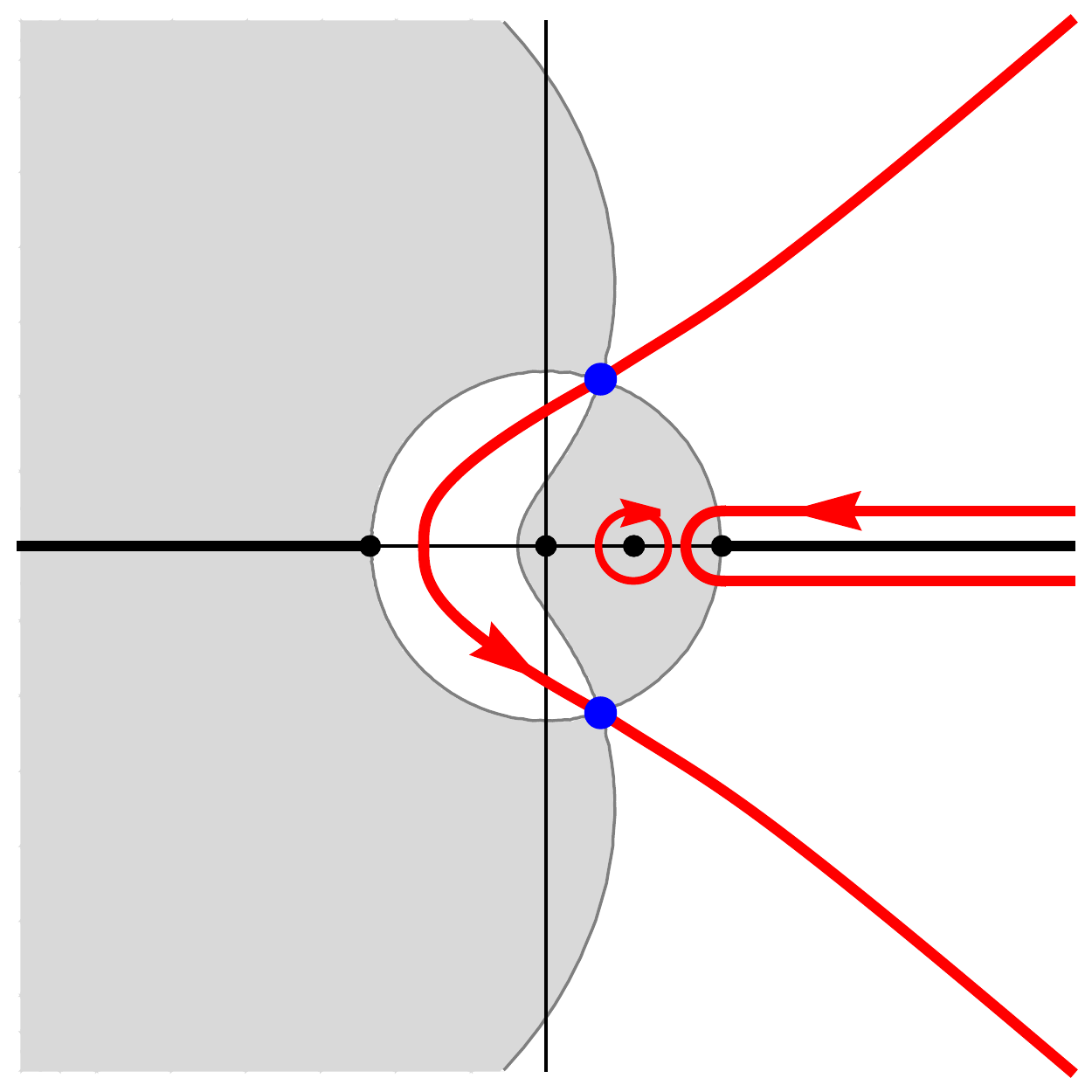}
    \put (51,68) {$a^{-1}$}
    \put (50,30) {$a$}
    \put (28,45) {$-1$}
    \put (63,45) {$1$}
    \put (54,45) {$\tau$} 
    \put (46,46) {$0$}
    \put (54,55) {$\gamma_2$}
    \put (65,55) {$\gamma_3$}
    \put (70,80) {$\gamma_1$}
\end{overpic}

\caption{Left: the integration contour for $z\in (-\sqrt 2, \sqrt 2)$ and $z = -z'$. Right: the integration contour for $z\in (-\sqrt 2, \sqrt 2)$ and $z = z'$. The shaded areas represent the regions where $\operatorname{Re} F(s)\geq \operatorname{Re} F(a^{-1})$, they are bounded. The contours passing through $a$ and $a^{-1}$ are the steepest descent contours. However, it suffices to have any contour (beginning and ending at $\infty$) that passes through $a$ and $a^{-1}$ and avoids the shaded areas. \label{Fig2a}} 
\end{figure}

\begin{theorem} \label{thm:InGeneralTau1}
Let $\tau=1$ and let $z,z'\in \mathbb C\setminus \{-\sqrt 2, \sqrt 2\}$.
\begin{itemize}
\item[(i)] If $(z,z')\not\in \Delta$ and $\tau\not\in S$, then we have 
\begin{align} \label{eq:generalSteepestDresult1Tau1}
I_n(d, \tau;z,z')
= - \frac{1}{2\pi i} \sum_{s\in S} \sqrt{\frac{2\pi}{-n F''(s)}} \frac{e^{n F(s)}}{s-1} \frac{1}{(1-s^2)^\frac{d}{2}} 
+ \mathcal O\left(\frac{e^{n \operatorname{Re} F(a^{-1})}}{n\sqrt n}\right),
\end{align}
as $n\to\infty$, where the $\mathcal O$ term is uniform on compact subsets.
\item[(ii)] If $(z,z')\in \Delta$, then we have 
\begin{align} \label{eq:generalSteepestDresult2Tau1}
I_n(d, \tau;z,z')
= - \frac{1}{2\pi i} \sum_{s\in S\setminus\{-1,1\}} \sqrt{\frac{2\pi}{-n F''(s)}} \frac{e^{n F(s)}}{s-1} \frac{1}{(1-s^2)^\frac{d}{2}} 
+  \frac{e^{n F(1)}}{2^{d} \Gamma\left(\frac{d}{2}+1\right)} |2 nF'(1)|^{\frac{d}{2}} 
+  \mathcal O\left((1+n^{\frac{d+1}{2}}) \frac{e^{n \operatorname{Re} F(a^{-1})}}{n\sqrt n}\right),
\end{align}
as $n\to\infty$, where the $\mathcal O$ term is uniform on compact subsets. 
\end{itemize}
\end{theorem}

\begin{proof}
The proof is almost entirely analogous to the proof of Theorem \ref{thm:InGeneral}, with only a few differences. Firstly, the residue at $s=\tau=1$ will never contribute. Secondly, in case (ii) the integral that bends around $[1,\infty)$ now has to be done with parameter $\frac{d}{2}+1$ instead of $\frac{d}{2}$ in \eqref{eq:bendGamma}. 
\end{proof}

In Theorem \ref{thm:InGeneral} and Theorem \ref{thm:InGeneralTau1} we have avoided the cases where $z$ or $z'$ equal $\pm \sqrt 2$. These cases can be done in principle, but they are different in that the saddle points are degenerate of order $2$. Then there are three rather than two steepest descent directions, and we have to pick two of them for the deformation of our integration contour. See also Figure \ref{fig:order2}. In general, rather than being of order $n^{-\frac{1}{2}} e^{n F(a^{-1})}$, the saddle point contribution will be of order $n^{-\frac{1}{3}} e^{n F(a^{-1})}$. An exception is when both $z$ and $z'$ equal $\pm \sqrt 2$, then the saddle point is in $s=\pm 1$, and the saddle point contribution will be of order $n^{\frac{d-2}{6}} e^{n F(1)}$, due to the factor $(1-s^2)^\frac{d}{2}$ in the denominator of the integrand. This special case is treated explicitly in Theorem \ref{thm:onepointAsymp}(iii) and Theorem \ref{thm:onepointAsympd}(iii).

\begin{figure} 
    \centering
    \begin{overpic}[width=0.5\textwidth]{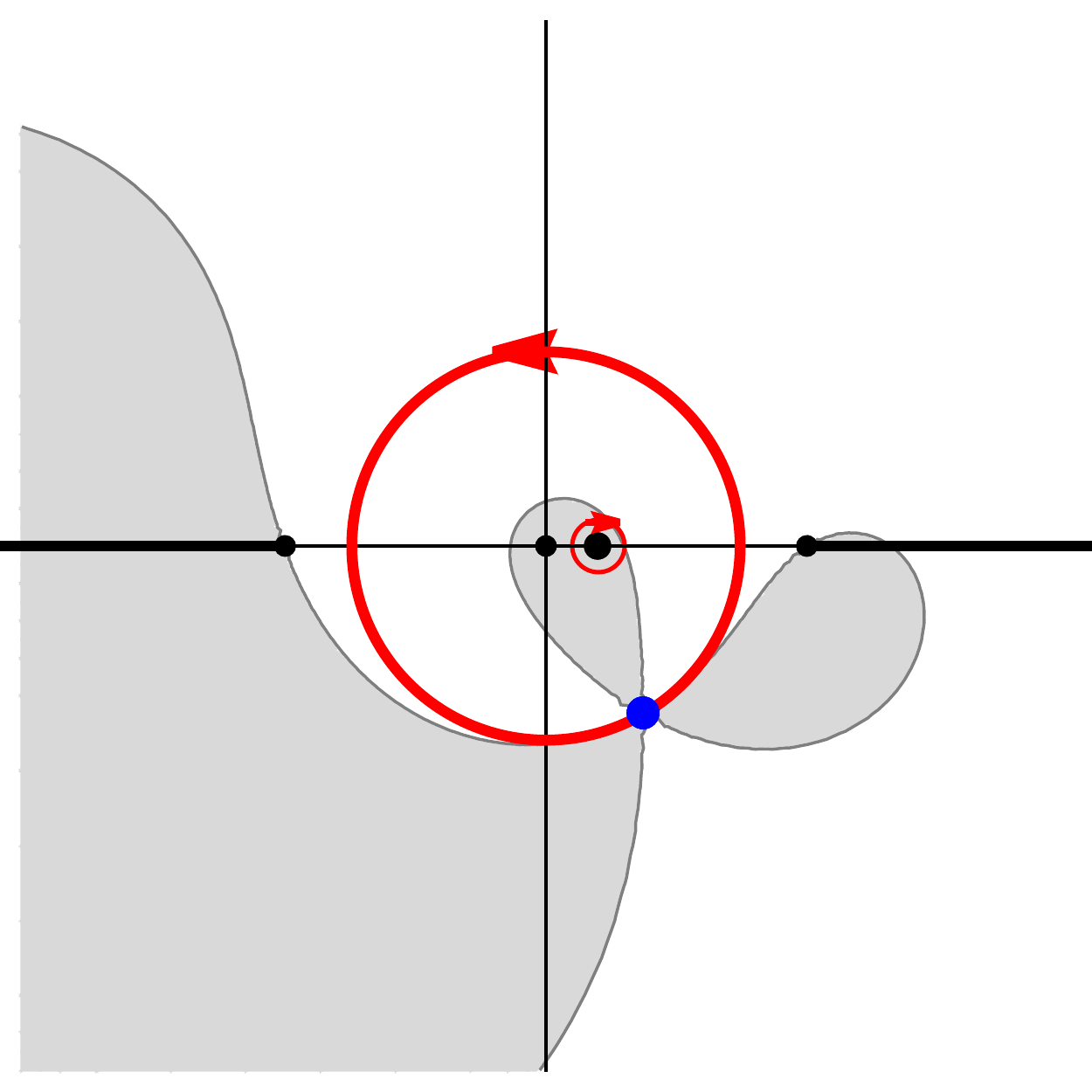}
       
        \put (60,28) {$a^{-1}$}
        \put (22,45) {$-1$}
        \put (75,45) {$1$}
        \put (54,45) {$\tau$}
        \put (48,46) {$0$}
    \end{overpic}
   \caption{The integration contour for $z\in \{-\sqrt 2, \sqrt 2\}$ and $z'\not \in \{-\sqrt 2, \sqrt 2\}$.  The shaded areas represent the regions where $\operatorname{Re} F(s)\geq \operatorname{Re} F(a^{-1})$. The three shaded areas are bounded. In this case, the saddle point $a^{-1}$ is of order 2.}
   \label{fig:order2}
    \end{figure}

\section{The elliptic Ginibre ensemble ($d=1$)} \label{sec:EGE}

In this section we let $\tau\in(0,1)$ be fixed. We recall that for $d=1$  the correlation kernel in \eqref{eq:general_kernel_d=1} can be expressed as 
\begin{align} \label{eq:kernelHermitePols}
\mathbb K_n(Z,Z') = \frac{n}{\pi\sqrt{1-\tau^2}} \sqrt{\omega\left(\sqrt n Z\right) \omega\left(\sqrt n Z'\right)} \sum_{j=0}^{n-1} \frac{1}{j!} \left(\frac{\tau}{2}\right)^j H_j\left(\sqrt{\frac{n}{2 \tau}} Z\right) H_j\left(\sqrt{\frac{n}{2 \tau}} \overline{Z'}\right),
\end{align}
where the weights are given by
\begin{align} \label{weight2}
\omega(Z) &= \exp\left(-\frac{1}{1-\tau^2} \left(|Z|^2 - \frac{\tau}{2} (Z^2+\overline{Z}^2)\right)\right)
= \exp\left(-\frac{\operatorname{Re}(Z)^2}{1+\tau}\right) \exp\left(-\frac{\operatorname{Im}(Z)^2}{1-\tau}\right).
\end{align}
At various places we will express $Z$ and $Z'$ in elliptic coordinates, that is
\begin{align*}
Z = 2\sqrt \tau\cosh(\xi+i\eta) \qquad\text{ and }\qquad Z' = 2\sqrt \tau\cosh(\xi'+i\eta'),
\end{align*}
where $\xi\geq 0$ and $\eta\in (-\pi,\pi]$ when $\xi>0$, $\eta\in [0,\pi]$ when $\xi=0$, and similar for $\xi'$. We identify
\begin{align} \label{eq:defzzprimEGE}
z = \frac{Z}{\sqrt{2\tau}}=\sqrt 2\cosh(\xi+i\eta) \qquad\text{ and }\qquad z' = \frac{\overline{Z'}}{\sqrt{2\tau}}=\sqrt 2\cosh(\xi'-i\eta'),
\end{align}
and then we can write the kernel, using \eqref{eq:integral_representation_kernel}, as
\begin{align} \label{eq:EGEInwithWeights}
\mathbb K_n(Z,Z') = \frac{n}{\pi\sqrt{1-\tau^2}} \sqrt{\omega\left(\sqrt n Z\right) \omega\left(\sqrt n Z'\right)} I_n(1, \tau;z,z'),
\end{align}
with $I_n(1, \tau;z,z')$ as in \eqref{eq:defIn}. Note that the expression for $z'$ in \eqref{eq:defzzprimEGE} has a minus sign in front of $\eta'$ (due to the conjugation). This will be important when using the results from Section \ref{sec:steepest}, where we have worked with \eqref{eq:defEllipticCoord}.

\subsection{Two preliminary lemmas}

Before we come to the proofs of our main results for the elliptic Ginibre ensemble, we first mention two lemmas that we will need. 

To formulate the first lemma we define
\begin{align*}
K_{\rho}^{\infty}(Z,Z') &= \frac{\rho}{\pi} \exp\left(-\rho \frac{|Z|^2+|Z'|^2-2 Z \overline{Z'}}{2}\right)
\exp\left(- i \sqrt{\rho(\rho-1)} \frac{\operatorname{Im}(Z^2-Z'^2)}{2}\right),
\end{align*}
which is a kernel for the $\infty$-Ginibre process with density $\rho>0$. The exponential on the far right is usually not included, but it has no real relevance, since, being a cocycle, it does not change the correlation functions.
\begin{lemma} \label{prop:valueomegaWZeFtau}
For all $n\geq 1$ we have
\begin{align} \label{eq:valueomegaWZeFtau}
\frac{n}{\pi(1-\tau^2)}\sqrt{\omega\left(\sqrt n Z\right) \omega\left(\sqrt n Z'\right)} e^{n F(\tau)}
=   n K_{(1-\tau^2)^{-1}}^{\infty}\left(\sqrt n Z, \sqrt n Z'\right).
\end{align}
In particular, we have
\begin{align} \label{eq:weightsFtauGaussian}
\left|\sqrt{\omega\left(\sqrt n Z\right) \omega\left(\sqrt n Z'\right)} e^{n F(\tau)}\right| = \exp\left(-\frac{n}{2} \frac{|Z-Z'|^2}{1-\tau^2}\right).
\end{align}
\end{lemma}

\begin{proof}
The result follows by a straightforward calculation, namely
\begin{align*}
\frac{1}{2 n} & \log (\omega\left(\sqrt n Z\right) \omega\left(\sqrt n Z'\right)) + F(\tau)\\
&= -\frac{1}{2(1-\tau^2)} \left(|Z|^2 - \frac{\tau}{2} (Z^2+\overline{Z}^2)\right)
 -\frac{1}{2(1-\tau^2)} \left(|Z'|^2 - \frac{\tau}{2} (Z'^2+\overline{Z'}^2)\right)
 - \frac{\frac{\tau}{2} (Z^2+Z'^2) -  
 ZZ'}{1-\tau^2}\\
&= -\frac{1}{2(1-\tau^2)} \left(|Z|^2+|Z'|^2 - \frac{\tau}{2} (Z^2+\overline{Z}^2)- \frac{\tau}{2} (Z'^2+\overline{Z'}^2)+\tau (Z^2+\overline{Z'}^2) -  2 Z \overline{Z'}\right)\\
 &= -\frac{|Z|^2+|Z'|^2-2 Z \overline{Z'}}{2(1-\tau^2)}+\frac{\tau (\overline{Z}^2+Z'^2 - Z^2 - \overline{Z'}^2)}{4(1-\tau^2)}.
\end{align*}
\end{proof}
For the second lemma, we need a definition.
\begin{definition}
We define $\xi_\tau = -\frac{1}{2} \log\tau$, and we define the function $g(\zeta) = g(\tau;z,z';\zeta)$ via
\begin{align} \label{eq:defpthetas}
-(\xi - \xi_\tau)^2 g(\xi+i\eta) =  \frac{1}{2} +\xi -\xi_\tau + \frac{1}{2} e^{-2\xi} \cos(2\eta)
- 2\tau \frac{\cosh^2 \xi \cos^2 \eta}{1+\tau} 
- 2\tau \frac{\sinh^2 \xi \sin^2 \eta}{1-\tau}.
\end{align}
for $\xi\geq 0$ with $\xi\neq \xi_\tau$, and $g\left(\xi_\tau+i\eta\right) 
= \frac{1+\tau^2- 2\tau \cos(2\eta)}{1-\tau^2}$.\\
\end{definition}


\begin{lemma} \label{prop:weightsFa-1}
The function $g$ is continuous and it is bounded from below by a positive number.
We have
\begin{align} \label{eq:weightsFa-1}
\left|\sqrt{\omega\left(\sqrt n Z\right) \omega\left(\sqrt n Z'\right)} e^{n F(a^{-1})}\right|
= e^{- n (\xi-\xi_\tau)^2 g(\xi+i\eta)} e^{ - n (\xi'-\xi_\tau)^2 g(\xi'+i\eta')}. 
\end{align}
\end{lemma}

\begin{proof}
Plugging in the expressions for $Z$ and $Z'$ in elliptic coordinates, we have
\begin{align*}
\frac{1}{2 n} \log(\omega\left(\sqrt n Z\right) \omega\left(\sqrt n Z'\right)) 
= - 2\tau \frac{\cosh^2 \xi \cos^2 \eta + \cosh^2 \xi' \cos^2 \eta'}{1+\tau} 
- 2\tau \frac{\sinh^2 \xi \sin^2 \eta+\sinh^2 \xi' \sin^2 \eta'}{1-\tau}.
\end{align*}
Combining this with \eqref{prop:Finsaddleelliptica-}, we obtain \eqref{eq:weightsFa-1} for $\xi\neq \xi_\tau$. 
Now we define $G_\eta(\xi)$ by the right-hand side of \eqref{eq:defpthetas}, that is,
\begin{align*}
G_\eta(\xi) &= \frac{1}{2} + \xi-\xi_\tau + \frac{1}{2} e^{-2\xi} \cos(2\eta)
- 2\tau \frac{\cosh^2 \xi \cos^2 \eta}{1+\tau} 
- 2\tau \frac{\sinh^2 \xi \sin^2 \eta}{1-\tau}.
\end{align*}
We see that
\begin{align*}
G_\eta(\xi_\tau) &= 
\frac{1}{2} + \frac{\tau}{2} \cos(2\eta) - \frac{1+\tau}{2} \cos^2 \eta - \frac{1-\tau}{2} \sin^2 \eta = 0.
\end{align*}
By differentiation we get
\begin{align*}
G_\eta'(\xi) 
&= 1 - e^{-2\xi} \cos(2\eta) - \frac{2\tau}{1-\tau^2} \sinh(2\xi) (1-\tau \cos(2\eta)).
\end{align*}
We infer from this that 
\begin{align*}
G_\eta'(\xi_\tau)  &= 1 - \tau \cos(2\eta) - (1-\tau \cos(2\eta)) = 0.
\end{align*}
and, differentiating again, that
\begin{align*}
G_\eta''(\xi_\tau) &= 
2\tau \cos(2\eta) - \frac{4\tau}{1-\tau^2} \frac{\tau^{-1}+\tau}{2} (1- \tau \cos(2\eta))
= - 2 \frac{1+\tau^2- 2\tau \cos(2\eta)}{1-\tau^2} < 0.
\end{align*}
We conclude that there is a neighborhood around $\xi = \xi_\tau$ on which $G_\eta$ is negative, except in $\xi=\xi_\tau$ itself. In particular, we have
\begin{align*}
G_\eta(\xi) = - (\xi - \xi_\tau)^2 g(\xi+i\eta),
\end{align*}
for some continuous function $g$ which is positive in a neighborhood of $\xi=\xi_\tau$, and in particular in $\xi=\xi_\tau$ itself. Of course, this is the function $g$ as defined above. Our goal is to show that $g$ is in fact uniformly bounded from below by a positive constant.  
The equation $G_\eta'(\xi)=0$ can be rewritten as a quadratic equation in $e^{2 \xi}$. We already know that $e^{2\xi}=e^{2\xi_\tau} = \tau^{-1}$ is a solution. We want to show that the other solution, let us call it $\lambda$, is smaller than or equal to $1$. We should have that
\begin{multline*}
-\frac{\tau}{1-\tau^2} (1-\tau \cos(2\eta)) \left(e^{2\xi} - \frac{1}{\tau}\right)\left(e^{2\xi}-\lambda\right)
= -\frac{\tau}{1-\tau^2} (1-\tau \cos(2\eta)) e^{4 \xi} + e^{2\xi}- \cos(2\eta) + \frac{\tau}{1-\tau^2} (1-\tau \cos(2\eta)).
\end{multline*}
Thus, comparing coefficients of $e^{2\xi}$ on both sides, we have
\begin{align*} 
\lambda = \frac{\cos(2\eta)-\tau}{1-\tau \cos(2\eta)}.
\end{align*}
This has the highest value when $\cos(2\eta)=1$, thus $\lambda \leq \frac{1-\tau}{1-\tau} = 1$. We conclude that $G_\tau'(\xi)$ is strictly positive for $\xi\in [0,\xi_\tau)$ and strictly negative for $\xi\in (\xi_\tau,\infty)$. This implies that $G_\eta(\xi)$ is negative for all $\xi\neq\xi_\tau$. Hence $g$ is strictly positive on $\mathbb C$. Furthermore, we have $g(\xi+i\eta)\to \infty$ as $\xi\to\infty$, uniformly for all $\eta$. Hence there exists a $\tilde\xi>0$ such that $g(\xi+i\eta)\geq 1$ whenever $\xi>\tilde\xi$. Since $g$ is continuous and strictly positive on the compact (elliptic) region $\xi\leq \tilde\xi$, it must be bounded from below by a positive constant there as well. 
\end{proof}

\begin{remark} \label{remark:ReFtauReFa-1}
A consequence of Lemma \ref{prop:valueomegaWZeFtau} and Lemma \ref{prop:weightsFa-1} is that
\begin{align} \label{eq:RtauRa-Compare}
\operatorname{Re} F(\tau) - \operatorname{Re} F(a^{-1}) 
= -\frac{1}{2} \frac{|Z-Z'|^2}{1-\tau^2} + (\xi-\xi_\tau)^2 g(\xi+i\eta) + (\xi'-\xi_\tau)^2 g(\xi'+i\eta').
\end{align}
We deduce from this that the residue contribution 
in $I_n$ at $s=\tau$  is dominant when $Z$ and $Z'$ are inside  $\mathcal E_\tau$ and close to each other. The residue thus gives the ``bulk'' behavior. On the other hand, when both $Z$ and $Z'$ are on or near the boundary of $\mathcal E_\tau$, then the saddle point contribution at $s=a^{-1}$ is dominant (if $Z$ and $Z'$ are not too close to each other). The ``edge'' behavior is thus given by the saddle point contribution. 
\end{remark}

\subsection{Proofs of Theorems \ref{thm:largenKnZW_intro} and \ref{thm:2ptCluster_intro}}
We now prove our main results for the elliptic Ginibre ensemble. Starting with the

\begin{proof}[Proof of Theorem \ref{thm:largenKnZW_intro}]
The set under consideration falls under case (i) of Theorem \ref{thm:InGeneral}. Hence we have
\begin{align*}
I_n(1, \tau;z,z')
= \frac{e^{n F(\tau)}}{(1-\tau^2)^\frac{1}{2}} \mathfrak{1}_{|a|^{-1}>\tau} 
- \frac{1}{2\pi i} \sqrt{\frac{2\pi}{-n F''(a^{-1})}} \frac{e^{n F(a^{-1})}}{(1-a^{-2})^\frac{1}{2}} \frac{1}{a^{-1}-\tau}
+ \mathcal O\left(\frac{e^{n \operatorname{Re} F(a^{-1})}}{n\sqrt n}\right),
\end{align*}
where the square root of $F''(a^{-1})$ is chosen in such a way that the direction of the integration contour is respected. To find out how to pick the square root, it is somewhat easier to do a substitution $s\to a^{-1} e^{it}$ and look at $f(t):= F(a^{-1} e^{it})$ instead, for $t\neq \eta-\eta'-i(\xi+\xi') \mod \pi$. Then we consider an integral from $-\pi$ to $\pi$ (in that direction), and there is a saddle point in $t=0$. In this setup, it is a well-known fact that the leading order asymptotic behavior of the integral is given by
\begin{align*}
-\frac{1}{2\pi i} \sqrt{\frac{2\pi}{n}} e^{n F(a^{-1})} \frac{1}{\sqrt{1-a^{-2}}} \frac{i a^{-1}}{a^{-1}-\tau} \frac{1}{(-f''(0))^{1/2}},
\end{align*} 
where $(-f''(0))^{1/2}$ has to be picked such that the argument is in $(-\frac{\pi}{4}, \frac{\pi}{4})$. However, it is clear from the proof of Lemma \ref{prop:localmax1} that $-f''(0)$ is in the right half-plane. This means that $(-f''(0))^{1/2} = \sqrt{-f''(0)}$, where, by convention, the square root is taken with a cut $(-\infty,0]$ and it is positive for positive values. Furthermore, $1-a^{-2}$ is also in the right half-plane. This implies that 

\begin{align*}
\sqrt{1-a^{-2}} \sqrt{-f''(0)} = \sqrt{-(1-a^{-2}) f''(0)} = 2 \sqrt{a^{-1} \sinh(\xi+i\eta) \sinh(\xi'-i\eta')},
\end{align*}
where we used Proposition \ref{prop:F''a-1} in the last line. We note that 
\begin{align} \label{eq:pmSqrtSaddle}\sqrt{a^{-1} \sinh(\xi+i\eta) \sinh(\xi'-i\eta')} = \pm \frac{1}{\sqrt a} \sqrt{\sinh(\xi+i\eta) \sinh(\xi'-i\eta')},
\end{align}
for some choice of $\pm$ sign that depends explicitly on $(Z,Z')$. Upon substituting elliptic coordinates in the remaining part of our expressions, and in particular using Lemma \ref{prop:valueomegaWZeFtau} and Lemma \ref{prop:weightsFa-1}, we arrive at \eqref{eq:behavKnZWgeneral_intro}.
\end{proof}
\begin{proof}[Proof of Theorem \ref{thm:2ptCluster_intro}]
    By \eqref{eq:RtauRa-Compare} we have $\operatorname{Re} F(a^{-1}) > \operatorname{Re} F(\tau)$ when $\xi=\xi'=\xi_\tau$.   By continuity, this strict inequality remains true on a small enough open neighborhood of $(Z_0, Z'_0)$. The contribution of the residue at $s=\tau$ is therefore negligible in this open neighborhood, and we may immediately extract from Theorem \ref{thm:largenKnZW_intro} that
    \begin{multline*}
    \mathbb K_n(Z,Z') 
    = \pm \sqrt{\frac{n}{32 \pi^3 \tau(1-\tau^2)}} 
    \frac{e^{- n (\xi-\xi_\tau)^2 g(\xi+i\eta)} e^{- n (\xi'-\xi_\tau)^2 g(\xi'+i\eta')}}{\sinh\left(\xi_+-\xi_\tau+i\frac{\eta-\eta'}{2}\right)\sqrt{\sinh(\xi+i\eta) \sinh(\xi'-i\eta')}}
    e^{i n (\eta-\eta' - \frac{\tau}{2} \sin(2\eta)+\frac{\tau}{2} \sin(2\eta'))}\\
    + \mathcal O\left(\frac{1}{\sqrt n} e^{- n (\xi-\xi_\tau)^2 g(\xi+i\eta)} e^{- n (\xi'-\xi_\tau)^2 g(\xi'+i\eta')}\right)
    \end{multline*} 
    as $n\to\infty$. Taking the modulus squared, we arrive at \eqref{eq:cluster2_intro}.  
    \end{proof}
    
\subsection{Additional results on the large $n$ behavior}

In this subsection we will discuss several other results of interest. 

We start with the following inequality, valid for all $n=1,2,\ldots$  The inequality produces meaningful estimates when the average of $\xi$ and $\xi'$ is of order $\sqrt\frac{\log n}{n}$ away from $\xi_\tau$. We denote the average of $\xi$ and $\xi'$ by $\xi_+$. 

\begin{theorem} \label{thm:errorTermUniform}
For all $Z, Z'\in\mathbb C$, $\tau\in (0,1)$ and all $n=1,2,\ldots$, we have the inequality
\begin{align} \label{eq:errorTermUniform1}
\left|\mathbb K_n(Z,Z') - n K_{(1-\tau^2)^{-1}}^{\infty}\left(\sqrt n Z, \sqrt n Z'\right) \mathfrak{1}_{\xi_+<\xi_\tau}\right|
&\leq  \frac{K}{2\pi \sqrt{1-\tau^2}}
\frac{n}{\lvert 1-e^{2(\xi_+-\xi_\tau)}\rvert} 
e^{- n (\xi-\xi_\tau)^2 g(\xi+i\eta)} e^{- n (\xi'-\xi_\tau)^2 g(\xi'+i\eta')},
\end{align} 
where $K= \displaystyle\frac{1}{\pi} \int_0^\pi \frac{dt}{\sqrt{2 \sin t}}$.
Furthermore, for fixed $\tau\in(0,1)$, if $(Z,Z')$ is in a set such that $(\xi_+ - \xi_\tau)^2+(\eta-\eta')^2$ is bounded from below by a positive constant, then there exists a constant $C>0$, depending only on $\tau$, such that
\begin{align} \label{eq:errorTermUniform2}
\left|\mathbb K_n(Z,Z') - n K_{(1-\tau^2)^{-1}}^{\infty}\left(\sqrt n Z, \sqrt n Z'\right) \mathfrak{1}_{\xi_+<\xi_\tau}\right|
&\leq C n e^{- n (\xi-\xi_\tau)^2 g(\xi+i\eta)} e^{- n (\xi'-\xi_\tau)^2 g(\xi'+i\eta')}.\\ \nonumber
\end{align} 
\end{theorem}


\begin{proof}
To obtain the first inequality \eqref{eq:errorTermUniform1}, we simply multiply \eqref{thm:errorTermUniform1d=1} with the relevant factor involving the weights as in \eqref{eq:EGEInwithWeights}, and then insert \eqref{eq:valueomegaWZeFtau} and \eqref{eq:weightsFa-1}. 
Let us move to the second inequality \eqref{eq:errorTermUniform2}. Then we have $(\xi_+ - \xi_\tau)^2+(\eta-\eta')^2\geq \mu$ for some $\mu>0$. When, say, $\xi+\xi' \not\in [\xi_\tau, 3\xi_\tau]$, we have $|1-e^{2(\xi_+-\xi_\tau)}|>1-\sqrt \tau$, and \eqref{eq:errorTermUniform2} follows directly from \eqref{eq:errorTermUniform1} for some choice of $C$. In the remaining cases we have $0<\xi_\tau \leq\xi+\xi' \leq 3\xi_\tau<\infty$, and this situation is described by Theorem \ref{thm:InGeneral}(i) (our set is compact in particular). The only difference with Theorem \ref{thm:largenKnZW_intro}, is that there can be more saddle points that contribute. By Theorem \ref{lem:as=1}, these contributions are of the same size as the one in $s=a^{-1}$ though. We conclude that there exists a constant $C>0$ such that, uniformly for $\xi+\xi'\in [\xi_\tau, 3\xi_\tau]$ and $(\xi_+ - \xi_\tau)^2+(\eta-\eta')^2\geq \mu$, we have
\begin{align*}
\left|\mathbb K_n(Z,Z') - n K_{(1-\tau^2)^{-1}}^{\infty}\left(\sqrt n Z, \sqrt n Z'\right) \mathfrak{1}_{\xi_+<\xi_\tau}\right|
\leq C\sqrt n e^{- n (\xi-\xi_\tau)^2 g(\xi+i\eta)} e^{- n (\xi'-\xi_\tau)^2 g(\xi'+i\eta')}.
\end{align*}
Of course, we may replace $\sqrt n$ by $n$, and, by possibly taking $C$ larger, we get \eqref{eq:errorTermUniform2} uniformly for $(\xi_+ - \xi_\tau)^2+(\eta-\eta')^2\geq \mu$, with no restriction on $\xi+\xi'$. 
\end{proof}

It is likely that \eqref{eq:errorTermUniform2} holds under the weaker condition that $(\xi - \xi_\tau)^2+(\xi' - \xi_\tau)^2+(\eta-\eta')^2$ is bounded from below by a positive constant, but we were not able to prove this. It is an immediate consequence of Theorem \ref{thm:errorTermUniform} that the mean density satisfies:
\begin{align} \label{eq:limitingDensityG}
    \frac{1}{n} \mathbb K_n(Z,Z) = \frac{\mathfrak{1}_{Z\in \mathcal E_\tau}}{\pi (1-\tau^2)} + \mathcal O\left(e^{- 2 n (\xi-\xi_\tau)^2 g(\xi+i\eta)}\right)
    \end{align}
as $n\to\infty$, uniformly on compact subsets of $\mathbb C\setminus \partial \mathcal E_\tau$. This should be compared to a general result in \cite{AmHeMa1}, that the mean density of random normal matrix models converges to an explicitly prescribed density, with uniform error term $\mathcal O\left(\frac{1}{n}\right)$ on compact subsets of the droplet. It is an interesting question whether a version of Theorem \ref{thm:errorTermUniform} holds for normal matrix models with general potentials.

Our methods even allow us to compute the correction term to the limiting density in \eqref{eq:limitingDensityG}. Interestingly, these first correction terms has a slightly different form for $z \in \mathcal C\setminus [-2\sqrt \tau,2\sqrt \tau]$, $z\in (-2\sqrt \tau, 2 \sqrt \tau)$  and $z=\pm 2 \sqrt{\tau}$.  This appears to be a trace of the zeros of the Hermite polynomials, that accumulate precisely on the interval $[-2\sqrt \tau, 2\sqrt \tau]$.


\begin{theorem} \label{thm:onepointAsymp}
The one-point correlation function of the elliptic Ginibre ensemble behaves asymptotically as follows. 
\begin{itemize}
\item[(i)] When $Z\in \mathbb C\setminus [-2\sqrt \tau, \sqrt \tau]$ and $Z\not\in \partial \mathcal E_\tau$, we have uniformly on compact subsets that, as $n\to\infty$
\begin{align*}
\frac{1}{n} \mathbb K_n(Z, Z)
= \frac{\mathfrak{1}_{Z\in \mathcal E_\tau}}{\pi (1-\tau^2)}
+\frac{1}{\sqrt{32 \pi^3 n}} \frac{1}{\sqrt{\tau(1-\tau^2)}} 
\frac{e^{-2 n (\xi-\xi_\tau)^2 g(\xi+i\eta)}}{\sinh(\xi-\xi_\tau) \lvert \sinh(\xi+i\eta)\rvert} 
+ \mathcal O\left(\frac{1}{n\sqrt n} e^{-2 n (\xi-\xi_\tau)^2 g(\xi+i\eta)}\right).
\end{align*}
\item[(ii)] When $Z\in (-2\sqrt\tau,2\sqrt\tau)$, we have uniformly on compact subsets that, as $n\to\infty$
\begin{align} \label{eq:EGEiiOnePointC}
\frac{1}{n} \mathbb K_n(Z, Z)
= \frac{1}{\pi (1-\tau^2)}
+\frac{1}{\sqrt{2 \pi^3 n}} \frac{1}{\sqrt{1-\tau^2}} f_n(Z) 
\frac{e^{-2 n \xi_\tau^2 g(i\eta)}}{\sin\eta}
+ \mathcal O\left(\frac{1}{n\sqrt n} e^{-2 n \xi_\tau^2 g(i\eta)}\right),
\end{align}
where $f_n:(-2\sqrt \tau, 2\sqrt\tau)\to \mathbb R$ is defined by
\begin{align*}
f_n(Z) =- \frac{1}{1-\tau} + 
\begin{cases}
\displaystyle  \frac{(1-\tau) \cos \eta \sin\left(n(2\eta-\sin 2\eta) - \frac{\pi}{4}\right) + (1+\tau) \sin \eta \cos\left(n(2\eta-\sin 2\eta) - \frac{\pi}{4}\right)}{1+\tau^2-2\tau\cos 2\eta}, & Z> 0,\\
\displaystyle\frac{(-1)^n}{1+\tau}, & Z=0,\\
\displaystyle \frac{(1-\tau) \cos \eta \sin \left(n(2\eta-\sin 2\eta) + \frac{\pi}{4}\right) + (1+\tau) \sin \eta \cos\left(n(2\eta-\sin 2\eta) + \frac{\pi}{4}\right)}{1+\tau^2-2\tau\cos 2\eta}, & Z< 0.
\end{cases}
\end{align*}
\item[(iii)] When $Z = \pm 2\sqrt\tau$, we have as $n\to\infty$
\begin{align*}
\frac{1}{n} \mathbb K_n(Z, Z) 
= \frac{1}{\pi (1-\tau^2)}  
-\frac{\Gamma\left(\frac{1}{6}\right)}{\pi^2 2^\frac{7}{6} 3^\frac{1}{3}
n^\frac{1}{6}} \frac{1}{\sqrt{1-\tau^2}} \frac{e^{- 2n \xi_\tau^2 g(0)}}{1-\tau}
+ \mathcal O\left(n^{-\frac{2}{3}} e^{- 2n \xi_\tau^2 g(0)}\right).
\end{align*}
\end{itemize}
\end{theorem}

\begin{proof}
Case (i) can directly be extracted from Theorem \ref{thm:largenKnZW_intro}. When $Z=Z'$, we always have to take the $+$ sign in \eqref{eq:pmSqrtSaddle}. In this case, we can write
\begin{align*}
\sinh\left(\xi_+-\xi_\tau + i \frac{\eta-\eta'}{2}\right)\sqrt{\sinh(\xi+i\eta) \sinh(\xi'-i\eta')}
= \sinh(\xi-\xi_\tau) \lvert \sinh(\xi+i\eta)\rvert.
\end{align*}
Case (ii) falls under Theorem \ref{thm:InGeneral}(ii). 
Let us first consider the case that $Z\neq 0$. By symmetry, we may assume without loss of generality that $Z>0$, i.e., $\eta\in (0,\frac{\pi}{2})$. Then $a = e^{2i\eta}$ is in the upper half-plane, while $a^{-1}=e^{-2i\eta}$ is in the lower half-plane. 
The path through $a$ and $a^{-1}$ begins in $e^{i(2\eta-\sin 2\eta)} \infty$ and ends at $e^{-i(2\eta-\sin 2\eta)}\infty$ (see Figure \ref{Fig2a}). 
Note that, by Proposition \ref{prop:F''a-1}, 
\begin{align*}
F''(a^{-1}) = 2 i e^{4i\eta} \frac{\sin^2 \eta}{\sin 2\eta}.
\end{align*}
Thus the direction of our path in $s=a^{-1}$ is $-2\eta+\frac{\pi}{2}$ (the argument of the integration contour increases in this direction). The corresponding saddle point contribution for $I_n$ is thus given by
\begin{align*}
-\frac{1}{2\pi i} \sqrt{\frac{2\pi}{n}} \frac{e^{n F(a^{-1})}}{\sqrt{1-a^{-2}}} \frac{1}{a^{-1}-\tau} \frac{e^{i(-2\eta+\frac{\pi}{2})}}{|F''(a^{-1})|^{1/2}}
&= \frac{1}{2\pi i} \sqrt{\frac{2\pi}{n}} e^{n F(a^{-1})} \frac{1-i}{2\sqrt 2} \frac{e^{-i\eta}}{e^{-2i\eta}-\tau} \frac{1}{\sin \eta},
\end{align*}
where we used Proposition \ref{prop:F''a-1} for the value of $F''(a^{-1})$. 
The saddle point contribution of $a$ is just the complex conjugate of this, hence the sum of these contributions equals
\begin{multline*}
\frac{1}{4\pi i} \sqrt{\frac{2\pi}{n}} \frac{\tau^n e^{n (1+\cos 2\eta)}}{\sin \eta} 
\left(\frac{e^{i n (2\eta - \sin 2\eta) - i \frac{\pi}{4}}}{e^{-i\eta} - \tau e^{i\eta}}-
\frac{e^{i n (-2\eta + \sin 2\eta) + i \frac{\pi}{4}}}{e^{i\eta} - \tau e^{-i\eta}}\right)\\
= \frac{1}{\sqrt{2\pi n}} \frac{\tau^n e^{n (1+\cos 2\eta)}}{\sin \eta} 
\frac{(1-\tau) \cos \eta \sin\left(n(2\eta-\sin 2\eta) - \frac{\pi}{4}\right) + (1+\tau) \sin \eta \cos\left(n(2\eta-\sin 2\eta) - \frac{\pi}{4}\right)}{1+\tau^2-2\tau\cos 2\eta}
\end{multline*}
The contribution of $s=1$ gives 
\begin{align*}
- \frac{e^{n F(1)}}{\Gamma\left(\frac{1}{2}\right)} \frac{|2 nF'(1)|^{\frac{1}{2}-1}}{1-\tau}
= - \frac{e^{n(1+\cos 2\eta)} \tau^n}{\sqrt{2 \pi n}} \frac{1}{\sin \eta} \frac{1}{1-\tau}.
\end{align*}
Combining these contributions gives \eqref{eq:EGEiiOnePointC}, with $f$ as defined in the theorem.  
In the case $Z=0$, there are no saddle point contributions and, instead, the contributions come from $-1$ and $1$ (the integration contour is deformed to a band around $[1,\infty)$ and a band around $(-\infty,-1])$, see the proof of Theorem \ref{thm:InGeneral}(ii)). The dominant contribution of $I_n$ is then given by
\begin{align*}
 -\frac{e^{n F(-1)}}{\Gamma\left(\frac{1}{2}\right)} \frac{|2 nF'(-1)|^{\frac{1}{2}-1}}{-1-\tau}
- \frac{e^{n F(1)}}{\Gamma\left(\frac{1}{2}\right)} \frac{|2 nF'(1)|^{\frac{1}{2}-1}}{1-\tau}
= \frac{\tau^n}{\sqrt{2\pi n}} \left(\frac{(-1)^n}{1+\tau}-\frac{1}{1-\tau}\right)
= \frac{e^{n \operatorname{Re} F(a^{-1})}}{\sqrt{2\pi n}} f(0).
\end{align*} 

Finally, we treat the case (iii). By symmetry, we may limit ourselves to $Z=2\sqrt \tau$. We have only one saddle point, $s= 1$, which is degenerate of order $2$. Clearly, one steepest ascent direction in $s=1$ has angle $\pi$. Then the three steepest descent directions have angles $\frac{2\pi}{3}, -\frac{2\pi}{3}$ and $0$. The integration contour $\gamma_0$ can only be deformed to a contour that uses the first two directions.  The saddle point contribution of $I_n$ is then given by 
\begin{multline*}
-\frac{1}{2\pi i} \frac{\Gamma\left(\frac{1/2}{3}\right)}{3} \frac{e^{n F(1)}}{(1-\tau) \sqrt 2} \left(\frac{3!}{|F'''(1)|n}\right)^{\frac{1/2}{3}} e^{\frac{2\pi i}{3}} 
+\frac{1}{2\pi i} \frac{\Gamma\left(\frac{1/2}{3}\right)}{3} \frac{e^{n F(1)}}{(1-\tau) \sqrt 2} \left(\frac{3!}{|F'''(1)|n}\right)^{\frac{1/2}{3}} e^{-\frac{2\pi i}{3}} 
+ \mathcal O\left(n^{-\frac{2}{3}} e^{n F(1)}\right)\\
= -\frac{\sqrt{3}\Gamma\left(\frac{1}{6}\right)}{4\sqrt 2 \pi} \frac{(12)^\frac{1}{6} e^{n F(1)}}{n^\frac{1}{6}(1-\tau^2)}
+ \mathcal O\left(n^{-\frac{2}{3}} e^{n F(1)}\right). 
\end{multline*}
The residue will also contribute, and, after multiplying by the factor involving the weights, we obtain the result. 
\end{proof}

%






\section{Fermionic process in $d$ dimensions} \label{sec:fermions}

In this section we consider a model of non-interacting fermions in $d$ dimensions. 
After setting 
\begin{align*}
z &= \frac{\lvert X+X' \rvert + \lvert X-X' \rvert}{2},\\
z' &= \frac{\lvert X+X' \rvert - \lvert X-X' \rvert}{2},
\end{align*}
or equivalently $z\pm z'=|X\pm X'|$, then we see that we can write $\mathbb K_n^{\rm Fermi}$ in \eqref{eq:kernelferm} as 
\begin{equation}\label{eq:fermionsKnI}
\mathbb K_n^{\rm Fermi}(X,X') 
= \left(\frac{n}{\pi}\right)^\frac{d}{2} e^{-\frac{1}{2} n (z^2+z'^2)} I_n(d;1;z,z'). 
\end{equation}
where $I_n$ is as in \eqref{eq:defIn}.
 As before, we will write $z$ and $z'$ in elliptic coordinates as defined in \eqref{eq:defEllipticCoord}.  Note that since we are in the real line, using these coordinates we find  
 $$
   \lvert X\rvert = 
        \sqrt 2 \cosh \xi, \qquad \xi>0, $$
if $\lvert X\rvert>\sqrt 2$, and 
$$
\lvert X\rvert =    \sqrt 2 \cos \eta, \quad \eta \in[0,\frac{\pi}{2}]
 $$
 for  $\lvert X\rvert<\sqrt 2$.

\begin{proof}[Proof of Theorem \ref{thm:fermionsBulkDensity_intro}.]
Since $\tau=1$, we have to apply Theorem \ref{thm:InGeneralTau1}. When $X=X'$, we have $z=z'=\lvert X\rvert$. Let us first look at the case $0<\lvert X\rvert < \sqrt 2$. Then we are in case (ii) of Theorem \ref{thm:InGeneralTau1}. The saddle points are $a=e^{2i\eta}$ and $a^{-1}=e^{-2i\eta}$ with $\eta\in (0,\frac{\pi}{2})$. The contribution of $s=-1$ in \eqref{eq:generalSteepestDresult2Tau1} should be excluded, since it only occurs when $z=-z'$. Hence we have
\begin{align*}
I_n(d;1;z,z) &= -2\operatorname{Re}\left( \frac{1}{2\pi i} \sqrt{\frac{2\pi}{- n F''(a)}} \frac{e^{n F(a)}}{a-1} \frac{1}{(1-a^2)^\frac{d}{2}} \right) +  \frac{e^{n F(1)}}{2^d \Gamma\left(\frac{d}{2}+1\right)} \lvert 2 n F'(1)\rvert^\frac{d}{2} 
+ \mathcal O\left(n^{\frac{d}{2}-2} e^{n \operatorname{Re} F(a^{-1})}\right).
\end{align*}
Substituting $|2 F'(1)| = |z^2-2| = 2-|X|^2$, and plugging in elliptic coordinates where applicable, we find
\begin{align*}
I_n(d;1;z,z) &=
-2\frac{1}{\sqrt{2\pi n}} \operatorname{Re}\left(\frac{e^{n z^2 - i n(2\eta-\sin 2\eta) - i d \eta}}{\sin^2 \eta \sin^{\frac{d}{2}-1}2\eta} \right)
+ \frac{e^{n z^2}}{2^d \Gamma\left(\frac{d}{2}+1\right)} n^\frac{d}{2} (2-|X|^2)^\frac{d}{2} 
+ \mathcal O\left(n^{\frac{d}{2}-1} e^{n z^2}\right).
\end{align*}
Multiplying by $e^{-n z^2} \frac{1}{n^d/d!} (n/\pi)^{\frac{d}{2}}$, we arrive at \eqref{eq:fermionsBulkDensity1_intro}. The case $X=0$ follows immediately from Theorem \ref{thm:InGeneralTau1}(ii). 
When $\lvert X\rvert >\sqrt 2$, we use Theorem \ref{thm:InGeneralTau1}(i). Here we have $\xi>0$ and $\eta=0$. Only the saddle point $a^{-1} = e^{-2\xi}$ contributes, hence
\begin{align*}
I_n(d;1;z,z) = - \frac{1}{\sqrt{\pi n}} \frac{e^{d\xi}}{2^{\frac{d}{2}+2}\sinh^2 \xi \sinh^{\frac{d}{2}-1} 2\xi} e^{n(1+2\xi+e^{-2\xi})}
\left(1+\mathcal O\left(\frac{1}{n}\right)\right),
\end{align*}
as $n\to\infty$. Multiplying by $e^{-n z^2} \frac{1}{n^d/d!} (n/\pi)^{\frac{d}{2}}$, and realizing that $2\xi - e^{2\xi} < 0$, we obtain the result.\\
\end{proof}



\begin{proof}[Proof of Theorem \ref{thm:fermionsScalingBulk_intro}]
We deform $\gamma_0$ to the contour in Figure \ref{Fig2a} (where $\eta = \arccos(\lvert X\rvert/\sqrt 2)$). Similar to the proof of Theorem \ref{thm:fermionsBulkDensity_intro} above, we can show that the saddle point contributions in $e^{\pm i \eta}$ are negligible. For the band $\gamma_3$ around $[1,\infty)$, we argue as follows. Using local coordinates $s=1-t/n$, we have
\begin{align} \label{eq:fermionsScalingBulk1}
n F(s) = n \lvert X\rvert^2 - \frac{\lvert X\rvert^2 t}{2} + t 
+2\frac{\langle X,U+V\rangle}{(2-|X|^2)^{\frac12}}
- 2\frac{\lvert U-V\rvert^2}{(2-\lvert X\rvert^2)\, t} + \mathcal O\left(\frac{1}{n}\right)
\end{align}
as $n\to\infty$, where the $\mathcal O$ term is uniform for $t$ in compact sets. Note that, in these coordinates, the integration contour, let us call it $\tilde\gamma$, 
bends around $(-\infty, 0]$, with positive orientation. 
Consequently, we have as $n\to\infty$ that
\begin{align*}
-\frac{1}{2\pi i}\int_{\gamma_3} \frac{e^{n F(s)}}{1-s} \frac{1}{(1-s^2)^\frac{d}{2}} ds
&= \frac{1}{2\pi i} \frac{n^{\frac{d}{2}}}{2^\frac{d}{2}} 
\exp\left(n\lvert X\rvert^2+2\frac{\langle X,U+V\rangle}{(2-|X|^2)^{\frac12}}\right)
\int_{\tilde\gamma} e^{\frac{2-\lvert X\rvert^2}{2} t - 2\frac{\lvert U-V\rvert^2}{(2-\lvert X\rvert^2)\, t}} \frac{dt}{t^{\frac{d}{2}+1}}\left(1+ \mathcal O\left(\frac{1
}{n}\right)\right)\\
&= \frac{1}{2\pi i} \frac{n^{\frac{d}{2}}}{2^d} \left(2-\lvert X\rvert^2\right)^\frac{d}{2} 
\exp\left(n\lvert X\rvert^2+2\frac{\langle X,U+V\rangle}{(2-|X|^2)^{\frac12}}\right) \int_{\tilde\gamma} e^{t - \frac{\lvert U-V\rvert^2}{t}} \frac{dt}{t^{\frac{d}{2}+1}}
\left(1+\mathcal O\left(\frac{1
}{n}\right)\right).
\end{align*}
Here we have tacitly extended \eqref{eq:fermionsScalingBulk1} to $t$ on the entire contour $\tilde\gamma$, which is allowed because the tail of the integrand tends to $0$ fast enough. By \cite[10.9.19]{DLMF}, we have
\begin{align*}
\frac{1}{2\pi i}\int_{\tilde\gamma} e^{t - \frac{\lvert U-V\rvert^2}{4 t}} \frac{dt}{t^{\frac{d}{2}+1}}
= 
\frac{J_\frac{d}{2}\left(2\lvert U-V\rvert\right)}{\lvert U-V\rvert^\frac{d}{2}}.
\end{align*}
We conclude that
\begin{align*}
n^{-\frac{d}{2}} e^{-n \lvert X\rvert^2} I_n\left(d;1;X+\frac{U}{\nu(X)^\frac{1}{d} n}, X + \frac{V}{\nu(X)^\frac{1}{d} n}\right)
= 2^{1-d}(2-\lvert X\rvert^2)^\frac{d}{2} \frac{J_\frac{d}{2}\left(\lvert U-V\rvert\right)}{\lvert U-V\rvert^\frac{d}{2}}
+\mathcal O\left(\frac{1}{n}\right)
\end{align*}
as $n\to\infty$. Inserting the relevant remaining factors, we arrive at \eqref{eq:fermionsScalingBulk_intro}. Everywhere in our derivation, the $\mathcal O$ terms are uniform for $(U,V)$ in compact sets.\\ 
\end{proof}



\begin{figure}
    \centering
    \begin{overpic}[width=0.5\textwidth]{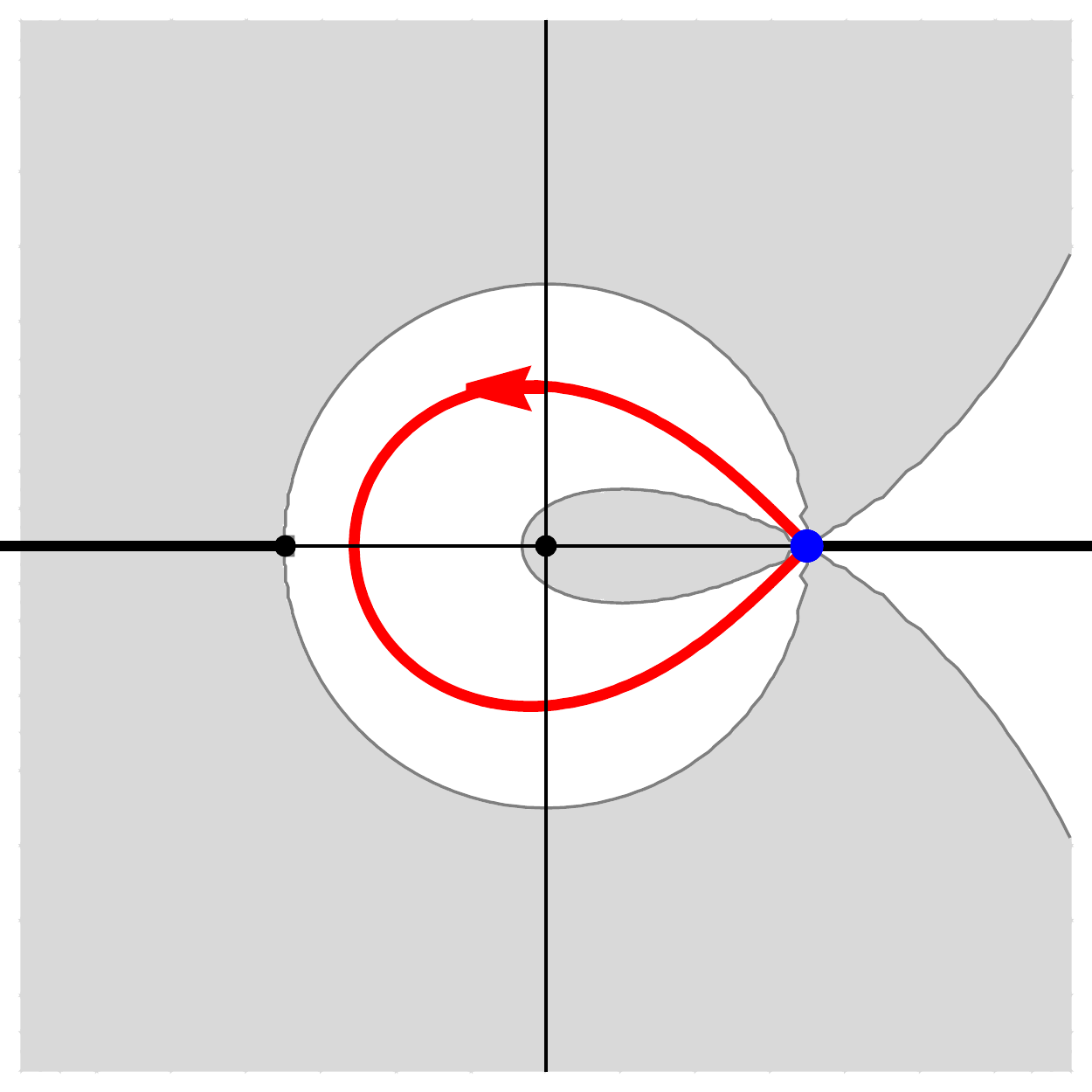}
       
        \put (25,45) {$-1$}
        \put (75,45) {$1$}
       
        \put (46,46) {$0$}
      
    \end{overpic}
    \caption{In the proof of Theorem \ref{thm:fermionsEdgeDensity_intro}  we deform the contour $\gamma_0$ to the contour in the picture. The shaded areas represent the regions where $\operatorname{Re} F(s)\geq \operatorname{Re} F(1)$. The two shaded areas are bounded. \label{Fig2aaa}} 
    \end{figure}

\begin{proof}[Proof of Theorem \ref{thm:fermionsEdgeDensity_intro}.]
If $U,V=0$ then $z,z'\in \sqrt 2$ and, as indicated in Proposition \ref{prop:saddlePointsDef}, we have that all four  saddle points for $\Re F(z)$ coalesce with the pole at $s=1$, forming a critical point of order 2 in $s=1$.  In Figure \ref{Fig2aaa} we have plotted the regions where $\Re F(s)>\Re F(1)$ in gray.  We first deform the contour $\gamma_0$ so that it lies entirely in the white region but passing through $s=1$. However, the integrand has a singularity in the term $(1-s^2)^{-d/2}(s-1)^{-1}$. While deforming $\gamma_0$ we make sure that the contour intersects the real axis slightly at the left of $1$. We take this choice of contour for all choices of $U$ and $V$.

The main contribution in the asymptotic expansions now comes from a small neighborhood near $s=1$. We introduce the local variable $s = 1- n^{-\frac{1}{3}} t $. 
Some bookkeeping will yield 
\begin{align*}
n F(s) =& 2 n + \frac{1}{\sqrt 2} \langle X,U+V\rangle n^{\frac{1}{3}} - \frac{1}{2} \frac{1}{\sqrt 2} \langle X,U+V\rangle t + \frac{t^3}{12} - \frac{\lvert U-V\rvert^2}{4 t} +  \mathcal O(n^{-1/3}).
\end{align*}
Notice that we have
\begin{align*}
e^{-\frac{1}{2}n\left(\lvert X+U n^{-\frac{2}{3}}/\sqrt 2 \rvert^2+\lvert X+V n^{-\frac{2}{3}}/\sqrt 2 \rvert^2\right)} 
= e^{-2n - \langle X,U+V\rangle n^{\frac{1}{3}}/\sqrt 2} (1+\mathcal O(n^{-1/3})).
\end{align*}
Consequently, we have as $n\to\infty$ that
\begin{align}
\frac{1}{(n^\frac{2}{3} \sqrt 2)^d}{\mathbb K}^{\rm Fermi}_n\left(X+\frac{U}{n^\frac{2}{3}\sqrt 2}, X + \frac{V}{n^\frac{2}{3}\sqrt 2}\right) 
&=  \frac{1}{(2 \pi i) \pi^\frac{d}{2}}\frac{n^\frac{d}{2}}{n^\frac{2 d}{3} 2^\frac{d}{2}} \frac{n^\frac{d}{6}}{2^\frac{d}{2}} \int_{\tilde\gamma} e^{-\frac{1}{2\sqrt 2} \langle X,U+V\rangle t + \frac{t^3}{12} - \frac{\lvert U-V\rvert^2}{4 t}} \frac{dt}{t^{\frac{d}{2}+1}}+ \mathcal O\left(n^{-\frac{1}{3}}\right)
\nonumber\\
&= \frac{1}{(2\pi i) (2^\frac{4}{3} \sqrt \pi)^d} \int_{\tilde\gamma} e^{-2^{-1/3} \langle X,U+V\rangle t/\lvert X\rvert  + \frac{t^3}{3} - \frac{\lvert U-V\rvert^2}{2^{8/3} t}} \frac{dt}{t^{\frac{d}{2}+1}} + \mathcal O\left(n^{-\frac{1}{3}}\right),
\label{eq:KAiRd}
\end{align}
where $\tilde\gamma$ is a curve in the right half place that starts at $e^{-\pi i/3} \infty$ and ends at $e^{\pi i/3} \infty$. 
This is the expression (19) in \cite{DeDoMaSc}, except that they use a Bromwich contour $\Gamma$, to arrive at an integral representation of the Airy function
\begin{equation}\label{eq:Airy-rep}
\operatorname{Ai}(z)=\frac{1}{2\pi i}\int_{\Gamma}\exp[-z t +\frac13 t^3]dt .
\end{equation}
Indeed, we may simply deform our integration contour $\tilde\gamma$ to a vertical contour in the right half-plane $\Gamma$. 
For completeness and later comparison to the weak non-Hermiticity limit we repeat the steps presented in 
\cite{DeDoMaSc} to arrive at \eqref{eq:fermionsAiryd_intro}. As first step, we need to get rid of the inverse power in $t$ in the exponential in \eqref{eq:KAiRd}. This is achieved by the following $d$ dimensional integral representation over a vector $Q\in\mathbb{R}^d$
\begin{equation}\label{eq:propagator}
(4\pi D\tau)^{-\frac{d}{2}}e^{-\frac{|U-V|^2}{4D\tau}}
=\frac{1}{(2\pi)^d}\int_{\mathbb{R}^d} e^{-D|Q|^2\tau-i\langle Q, U-V\rangle} d^dQ.
\end{equation}
There is a mismatch by one power $1/t$ in \eqref{eq:KAiRd} compared to the normalisation of this integral. It is moved to the exponent by expressing it though an integral, valid for $\Re(t)>0$
\begin{equation}\label{eq:1/t}
\frac{1}{t} =\int_0^\infty e^{-st} ds.
\end{equation}
Following these steps we arrive at 
\begin{multline*}
\frac{1}{(n^\frac{2}{3} \sqrt 2)^d}{\mathbb K}^{\rm Fermi}_n\left(X+\frac{U}{n^\frac{2}{3}\sqrt 2}, X + \frac{V}{n^\frac{2}{3}\sqrt 2}\right) 
\\
=\frac{1}{2\pi i} \int_{\Gamma} \int_0^\infty \int_{\mathbb{R}^d} 
\exp\left(-t(s+2^{2/3}|Q|^2+\frac{\langle X,U+V\rangle}{2^{1/3}|x|})+\frac13 t^3-i\langle Q,U-V\rangle\right) \frac{d^dQ}{(2\pi)^d}ds dt,\nonumber
+ \mathcal O\left(n^{-\frac{1}{3}}\right).
\end{multline*}
which upon using \eqref{eq:Airy-rep} and an appropriate interchange of integrals leads to the desired result in Theorem \ref{thm:fermionsEdgeDensity_intro}.
\end{proof}

%
%
%
%
%
%
%
%

\section{The $d$-dimensional elliptic Ginibre ensemble} \label{sec:general}

In this section we analyze the the generalization to $\mathbb C^d$ of the elliptic Ginibre ensemble for $d >1$.  

In what follows, we will use the short-hand notation $\zeta^2 = \zeta_1^2+\ldots+\zeta_d^2$ for any $\zeta\in\mathbb C^d$, not to be confused with the squared norm $\lvert \zeta\rvert^2 = |\zeta_1|^2+\ldots+|\zeta_d|^2$. With this notation we start with the following lemma. 

\begin{lemma}\label{lem:general_kernel}
    The kernel $\mathbb K_n$ in  \eqref{eq:general_kernel_d=1} can be written as
\begin{align} \label{eq:defEllipdIn2}
\mathbb K_n(Z, Z')
= \frac{n^d}{\pi^d (1-\tau^2)^\frac{d}{2}} 
\sqrt{ \omega\left(\sqrt{2\tau n} z\right)  \omega\left(\sqrt{2\tau n} z'\right)} I_n\left(d;\tau; z, z'\right). 
\end{align}
where $I_n$ is as in \eqref{eq:defIn} and $z,z'$ are chosen such that 
\begin{align} \label{eq:ztoZ1}
(z+z')^2&=\frac{1}{2\tau} (Z+\overline{Z'})^2,\\
(z-z')^2&=\frac{1}{2\tau} (Z-\overline{Z'})^2.\label{eq:ztoZ2}
\end{align}
\end{lemma}
Before we come to the proof we note that $Z,Z'$ are in $\mathbb C^d$ and $z,z'\in \mathbb C$. Note also that there is an ambiguity in this definition: if $z,z'$ satisfy \eqref{eq:ztoZ1}--\eqref{eq:ztoZ2}, then so do $-z,-z'$. However, by \eqref{eq:defF} the integral $I_n(d;\tau;z,z')$ only depends on the values of $(z+z')^2$ and $(z-z')^2$. The same is true for the product of the weights, as will be clear from the proof of the lemma. The ambiguity is thus irrelevant and one can make either choice. 

\begin{proof}
    We start by writing 
    \begin{equation}
        \prod_{j=1}^d\omega(\sqrt n Z_j) \omega(\sqrt n Z_j')= \exp\left(-\frac{n}{1-\tau^2} \left(|Z|^2+|Z'|^2-\frac{\tau}{2} (Z^2+\overline Z ^2+(Z')^2+\overline{ (Z')} ^2)\right)\right)=\omega(\sqrt n Z)\omega(\sqrt n Z'),
    \end{equation}
    which, by the parallelogram law, we can write as
    \begin{equation} \label{eq:weightproducts}
        \prod_{j=1}^d\omega(\sqrt n Z_j) \omega(\sqrt n Z_j')= \exp\left(-\frac{n}{1-\tau^2} \left(\frac12 |Z+\overline {Z'}|^2+\frac12|Z-\overline {Z'}|^2-\frac{\tau}{4} ((Z+ \overline{Z'})^2+(\overline Z + {Z'})^2+(Z- \overline{Z'})^2+(\overline Z- {Z'})^2\right)\right).
    \end{equation}
   By \eqref{eq:ztoZ1}--\eqref{eq:ztoZ2} this shows that 
   \begin{equation}
        \prod_{j=1}^d\omega(\sqrt n Z_j) \omega(\sqrt n Z_j')=\omega(\sqrt{2 \tau n} z) \omega(\sqrt {2 \tau n} z').
   \end{equation}
The statement now follows from  \eqref{eq:kernelboldKtoK}, \eqref{eq:integral_representation_kernel} and \eqref{eq:defIn}.
\end{proof}

\subsection{Strong non-Hermiticity: proofs of Theorems \ref{thm:limitingDensity} and \ref{thm:Cd_bulk_intro}}
The proofs of Theorem \ref{thm:limitingDensity} and \ref{thm:Cd_bulk_intro} will follow (partly) from the following proposition, which is a generalization of \text{Theorem \ref{thm:errorTermUniform}} to $d>1$. We remind the reader that $\xi_+$ is the average of $\xi$ and $\xi'$. For the proofs below, it is convenient to define the (closed) $d$-dimensional ball
\begin{align*}
B_\tau^d = \{Z\in \mathbb R^d : \lvert Z\rvert \leq 2\sqrt \tau\}.
\end{align*}
Its boundary is a $d$-dimensional sphere with radius $2\sqrt \tau$, which we denote by $\partial B_\tau^d$. 

\begin{proposition} \label{thm:errorTermUniformd>1}
    Let $d$ be a positive integer. For all $Z, Z'\in\mathbb C$, $\tau\in (0,1)$ and all $n=1,2,\ldots$, we have the inequality
    \begin{multline} \label{eq:errorTermUniform1d>1}
    \left|\mathbb K_n(Z,Z') - \mathfrak{1}_{\xi_+<\xi_\tau} n^d \prod_{j=1}^d K_{(1-\tau^2)^{-1}}^{\infty}\left(\sqrt n Z_j, \sqrt n Z'_j\right)\right|\\
    \leq  \frac{n^d}{\lvert 1-e^{-2(\xi_+-\xi_\tau)}\rvert} \left(\frac{\tau}{\pi(1-\tau^2) \sqrt{\sinh \xi_+}}\right)^d 
     e^{- n (\xi-\xi_\tau)^2 g(\xi+i\eta)} e^{- n (\xi'-\xi_\tau)^2 g(\xi'+i\eta')}.
    \end{multline} 
    \end{proposition}
    
    \begin{proof}
    We notice that
    \begin{align*}
    F(\tau;z,z';s) = -(d-1) (\log s - \log \tau) + \sum_{j=1}^d F\left(\tau;\frac{Z_j}{\sqrt{2\tau}}, \frac{Z'_j}{\sqrt{2\tau}};s\right).
    \end{align*}
    Thus, taking $s=\tau$, by 
    Lemma \ref{prop:valueomegaWZeFtau} we have
    \begin{align} \label{eq:errorTermUniformd>1,1}
    \frac{1}{\pi^d (1-\tau^2)^\frac{d}{2}} e^{n F(\tau;z,z';\tau)}
    \sqrt{ \omega\left(\sqrt{2\tau n} z\right)  \omega\left(\sqrt{2\tau n} z'\right)}
    = \prod_{j=1}^d K_{(1-\tau^2)^{-1}}^{\infty}\left(\sqrt n Z_j, \sqrt n Z'_j\right).
    \end{align}
    Similarly, we have by 
    Lemma \ref{prop:weightsFa-1} that
    \begin{align} \label{eq:errorTermUniformd>1,2}
   \left| e^{n F(\tau;z,z';a^{-1})}
    \sqrt{ \omega\left(\sqrt{2\tau n} z\right)  \omega\left(\sqrt{2\tau n} z'\right)}\right|
    &= e^{- n (\xi-\xi_\tau)^2 g(\xi+i\eta)} e^{- n (\xi'-\xi_\tau)^2 g(\xi'+i\eta')}.
    \end{align}
    We plug these formulas into \eqref{eq:defEllipdIn2}. Combining \eqref{eq:errorTermUniformd>1,1} and \eqref{eq:errorTermUniformd>1,2}, we get, using arguments analogous to the proof of Proposition  \ref{thm:errorTermUniform1}, that \eqref{eq:errorTermUniform1d>1} holds. 
    \end{proof}
    
\begin{proof}[Proof of Theorem \ref{thm:limitingDensity}]
Denote by $\mathcal N$ a small compact neighborhood of $B_\tau^d$. For $Z\not\in \mathcal N$, the result follows trivially by \text{Proposition \ref{thm:errorTermUniformd>1}}. Now consider $Z\in \mathcal N$. We cannot directly apply Proposition \ref{thm:errorTermUniformd>1}, because there is singular behavior in \eqref{eq:errorTermUniform1d>1} when $Z\in B_\tau^d$. Therefore, we use an adaptation of the argument in Proposition \ref{thm:errorTermUniform1}. First, we notice that
\begin{align*}
F(s) = 2 \operatorname{Re}(z)^2 \frac{s}{1+s} + 2 \operatorname{Im}(z)^2 \frac{s}{1-s} - \log s + \log \tau. 
\end{align*}
Now let us consider the circle $|s|=r$ for some $\tau<r<|a|^{-1}$. The maps $s\mapsto s/(1+s)$ and $s\mapsto s/(1-s)$ are M\"obius transformations, mapping $|s|=r$ to circles. These circles are symmetric in the real line, and they intersect with the real line in $s=-r$ and $s=r$. The latter gives the maximum of the real part for both transformations, hence
\begin{align*}
\operatorname{Re} F(s) \leq 
2 \operatorname{Re}(z)^2 \frac{r}{1+r} + 2 \operatorname{Im}(z)^2 \frac{r}{1-r} - \log r + \log \tau 
=\operatorname{Re} F(r). 
\end{align*}
Since these M\"obius transformations are increasing on $(0,1)$, we infer that for all $s$ on the circle $|s|=r$
\begin{align} \nonumber
\operatorname{Re} F(s) &\leq 
2 \operatorname{Re}(z)^2 \frac{|a|^{-1}}{1+|a|^{-1}} + 2 \operatorname{Im}(z)^2 \frac{|a|^{-1}}{1-|a|^{-1}} - \log r + \log \tau \\ \nonumber
&= \operatorname{Re} F(|a|^{-1}) + \log |a|^{-1} - \log r\\
 \label{thm:limitingDensityIneqFsFr}
&= \operatorname{Re} F(a^{-1}) + \log |a|^{-1} - \log r.
\end{align}
In the last line, we used that $a^{-1}=|a|^{-1}$ when $Z\not\in B_\tau^d$, and we used Theorem \ref{lem:as=1}(iv) when $Z\in B_\tau^d$ (in which case $a^{-1} = e^{-2i\eta}$). Now denote by $\varepsilon$ the minimum of $\frac{1}{2}(\operatorname{Re} F(\tau) - \operatorname{Re} F(a^{-1})) > 0$ (see Remark \ref{remark:ReFtauReFa-1}) for $Z\in \mathcal N$. Given $Z\in \mathcal N$, we pick
\begin{align*}
r = \max\left(|a|^{-1} e^{-\varepsilon},\frac{|a|^{-1}+\tau}{2}\right).
\end{align*}
According to \eqref{thm:limitingDensityIneqFsFr}, with this choice of $r$, we have 
$
\operatorname{Re} F(s) \leq  \operatorname{Re} F(|a|^{-1} e^{-\varepsilon})
\leq \operatorname{Re} F(a^{-1}) + \varepsilon$. We also used here that $\operatorname{Re} F$ is strictly decreasing on $(0,|a|^{-1})$. 
The argument behind \eqref{eq:errorTermUniform1d>1}, i.e., the proof of Proposition \ref{thm:errorTermUniform1}, can, instead of the circle $|s|=|a|^{-1}$, also be applied to the smaller circle $|s|=r$.
        Repeating the arguments from Proposition \ref{thm:errorTermUniform1} and Proposition \ref{thm:errorTermUniformd>1} then yields
        \begin{align*}
        \left|\mathbb K_n(Z,Z) - \frac{n^d}{\pi^d (1-\tau^2)^d}\right|
        &\leq 
        \frac{n^d}{1-\tau/r} \left(\frac{\tau}{\pi(1-\tau^2) \sqrt{|\sinh \log \sqrt r|}}\right)^d 
     e^{n (\operatorname{Re} F(a^{-1}) + \varepsilon - \operatorname{Re} F(\tau))}\\
     &\leq \frac{n^d}{1-\tau/r} \left(\frac{\tau}{\pi(1-\tau^2) \sqrt{|\sinh \log \sqrt r|}}\right)^d 
     e^{-n \varepsilon}
        \end{align*}
        uniformly, where we used Lemma \ref{prop:valueomegaWZeFtau} (for $Z=Z'$) in the first line. We note that $r\leq \max(e^{-\varepsilon}, \frac{\tau+1}{2}) < 1$ uniformly, and the theorem follows.  
    \end{proof}

    \begin{proof}[Proof of Theorem \ref{thm:Cd_bulk_intro}.]
Let $\mathcal N$ be a small compact neighborhood of $B_\tau^d$. When $Z\not\in \mathcal N$, the result follows directly from Proposition \ref{thm:errorTermUniformd>1} and the translation invariance of the Ginibre kernel (and $\xi, \xi'< \xi_\tau$ for $n$ big enough).  Now let us consider $Z\in \mathcal N$. 
Using \eqref{eq:MehlerFordFactors}, we have
\begin{align} \label{eq:sum-sumjgeqn}
\mathbb K_n(\zeta,\zeta')
-\mathfrak{1}_{\xi_+<\xi_\tau} n^d \prod_{j=1}^d K_{(1-\tau^2)^{-1}}^{\infty}\left(\sqrt n \zeta_j, \sqrt n \zeta'_j\right)
= 
\frac{-1}{\pi^d (1-\tau^2)^\frac{d}{2}} \sum_{\lvert j\rvert \geq n} \frac{\left(\frac{\tau}{2}\right)^{\lvert j\rvert}}{j_1!\cdots j_d!} \prod_{k=1}^d \sqrt{\omega(\zeta_{k}) \omega(\zeta'_{k})} H_{j_k} \left({\frac{\zeta_{k}}{\sqrt{2 \tau}}} \right)\overline{ H_{j_k}\left({\frac{\zeta'_{k}}{\sqrt{2 \tau}}} \right)}
\end{align}
for any $\zeta, \zeta'\in \mathbb C^d$. By the inequality of arithmetic and geometric means, we have
\begin{multline*}
\left|\prod_{k=1}^d \sqrt{\omega\left(\sqrt n \zeta_k\right) \omega\left(\sqrt n \zeta'_k\right)} 
H_{j_k}\left(\sqrt{\frac{n}{2\tau}} \zeta_k\right) \overline{H_{j_k}\left(\sqrt{\frac{n}{2\tau}} \zeta'_k\right)}\right|\\
\leq \frac{1}{2} \prod_{k=1}^d \omega\left(\sqrt n \zeta_k\right)
\left|H_{j_k}\left(\sqrt{\frac{n}{2\tau}}\zeta_k\right)\right|^2
+\frac{1}{2} \prod_{k=1}^d \omega\left(\sqrt n \zeta'_k\right)
\left|H_{j_k}\left(\sqrt{\frac{n}{2\tau}}\zeta'_k\right)\right|^2.
\end{multline*}
Combining this with \eqref{eq:sum-sumjgeqn}, first for general $(\zeta, \zeta')$, and then for $(\zeta, \zeta)$ and $(\zeta', \zeta')$, we get
\begin{align*}
\left|\mathbb K_n(\zeta,\zeta') - \mathfrak{1}_{\xi_+<\xi_\tau} n^d \prod_{j=1}^d K_{(1-\tau^2)^{-1}}^{\infty}\left(\sqrt n \zeta_j, \sqrt n \zeta'_j\right)\right|
\leq \left|\mathbb K_n(\zeta,\zeta) - \frac{n^d}{\pi^d (1-\tau^2)^d}\right| +  \left|\mathbb K_n(\zeta',\zeta') - \frac{n^d}{\pi^d (1-\tau^2)^d}\right|.
\end{align*}
Putting $\zeta = Z+\frac{U}{\sqrt n}$ and $\zeta'= Z+\frac{V}{n}$, the result follows from Theorem \ref{thm:limitingDensity}.\\
\end{proof}

\subsection{Further results for strong non-Hermiticity}



We investigate the average density of points. As expressed in Theorem \ref{thm:limitingDensity}, the model has a bulk, given by
\begin{align*}
\mathcal E_\tau^d = \left\{Z\in \mathbb C^d : \frac{\lvert \operatorname{Re} Z\rvert^2}{(1+\tau)^2}
+ \frac{\lvert \operatorname{Im} Z\rvert^2}{(1-\tau)^2} < 1\right\}. 
\end{align*}
Consequently, the edge is identified as the $2d-1$ dimensional manifold $\partial \mathcal E_\tau^d$, that consists of all $Z\in \mathbb C^d$ such that
\begin{align*}
\frac{\lvert \operatorname{Re} Z\rvert^2}{(1+\tau)^2}
+ \frac{\lvert \operatorname{Im} Z\rvert^2}{(1-\tau)^2} = 1.
\end{align*}

We can be more precise about the rate of convergence, if we restrict to appropriate subsets. The next theorem is the equivalent of Theorem \ref{thm:onepointAsymp} for $d>1$. Here, we shall be less precise about the error terms, but we mention that our method can be used to get more explicit asymptotic formulas. We remind the reader that $
B_\tau^d = \{Z\in \mathbb R^d : \lvert Z\rvert \leq 2\sqrt \tau\}
$. 

\begin{theorem} \label{thm:onepointAsympd}
Fix $\tau\in (0,1)$ and $d>1$. Let $Z\in \mathbb C^d$ and assume that $Z\not\in \partial \mathcal E_\tau^d$. Then we have
\begin{align} \label{eq:onepointAsympd}
\mathbb K_n(Z, Z) = \frac{n^d \mathfrak{1}_{Z\in \mathcal E_\tau^d}}{\pi^d (1-\tau^2)^d} 
+ \mathcal O\left(n^{d+\frac{\gamma}{2}} e^{-2 n (\xi-\xi_\tau)^2 g(\xi+i\eta)}\right)
\end{align}
as $n\to\infty$, where
\begin{align*}
\gamma = 
\begin{cases}
-1, & Z\in \mathbb C^d\setminus B_\tau^d,\\
d-2, & Z\in B_\tau^d\setminus \partial B_\tau^d,\\
d, & Z\in \partial B_\tau^d.
\end{cases}
\end{align*} 
In each of these three regions, the $\mathcal O$ term is uniform on compact subsets.
\end{theorem}

\begin{proof}
When $Z=Z'$, the expressions for $z$ and $z'$ simplify, and we take 
\begin{align*}
z = \frac{\lvert \operatorname{Re} Z\rvert + i \lvert \operatorname{Im} Z\rvert}{\sqrt{2\tau}}
\qquad\text{ and }\qquad
z' = \frac{\lvert \operatorname{Re} Z\rvert - i \lvert \operatorname{Im} Z\rvert}{\sqrt{2\tau}}. 
\end{align*}
This means that $\xi'=\xi$ and $\eta'= -\eta$ in elliptic coordinates. The situation is then the same as in Theorem \ref{thm:onepointAsymp}. There are two minor differences. We have a different overall constant factor, and the integrand of $I_n$ has an extra factor $(1-s^2)^{\frac{d-1}{2}}$ in the denominator. The proofs for the cases $Z\in \mathbb C^d\setminus B_\tau^d$ and $Z\in B_\tau^d\setminus \partial B_\tau^d$ are virtually analoguous to the one of Theorem \ref{thm:onepointAsymp}, we omit the details. 

For the case $Z\in \partial B_\tau^d$ we make a distinction based on the parity of $d$. Let us first assume that $d=2k$ is even. Then, after $k$ integrations by parts, we have
\begin{align*}
I_n(d;1;z,z') &= -\frac{(-1)^k}{2\pi i} \oint_{\gamma_0} \frac{d^k}{ds^k} \left(\frac{e^{n F(s)}}{s-\tau} \frac{1}{(1+s)^k}\right) \log(1-s) ds\\
&= -\frac{(-1)^k}{2\pi i} \oint_{\gamma_0} e^{n F(s)} \left(n^k h_k(s) + n^{k-1} h_{k-1}(s) + \ldots + h_0(s) \right) \log(1-s) ds,
\end{align*}
where $h_0, \ldots, h_k$ are $n$-independent functions that are analytic for $s\neq -1,0,\tau$. Next, we deform $\gamma_0$ such that it goes through the saddle point $1$, but stays away from $s=-1$, and take the steepest descent directions $\pm \frac{2\pi}{3}$ in $s=1$. The dominant contribution of the residue at $s=\tau$ is given by
\begin{align*}
k! \operatorname{Res}_\tau\left(\frac{e^{n F(s)}}{(\tau-s)^{k+1}} \frac{1}{(1+s)^k}\right)
= n^k F'(\tau)^k \frac{e^{n F(\tau)}}{(1+\tau)^k}+ \mathcal O\left(n^{k-1}\right)
= \mathcal O\left(n^k\right)
\end{align*}
as $n\to\infty$. The saddle point contributions are smaller. Let us focus on $h_k$ for example. Explicitly, we have
\begin{align*}
h_k(s) = \frac{F'(s)^k}{s-\tau} \frac{1}{(1+s)^k}.
\end{align*}
On the steepest descent paths through $s=1$ we have $\log (1-s) = \log|1-s| \mp \frac{\pi}{3}$. Then the $\log |1-s|$ parts will effectively cancel with each other, in the sense that they will not give the dominant behavior, which will come from the argument of the logarithm. The sum of the saddle point contributions thus yields, up to a sign,
\begin{align*}
\frac{1}{2\pi i} \frac{2\pi i}{3} \left(\frac{|F'''(1)|}{4}\right)^k \frac{1}{1-\tau} \frac{(3!)^\frac{2k+1}{3}}{3} \frac{\Gamma\left(\frac{2k+1}{3}\right)}{|n F'''(1)|^\frac{2k+1}{3}} \left(1+ \mathcal O\left(\frac{\log n}{n^\frac{1}{3}}\right)\right)
= \mathcal O\left(n^{-\frac{2k+2}{3}}\right).
\end{align*}
Here we used that $F'(s) = \frac{1}{2} F'''(1) (s-1)^2 + \mathcal O((s-1)^3)$. Together with the factor $n^k$, we end up with $\mathcal O(n^{\frac{d-2}{6}})$ as $n\to\infty$. The reader may convince oneself that the saddle point contributions concerning $h_0, \ldots, h_{k-1}$ are of lower or equal order.  

When $d = 2k+1$ is odd, similarly, we have
\begin{align*}
I_n(d;1;z,z') &= -\frac{(-1)^k}{2\pi i} \oint_{\gamma_0} e^{n F(s)} \left(n^k h_k(s) + n^{k-1} h_{k-1}(s) + \ldots + h_0(s) \right) \frac{ds}{\sqrt{1-s}}. 
\end{align*}
From here, the arguments are analogous, and slightly easier, because we do not have a logarithm in this case. 
\end{proof}

\begin{remark}
One may point out that, to get the average density of points, it is natural to divide $\mathbb K_n(Z, Z)$ by the total number of points $N_n = \binom{n+d-1}{d}$ of the process, and this effectively destroys the sharp error terms in \eqref{eq:onepointAsympd}, which should be replaced by $\mathcal O(1/n)$. However, the choice to scale with $\sqrt n$ in \eqref{eq:cluster1_intro} was mainly for convenience. We may instead consider to scale as
$$d! N_n \mathcal K_n\left((d! N_n)^\frac{1}{2d} \, Z, (d! N_n)^\frac{1}{2d}\, Z' \right).$$
Using $(d! N_n)^\frac{1}{2d} = n + \frac{d-1}{2} + \mathcal O(1/n)$, some rewriting shows that the corresponding steepest descent analysis is essentially unchanged, except that the integrand in \eqref{eq:defIn} gets an extra factor 
$
\exp{\frac{d-1}{2} \left(\frac{s(z+z')^2}{2(1+s)}-\frac{s (z-z')^2}{2(1-s)}\right)}.
$
This different scaling does not change our results up to leading order, but may change the error terms. A version of Theorem \ref{thm:errorTermUniformd>1} still holds, and one may use it to argue that we get the sharp bounds in Theorem \ref{thm:onepointAsympd}, that is, as $n\to\infty$
$$
\mathcal K_n\left((d! N_n)^\frac{1}{2d} \, Z, (d! N_n)^\frac{1}{2d}\, Z' \right) 
= \frac{d! \mathfrak{1}_{Z\in \mathcal E_\tau^d}}{\pi^d (1-\tau^2)^d} 
+ \mathcal O\left(n^{\frac{\gamma}{2}} e^{-2 n (\xi-\xi_\tau)^2 g(\xi+i\eta)}\right).
$$
\end{remark}


We show that the product of Ginibre kernels also emerges as a bulk limit. As in the $d=1$ case, we have a more general asymptotic behavior for the correlation kernel. 


\begin{theorem}
Let $\tau\in (0,1)$ be fixed. Suppose that $Z,Z'\in \mathbb C^d\setminus B_\tau^d$ and $(\xi_+-\xi_\tau)^2+(\eta+\eta')^2>0$. Then we have
\begin{multline*} 
\mathbb K_n(Z, Z') = \mathfrak{1}_{\xi_+<\xi_\tau} n^d \prod_{j=1}^d K_{(1-\tau^2)^{-1}}^{\infty}\left(\sqrt n Z_j, \sqrt n Z'_j\right)\\
\pm  \frac{n^{d-\frac{1}{2}}}{\sqrt{32\pi\tau}} \frac{1}{(\pi\sqrt{1-\tau^2})^d} 
\frac{\displaystyle e^{- n (\xi-\xi_\tau)^2 g(\xi+i\eta)} e^{- n (\xi'-\xi_\tau)^2 g(\xi'+i\eta')} D_{\tau}^{n}(Z,Z') }{\sinh\left(\xi_+-\xi_\tau + i \frac{\eta+\eta'}{2}\right) \sqrt{\sinh\left(\xi+i\eta\right) \sinh\left(\xi'+i\eta'\right)} \left(1-e^{-\xi_+-2i\eta-2i\eta'}\right)^\frac{d-1}{2} } 
\left(1+ \mathcal O\left(\frac{1}{n}\right)\right)
\end{multline*} 
as $n\to\infty$, uniformly for $(Z,Z')$ in compact subsets, where
\begin{align*}
D_{\tau}^n(Z,Z')= \exp\left(i n (\eta+\eta' - \frac{e^{-2\xi}}{2} \sin(2\eta)-\frac{e^{-2\xi'}}{2} \sin(2\eta'))\right),
\end{align*}
and the $\pm$ sign depends explicitly on $(Z,Z')$. 
\end{theorem}

The proof is analogous to that of Theorem \ref{thm:largenKnZW_intro}, and is therefore omitted.

\subsection{Weak non-Hermiticity: Proofs of Theorem \ref{thm:weakNonHd_intro} and Theorem \ref{thm:EGEdbulkWeakNonH_intro}}
We move on to the weak non-Hermiticity limits, to begin with the limit in the bulk, Theorem \ref{thm:weakNonHd_intro}.\\

\begin{figure}
    \centering
    \begin{overpic}[width=0.5\textwidth]{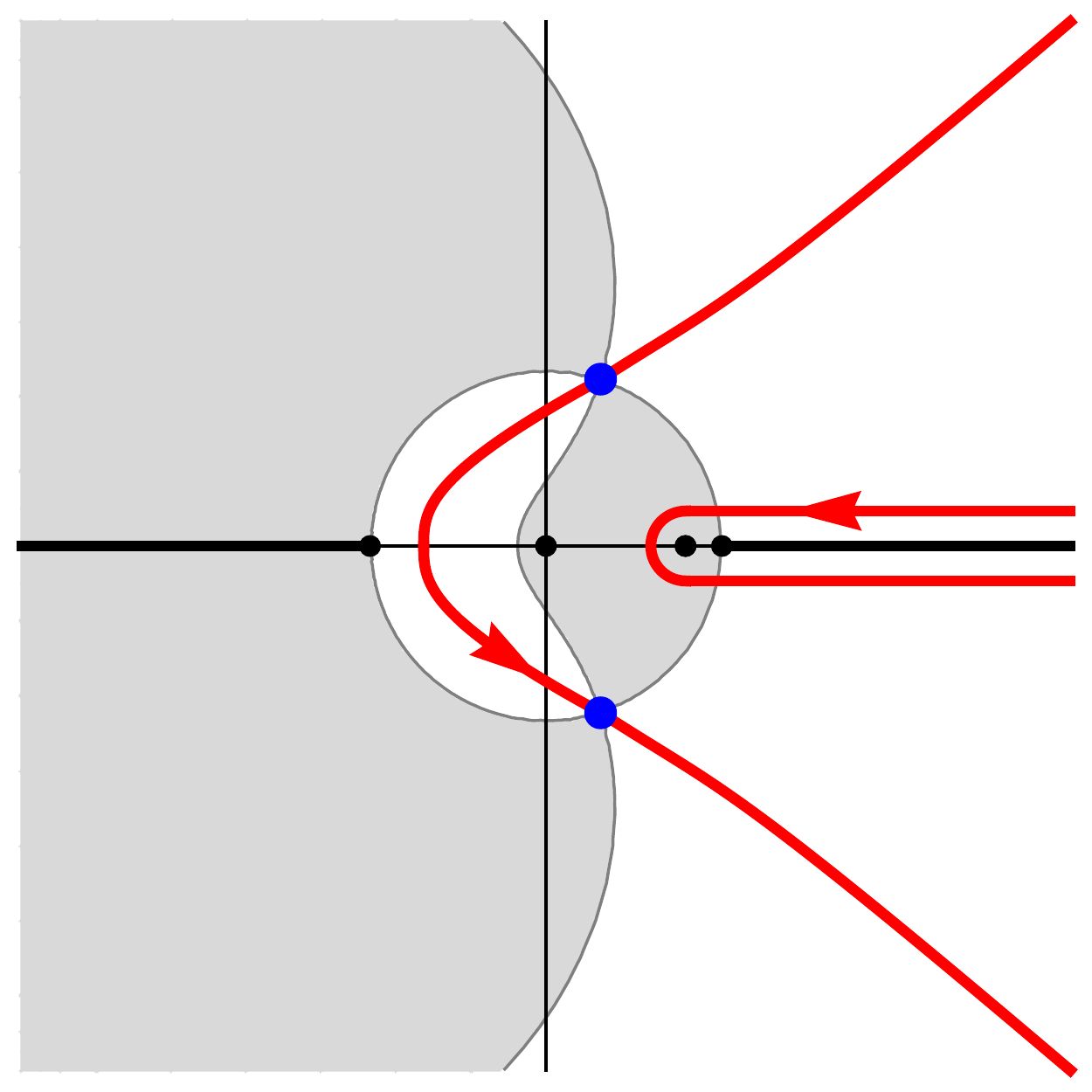}
        \put (51,68) {$a^{-1}$}
        \put (51,30) {$a$}
        \put (28,45) {$-1$}
        \put (66,43) {$1$}
        \put (62,43) {$\tau$} 
        \put (46,46) {$0$}
        \put (65,55) {$\gamma_2$}
        \put (70,80) {$\gamma_1$}
    \end{overpic}
    \caption{In the proof of Theorem \ref{thm:weakNonHd_intro} we change the choice of contours as in Figure \ref{Fig2a} as follows: we merge $\gamma_2$ and $\gamma_3$ into a single contour $\gamma_2$ that goes around $[\tau,\infty)$. The shaded areas represent the regions where $\operatorname{Re} F(s)\geq \operatorname{Re} F(a^{-1})$. \label{Fig2aa}} 
    \end{figure}

\noindent \textit{Proof of Theorem \ref{thm:weakNonHd_intro}.}
Our starting point is \eqref{eq:defEllipdIn2}. For convenience, we let $\tilde\kappa = \kappa/\tilde{\nu}(X)^2$ and $\tilde \nu(X) = \nu(X)^\frac{1}{d}$. We use the short-hand notation $\zeta \zeta'=\zeta_1\zeta'_1+\zeta_2 \zeta'_2 + \ldots + \zeta_d \zeta'_d$, for any $\zeta, \zeta'\in \mathbb C^d$, and write
\begin{align*}
n \frac{(z+z')^2}{2} &= n \frac{\left(2 X + \frac{U+\overline V}{(\nu(X) N)^\frac{1}{d}}\right)^2}{4\tau} 
= (n+\tilde\kappa) \lvert X\rvert^2 + \frac{X(U+\overline V)}{\tilde\nu(X)} + \mathcal O\left(\frac{1}{n}\right),\\
n \frac{(z-z')^2}{2} &= \frac{(U-\overline V)^2}{4 n \tilde\nu(X)^2}+ \mathcal O\left(\frac{1}{n^2}\right) ,
\end{align*}
as $n\to\infty$. We deform $\gamma_0$ to the path from Figure \ref{Fig2aa} (with $U=V=0$, effectively). As before, we can argue that the saddle point contributions are negligible. For the integral $\gamma_2$ bending around $[1,\infty)$, we use local coordinates $s = 1 - \frac{t}{n}$. Then we have
\begin{align*}
n F(s) = \frac{1}{2} (n+\tilde\kappa) \lvert X\rvert^2 + \frac{X(U+\overline V)}{2 \tilde\nu(X)} - \tilde\kappa
+(1- \frac{1}{4} \lvert X\rvert^2) t - \frac{(U-\overline V)^2}{4 \tilde\nu(X)^2} \frac{1}{t} + \mathcal O\left(\frac{1}{n}\right).
\end{align*}
By $\tilde\gamma$ we denote a curve that bends around $(-\infty, 0]$, with positive orientation. Then
\begin{align} \label{eq:defTnu3}
-\frac{1}{2\pi i} \int_{\gamma_3} \frac{e^{n F(s)}}{s-\tau} \frac{ds}{(1-s^2)^\frac{d}{2}}
= -\frac{1}{2\pi i} e^{\frac{1}{2} (n+\tilde\kappa) \lvert X\rvert^2 + \frac{X(U+\overline V)}{2 \tilde\nu(X)}-\tilde\kappa} \frac{n^\frac{d}{2}}{2^\frac{d}{2}} \int_{\tilde\gamma} e^{\pi^2 \tilde\nu(X)^2 t - \frac{(U-\overline V)^2}{4 \tilde\nu(X)^2} \frac{1}{t}} 
\frac{-dt}{t^\frac{d}{2} (t-\tilde\kappa)} \left(1+\mathcal O\left(\frac{1}{n}\right)\right).
\end{align}
Now, for any $\zeta\in\mathbb C$, let us, for arbitrary variable $\nu$, define the integral
\begin{align} \label{eq:defTnu}
T(\nu) = -\frac{1}{2\pi i} \int_{\tilde\gamma} e^{\pi^2 \nu^2 t - \frac{\zeta^2}{4 t}} \frac{dt}{t^\frac{d}{2} (t-\tilde\kappa)}. 
\end{align}
We notice that
\begin{align*}
\frac{d}{d\nu} \left(e^{-\tilde\kappa\pi^2 \nu^2} T(\nu)\right)
= -\frac{1}{2\pi i} 2\pi^2 \nu e^{-\tilde\kappa \pi^2 \nu^2} \int_{\tilde\gamma} e^{\pi^2 \nu^2 t - \frac{\zeta^2}{4 t}} \frac{dt}{t^\frac{d}{2}}
= -\frac{1}{2\pi i} 2 \pi^d \nu^{d-1} e^{-\tilde\kappa \pi^2 \nu^2} \int_{\tilde\gamma} e^{t - \frac{(\pi \nu \zeta)^2}{4 t}} \frac{dt}{t^\frac{d}{2}}.
\end{align*}
By \cite[10.9.19]{DLMF}, we have
\begin{align*}
\frac{1}{2\pi i} \int_{\tilde\gamma} e^{t - \frac{(\pi \nu \zeta)^2}{4 t}} \frac{dt}{t^\frac{d}{2}}
= 
\left(\frac{\pi\nu\zeta}{2}\right)^{1-\frac{d}{2}} J_{\frac{d}{2}-1}(\pi \nu \zeta). 
\end{align*}
We conclude that
\begin{align*}
e^{-\tilde\kappa \pi^2\nu^2} T(\nu) = T(0^+) + \frac12(2 \pi)^{1+\frac{d}{2}} \int_{0}^\nu e^{-\tilde\kappa \pi^2 t^2} t^{d-1} (\zeta t)^{1-\frac{d}{2}} J_{\frac{d}{2}-1}(\pi \zeta t) dt.
\end{align*}
Notice that $\widehat{J}_\nu(z^2)=(2/z)^\nu J_\nu(z)$ is an even function in $z$.
By the substitution $t\to \frac{1}{(\pi \nu)^2}t $ in \eqref{eq:defTnu}, we find that $T(0^+)=0$. 
Substituting $\nu = \tilde\nu(X)$ and $\zeta^2 = \frac{1}{\tilde\nu(X)^2} (U-\overline V)^2$, and applying a substitution $t\to \tilde\nu(X) t$ in the integration, we get
\begin{multline} \label{eq:weakNonHd2}
-\frac{1}{2\pi i} \int_{\gamma_3} \frac{e^{n F(s)}}{(1-s^2)^\frac{d}{2}} \frac{ds}{s-\tau}
=
\pi^{1+\frac{d}{2}} n^\frac{d}{2} e^{\frac{1}{2} (n+\tilde\kappa) \lvert X\rvert^2 + \frac{X(U+\overline V)}{2 \tilde\nu(X)}-\tilde\kappa} 
e^{\tilde\kappa\pi^2\tilde\nu(X)^2} \tilde\nu(X)^{d} 
(2/\pi)^{1-\frac{d}{2}}
\\
\times 
\int_0^{1}  e^{-\kappa \pi^2t^2} t^{d-1} \widehat{J}_{\frac{d}{2}-1}\left((U-\overline V)^2\pi^2 t^2\right) dt
\left(1 + \mathcal O\left(\frac{1}{n}\right)\right),
\end{multline}
as $n\to\infty$. Concerning the weights, we have for each $j=1,\ldots, d$ that
\begin{align*}
\omega_n\left(X_j+\frac{U_j}{\nu(X)^\frac{1}{d} n}\right) 
= e^{-\frac{1}{2} n |X_j|^2} e^{-\frac{\tilde\kappa}{4} |X_j|^2- X_j \frac{\operatorname{Re}(U_j)}{\tilde\nu(X)}}
e^{- \frac{\operatorname{Im}(U_j)^2}{\kappa}} \left(1 + \mathcal O\left(\frac{1}{n}\right)\right)
\end{align*}
as $n\to\infty$. Hence we have
\begin{multline} \label{eq:weakNonHdAA}
\frac{n^d}{\pi^d (1-\tau^2))^\frac{d}{2}}\prod_{j=1}^d \sqrt{\omega_n\left(X_j+\frac{U_j}{\nu(X)^\frac{1}{d} n}\right) \omega_n\left(X_j+\frac{V_j}{\nu(X)^\frac{1}{d} n}\right)}\\
= \tilde\nu(X)^d \frac{n^\frac{3 d}{2}}{\pi^\frac{d}{2} (2\pi \kappa)^\frac{d}{2}} e^{-\frac{1}{2} n \lvert X\rvert^2} 
e^{-\frac{\tilde\kappa}{4} \lvert X\rvert^2} 
e^{- \frac{X \operatorname{Re} (U+V)}{2\tilde\nu(X)}}
e^{- \frac{\lvert \operatorname{Im}(U)\rvert^2+\lvert \operatorname{Im}(V)\rvert^2}{2 \kappa}}
\left(1 + \mathcal O\left(\frac{1}{n}\right)\right).
\end{multline}
Combining \eqref{eq:weakNonHdAA} with \eqref{eq:weakNonHd2}, and inserting these in \eqref{eq:defEllipdIn2}, we obtain the result.\\
\qed

Finally, we prove the edge limit in the regime of weak non-Hermiticity, Theorem \ref{thm:EGEdbulkWeakNonH_intro}.\\
\begin{figure}
    \centering
    \begin{overpic}[width=0.5\textwidth]{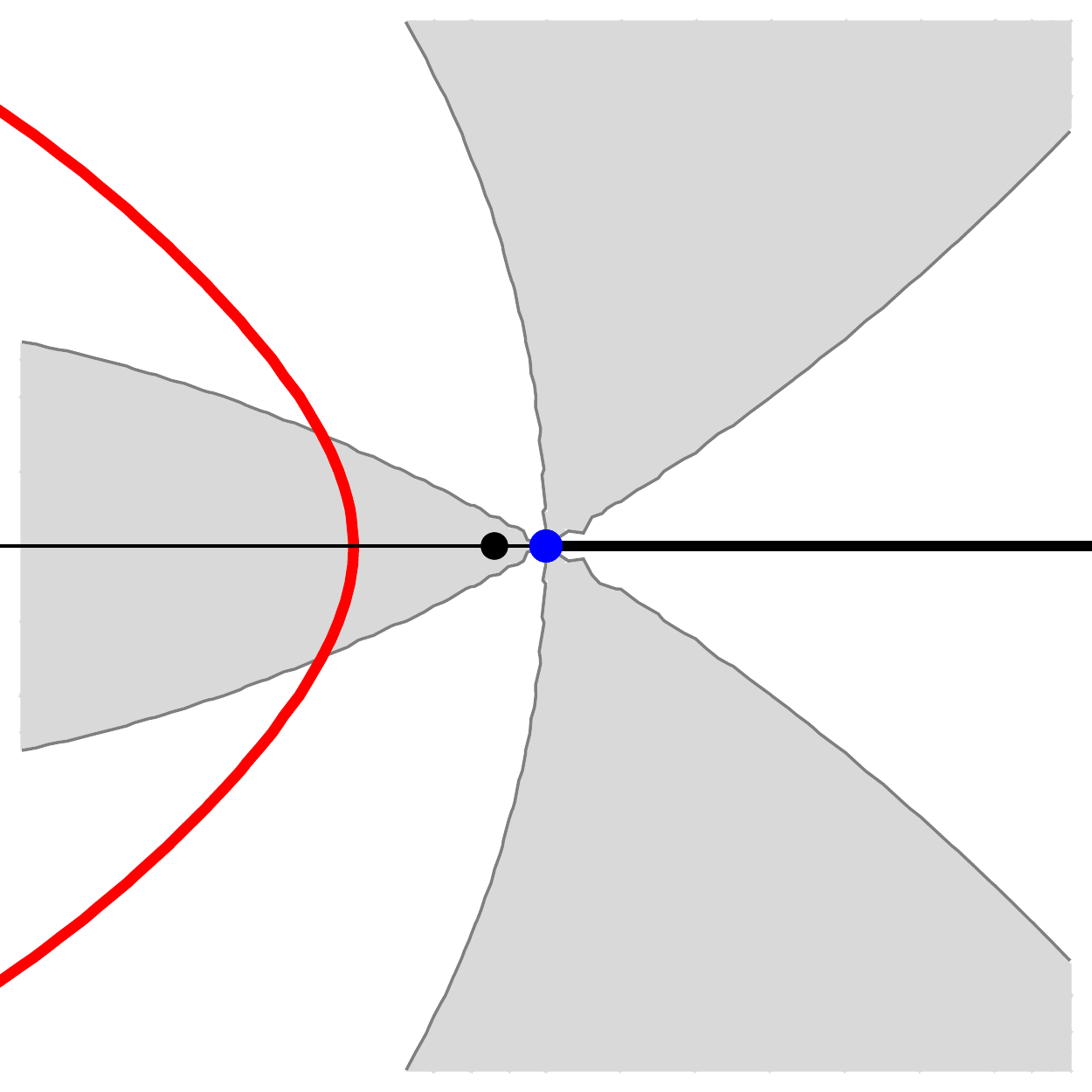}    
        \put (50,45) {$1$}
        \put (45,45) {$\tau$} 
    \end{overpic}
    \caption{In the proof of Theorem \ref{thm:EGEdbulkWeakNonH_intro}  we change the choice of contours as in Figure \ref{Fig2aaa}, but we make sure that it lies at the left of $\tau=1-\kappa/n^{1/3}$. In the picture we have zoomed in around $s=1$ to illustrate this. The shaded areas represent the regions where $\operatorname{Re} F(s)\geq \operatorname{Re} F(1)$. \label{Fig2aab}} 
    \end{figure}

\begin{proof}[Proof of Theorem \ref{thm:EGEdbulkWeakNonH_intro}.]
The proof is similar to the 
proof of Theorem \ref{thm:fermionsEdgeDensity_intro}.
First, we notice that 
\begin{align*}
\frac{n}{2} (z+z')^2 &= n\frac{\lvert X\rvert^2}{\tau} + n^\frac{1}{3} \frac{X (U+\overline V)}{\tau} + \mathcal O(n^{-\frac{1}{3}})
= 4 n + \kappa^2 n^\frac{1}{3} + \kappa^3 + (2 n^\frac{1}{3} + \kappa) \Omega (U+\overline V) + \mathcal O(n^{-\frac{1}{3}}),\\
\frac{n}{2} (z-z')^2 &= n^{-\frac{1}{3}} \frac{(U-\overline V)^2}{4\tau} = n^{-\frac{1}{3}} \frac{(U-\overline V)^2}{4} + \mathcal O(n^{-\frac{2}{3}}). 
\end{align*}
As in the proof of Theorem \ref{thm:fermionsEdgeDensity_intro}, we deform $\gamma_0$ to the contour from Figure \ref{Fig2aaa}. However, now $ \tau\neq 1$ (if $\kappa \neq 0$), and we have to make sure that $\gamma_0$ is at the left of $\tau$ but intersects the real line at a point at distance of order $n^{-1/3}$ to 1. See also Figure \ref{Fig2aab}.

Again, we use the local coordinates $s=1-t/n^{-1/3}$. Some bookkeeping  yields
\begin{align*}
\frac{n}{2} (z+z')^2 \frac{s}{s+1} &= 2n + \frac{1}{2}\kappa^2 n^\frac{1}{3} + \frac{1}{2}\kappa^3 + (n^\frac{1}{3} + \frac{1}{2}\kappa) \Omega (U+\overline V)\\
&\quad - (n^\frac{2}{3} + \frac{1}{4} \kappa^2) t - \frac{1}{2} \Omega (U+\overline V) t
 - \frac{1}{2} n^\frac{1}{3} t^2 - \frac{1}{4} t^3
 + \mathcal O(n^{-\frac{1}{3}})
\end{align*}
and
\begin{align*}
\frac{n}{2} (z-z')^2 \frac{s}{s-1} &= -\frac{(U-\overline V)^2}{4 t} + \mathcal O(n^{-\frac{1}{3}}).
\end{align*}
Hence, using $-n \log s = n^{\frac{2}{3}} t + \frac{1}{2} n^{\frac{1}{3}} t^2 + \frac{1}{3} t^3 +  \mathcal O(n^{-\frac{1}{3}})$, we have
\begin{align*}
nF(s) &=  2n - \kappa n^\frac{2}{3} + \frac{1}{6}\kappa^3 + (n^\frac{1}{3} + \frac{1}{2}\kappa) \Omega (U+\overline V)\\
&\quad - \frac{1}{4} (\kappa^2 + 2 \Omega (U+\overline V)) t + \frac{1}{12} t^3 -\frac{(U-\overline V)^2}{4 t} + \mathcal O(n^{-\frac{1}{3}}).
\end{align*}
Hence
\begin{multline} \label{eq:EGEdBulkWeakNonH2a}
-\frac{1}{2\pi i} \int_{\gamma_3} \frac{e^{n F(s)}}{(1-s^2)^\frac{d}{2}} \frac{ds}{s-\tau}\\
= -\frac{1}{2\pi i} e^{2n - \kappa n^\frac{2}{3} + \frac{1}{6}\kappa^3 + (n^\frac{1}{3} + \frac{1}{2}\kappa) \Omega (U+\overline V)} \frac{n^\frac{d}{6}}{2^\frac{d}{2}} \int_{\tilde\gamma} e^{- \frac{1}{4} (\kappa^2 + 2 \Omega (U+\overline V)) t + \frac{1}{12} t^3 -\frac{(U-\overline V)^2}{4 t}} \frac{dt}{t^\frac{d}{2} (t-\kappa)} \left(1+\mathcal O\left(n^{-\frac{1}{3}}\right)\right),
\end{multline}
where $\tilde\gamma$ is a curve in the right half place that starts at $e^{-\pi i/3} \infty$ and ends at $e^{\pi i/3} \infty$.
On the other hand, we have for each $j=1,\ldots, d$ that
\begin{align*}
\omega_n\left(X_j+\frac{U_j}{n^\frac{2}{3}}\right)
= e^{-(2n - \kappa n^\frac{2}{3}) \Omega_j^2} 
e^{-n^\frac{1}{3} 2\Omega_j \operatorname{Re}U_j
}
e^{- \frac{\lvert \operatorname{Im} U_j\rvert^2}{\kappa}}
\left(1 + \mathcal O\left(n^{-\frac{1}{3}}\right)\right).
\end{align*}
Thus we have
\begin{multline} \label{eq:weakNonHd}
\frac{n^d}{\pi^d (1-\tau^2)^\frac{d}{2}}\prod_{j=1}^d \sqrt{\omega_n\left(X_j+\frac{U_j}{n^\frac{2}{3}}\right)\omega_n\left(X_j+\frac{V_j}{n^\frac{2}{3}}\right)}\\
= \frac{n^\frac{7 d}{6}}{\pi^\frac{d}{2} (2\pi \kappa)^\frac{d}{2}} e^{-2n + \kappa n^\frac{2}{3}} 
e^{- n^\frac{1}{3} \Omega \operatorname{Re}(U+V)}
e^{- \frac{\lvert \operatorname{Im}(U)\rvert^2+\lvert \operatorname{Im}(V)\rvert^2}{2 \kappa }}
\left(1 + \mathcal O\left(n^{-\frac{1}{3}}\right)\right).
\end{multline}
Inserting \eqref{eq:EGEdBulkWeakNonH2a} and \eqref{eq:weakNonHd} in \eqref{eq:defEllipdIn2}, we obtain
\begin{multline*}
\frac{1}{n^\frac{4 d}{3}} \mathbb K_n\left(X+\frac{U}{n^\frac{2}{3}}, X+\frac{V}{n^\frac{2}{3}}\right)
= 
\frac{1}{(2\pi)^\frac{d}{2}}e^{\frac{1}{6}\kappa^3+\frac{1}{2}\kappa \Omega (U+\overline V)}
e^{i n^\frac{1}{3} \Omega \operatorname{Im}(U-V)}
\frac{e^{-\frac{\lvert \operatorname{Im} U\rvert^2+\lvert \operatorname{Im} V\rvert^2}{2\kappa}}}{(2\pi \kappa)^\frac{d}{2}}\\
\times \frac{1}{2\pi i} \int_{\tilde\gamma} e^{- \frac{1}{4} (\kappa^2 + 2 \Omega (U+\overline V)) t + \frac{1}{12} t^3 -\frac{(U-\overline V)^2}{4 t}} \frac{dt}{t^\frac{d}{2} (t-\kappa)} +\mathcal O\left(n^{-\frac{1}{3}}\right).
\end{multline*}
In order to move the factor $t-\kappa$ to the exponent we use \eqref{eq:1/t}.
Furthermore, we apply again \eqref{eq:propagator} to the term with $1/t$ in the exponent, and after changing variables $t\to 2^{2/3}t$, we obtain
\begin{multline*}
\frac{1}{n^\frac{4 d}{3}} \mathbb K_n\left(X+\frac{U}{n^\frac{2}{3}}, X+\frac{V}{n^\frac{2}{3}}\right)
= 
\frac{2^{2/3}}{(\pi\kappa)^\frac{d}{2}}e^{\frac{1}{6}\kappa^3+\frac{1}{2}\kappa \Omega (U+\overline V)}
e^{i n^\frac{1}{3} \Omega \operatorname{Im}(U-V)}
e^{-\frac{\lvert \operatorname{Im} U\rvert^2+\lvert \operatorname{Im} V\rvert^2}{2\kappa}}\\
\times \frac{1}{2\pi i} \int_{\tilde\gamma} 
\int_0^\infty\int_{\mathbb{R}^d}
e^{-iQ(U-\overline{V})+ \kappa s} e^{-t\Xi+\frac13 t^3}dtds \frac{d^dQ}{(2\pi)^d}+\mathcal O\left(n^{-\frac{1}{3}}\right),
\end{multline*}
with $\Xi=2^{2/3}(Q^2+\frac14(\kappa^2+2\Omega(U+\overline{V})+s)$. Using again the integral representation of the Airy function \eqref{eq:Airy-rep} we arrive at the desired result. 
\end{proof}

\section*{Acknowledgments}

It is a great pleasure and honour to submit this work to the special issue dedicated to the memory of Freeman Dyson. We are indebted to him for his work on the mathematical foundations of Random Matrix Theory that has found so many beautiful applications, apart from his further ground breaking contributions in theoretical physics, notably quantum field theory.

GA and LM are partly funded by the Deutsche Forschungsgemeinschaft (DFG, German Research Foundation) – SFB 1283/2 2021 – 317210226 "Taming uncertainty and profiting from randomness and low regularity in analysis, stochastics and their applications". The Knut and Alice Wallenberg Foundation are thanked for funding (GA) and the Department of Mathematics at KTH Stockholm for hospitality (GA, LM), where this work was initiated. LM was also supported by a PhD fellowship of the Flemish Science Foundation (FWO). 
MD was supported  by the Swedish Research Council (VR), grant no 2016-05450 and grant no. 2021-06015, and the European Research Council (ERC),  Grant Agreement No. 101002013.

\end{document}